\documentclass[11pt,a4paper]{article}  
\usepackage{amsfonts, amsmath,  balance, booktabs, caption, comment,  cmap,  environ, etoolbox, fancyhdr, float, fontaxes, geometry, graphics, hyperref, iftex, inconsolata, libertine, manyfoot, microtype, mmap, mweights,  nccfoots, refcount, setspace, textcase, totpages, trimspaces, upquote, url, xcolor, xkeyval, xstring}

\AtBeginDocument{%
  \providecommand\BibTeX{{%
    \normalfont B\kern-0.5em{\scshape i\kern-0.25em b}\kern-0.8em\TeX}}}

\usepackage[utf8]{inputenc}
\usepackage{graphicx}
\usepackage{amsmath,amsthm,amssymb}
\usepackage[english]{babel}
\usepackage{float}
\usepackage{subcaption}
\usepackage{fullpage}
 \usepackage{hyperref}
\usepackage{xcolor}
 \hypersetup{
 	colorlinks,
 	linkcolor={red!90!black},
 	citecolor={green!80!black},
 	urlcolor={blue!80!black}
 }

\newtheorem{theorem}{Theorem}

\usepackage{todonotes}
\usepackage{pgfplots}
\pgfplotsset{width=4.5cm,compat=1.9}
\usepackage{pgfplotstable}
\usepgfplotslibrary{groupplots}
\usepackage{algorithm}
\usepackage{algorithmic}
\usepackage{verbatim}
\usepackage{multirow}

\newcommand{\Scal}[0]{\ensuremath{{\mathcal S}}}
\newcommand{\Sstarcal}[0]{\ensuremath{{\mathcal S_2}}}

\newcommand{\NN}[0]{\ensuremath{{\mathit N\!N}}}
\newcommand{\kNN}[0]{\ensuremath{{\mathit kN\!N}}}
\newcommand{\kay}[0]{\ensuremath{{\mathit k}}}

\newcommand{\Prm}[0]{\ensuremath{{\mathit P}}}

\newcommand{\Qrm}[0]{\ensuremath{{\mathit Q}}}

\newcommand{\FreDist}[0]{\ensuremath{{\delta_{F}(\!\Prm\!,\!\Qrm)}}}

\newcommand{\True}[0]{\ensuremath{{\mathit true}}}
\newcommand{\False}[0]{\ensuremath{{\mathit f\!alse}}}

\newcommand{\Elmt}[0]{\ensuremath{{\mathit \tau}}}
\newcommand{\Eadd}[0]{\ensuremath{{\mathit \varepsilon^+}}} 
\newcommand{\Erel}[0]{\ensuremath{{\mathit \varepsilon^*}}}  

\renewcommand{\O}{\ensuremath{\mathcal{O}}}

\newcommand{\CandCurv}[0]{\ensuremath{{\Scal_1}}}

\newcommand{\Rd}[0]{\ensuremath{{\mathbb{R}^{d}}}}

\newcommand{\ereachq}[0]{\ensuremath{{\varepsilon \cdot reach(\Qrm)}}}
\newcommand{\reach}[0]{\ensuremath{{reach}}}

\newcommand{\RNN}[0]{\ensuremath{{\mathit RN\!N}}}
\newcommand{\BB}[0]{\textsc{BB}}
\newcommand{\LB}[0]{\textsc{LB}}
\newcommand{\st}[0]{\textsc{st}}

\newcommand{\tr}[0]{\textsc{tr}}
\newcommand{\fd}[0]{\textsc{fd}}
\newcommand{\UB}[0]{\textsc{UB}}
\newcommand{\adf}[0]{\textsc{adf}}
\newcommand{\adff}[0]{\textsc{adf1}}
\newcommand{\adfr}[0]{\textsc{adf2}}
\newcommand{\adfd}[0]{\textsc{adf3}}

\newcommand{\CCT}[0]{\textsc{CCT}}

\newcommand{\R}{\mathbb{R}}

\title{A Practical Index Structure Supporting Fréchet Proximity Queries Among Trajectories}

\author{Joachim Gudmundsson\thanks{joachim.gudmundsson@sydney.edu.au} \and Michael Horton\thanks{michael.horton@sportlogiq.com} \and John Pfeifer\thanks{johnapfeifer@yahoo.com} \and Martin P. Seybold\thanks{martin.seybold@sydney.edu.au}}       

\date{}

\begin{document}
	\maketitle
	\begin{abstract}
		We present a scalable approach for range and $k$ nearest neighbor queries under computationally expensive metrics, like the continuous Fréchet distance on trajectory data.
		Based on clustering for metric indexes, we obtain a dynamic tree structure whose size is linear in the number of trajectories, regardless of the trajectory's individual sizes or the spatial dimension, which allows one to exploit low `intrinsic dimensionality' of data sets for effective search space pruning.
		
		Since the distance computation is expensive, generic metric indexing methods are rendered impractical.
		We present strategies that
		(i)   improve on known upper and lower bound computations,
		(ii)  build cluster trees without any or very few distance calls, and
		(iii) search using bounds for metric pruning, interval orderings for reduction, and randomized pivoting for reporting the final results.
		
		We analyze the efficiency and effectiveness of our methods with extensive experiments on diverse synthetic and real-world data sets.
		The results show improvement over state-of-the-art methods for exact queries, and even further speed-ups are achieved for queries that may return approximate results.
		Surprisingly, the majority of exact nearest-neighbor queries on real data sets are answered \emph{without any} distance computations.
	\end{abstract}

\newcommand{\mathsc}[1]{{\normalfont\textsc{#1}}}



\section{Introduction} \label{sec:introduction}
The rapid growth of movement data diversity and acquisition over the past decade poses expanding scalability \emph{and} flexibility demands on information systems.
Tracking technologies such as video analysis, RFIDs, and GPS have enabled experts to collect trajectory data on objects as diverse as
flying animals~\cite{pigeon16,bats15,masked17,gulls15}, shipping vessels~\cite{vessel05}, basketballs~\cite{NBA16}, humans~\cite{geo2012},~vehicles~\cite{truckbus05,taxiA11,taxiB10}, hurricanes~\cite{hurdat217}, athletes~\cite{soccer15}, terrestrial animals~\cite{kruger09,cats16}, and tablet pen-tip writing~\cite{pentip06}. 
The size of trajectory data sets continues to increase as improved tracking technology records higher frequencies and larger numbers of objects.
Real-world data sets~\cite{taxiA11,taxiB10,masked17,NBA16,geo2012} consist of tens of thousands trajectories with thousand or more vertices per trajectory and keep growing.
Moreover, tracking complex objects whose position consists of several spatial coordinates (e.g. a Bison cow and its calf), challenges researchers to provide \emph{computational} solutions for trajectory data in high dimensions.

A research problem that has recently received considerable attention~\cite{ast17,bal17,ber13,ber17,dre17,gud15}, is the search for efficient data structures and algorithms that enable nearest-neighbor and range queries on large trajectory data sets.
Proximity searches are a core engine underlying visualization and classification applications that provide domain-specific researchers with better insight regarding their trajectory data.
Example applications are diverse, such as:
identifying potential changes in the migration paths of birds~\cite{masked17,gulls15},
locating similar European Football player ball possession trajectories when driving towards the opponent's net~\cite{soccer15}, 
determining if shipping vessels stay within range of a shipping path~\cite{vessel05}, and 
discovering how many people have a similar commute along a specified route~\cite{geo2012}.

A challenging task in trajectory data analysis is choosing an appropriate trajectory similarity measure. 
Common measures include the discrete or continuous Fr\'echet~\cite{alt95,bri14,bri16,buc217} and   Hausdorff~\cite{alt2009} distances, which fulfill the triangle inequality, and the non-metric Dynamic Time Warping (DTW)~\cite{keogh2005} and Longest Common Subsequence (LCSS)~\cite{vlachos2002} similarity measures.
We focus on the continuous Fr\'echet distance for high dimensional trajectory data for a variety of reasons.
First, it jointly captures the similarity in the position, shape, and direction between two trajectories.
The Hausdorff distance does not capture similarity of directions, which is a requirement for many real-world applications such as human body movement classification.
Second, it is less affected by irregularly sampled trajectories and thus suited for simplified trajectories. The latter is particularly useful in practice as real-world data sets are typically simplified in a pre-processing step using standard trajectory simplification algorithms~\cite{bl-fstts-17,dp-arnpr-73,lwj-tsmdb-14,zhang-18}.
Third, it is a metric (unlike DTW or LCSS) and hence it can take advantage of metric indexing~\cite{hetland-09} techniques.

Proximity search problems present difficulties in several regards, which renders asymptotic worst-case analysis often meaningless for concrete instances~\cite{moret2002}.
In such cases, empirical evidence is especially pertinent to compare solution strategies~\cite{hetland-09}.
For example, real-world trajectory data sets may not contain attributes that lead to worst-case runtimes, but instead behave more 'reasonably' and perform much better in practice.
Though we state asymptotic worst case bounds for our algorithms, the evaluation of our proposed solution strategies focuses heavily on a set of robust experiments using a large variety of data sets.

\subsection{Related Work} \label{ssec:related_work}
Search problems bound to find $k$ nearest neighbors ($\kNN$) and neighbors within a spherical range ($\RNN$) in vector spaces under a norm have a long and rich history. 
The well known $d$D-Tree~\cite{bentley75} (a.k.a KD-Tree) successively partitions the input point set $\Scal \subseteq \Rd$ with alternating axis-orthogonal hyperplanes to obtain a balanced binary tree in the confines of $\O(|\Scal|)$ space.
However, axis-orthogonal range search, using only linear space, requires $\Theta({|\Scal|}^{1-1/d})$ time in the worst-case.
The Range-Tree~\cite{bentley79-range-tree} improves this worst-case time with the expense of storage that is exponential in $d$. 
This frequent, underlying phenomenon is well known as the `curse of dimensionality' and Weber et al.~\cite{weber98} show that the \emph{naive scan} outperforms partitioning and clustering techniques for proximity search on average if $d$ exceeds $10$.
Theoretical and experimental works on general proximity search problems mainly assume that the distance of two elements can be determined in negligible time, e.g. in $\O(d)$ or $\O(1)$.
Exact proximity searches on trajectories in $\Rd$ under the continuous Fréchet distance $\delta_{F}$ however are a computationally harder problem than proximity search on mere points of $\Rd$ under Euclidean distances.

%
%
Alt and Godau~\cite{alt95} provide an $\O(n^2)$ time algorithm for deciding if the Fréchet distance is at most some given value.
Combining this algorithm with Cole's Parametric Search~\cite{Cole1987} gives an $\O(n^2 \log n)$ time algorithm that determines $\delta_{F}$.
The decision procedure $\delta_{F\!D}$
does not allow strongly sub-quadratic algorithms, unless a common complexity theory conjecture (SETH) fails \cite{bri14}.
Recently, Buchin et al.~\cite{buc217} gave a randomized algorithm that computes $\delta_{F}$ in $\mathcal{O}(n^2 (\log\log n)^2)$ time on a word RAM. 

Clearly, for exact $\RNN$ trajectory queries only $\delta_{F\!D}$ computations suffice, whereas exact $\kNN$ queries might well require exact $\delta_{F}$ computations.
%
The 2017 SIGSPATIAL Cup~\cite{sigcup17} asked for practical data structures to answer $\RNN$ queries under $\delta_{F}$ on trajectories in $d=2$ dimensional space.
Top ranked competitors \cite{bal17,buch17,dut17} apply filter-\&-refine strategies that often use spatial hashing~\cite{buch17,dut17} or a quad tree~\cite{bal17} over the trajectory's start point, end point, and bounding box points to determine a potentially smaller list of candidates.
Recently Bringman et al. \cite{bri19} improved further upon their winning submission with an orthogonal-range search in a $(4d)$D-Tree (i.e. an $8$ dimensional {KD-Tree}) to obtain a candidate result list, which is then refined by heuristic distance computations and an even further tuned decision procedure, to achieve practically fast range queries on three real-world data sets in the plane ($d=2$).

There is also work on data structures for approximate proximity queries under $\delta_{F}$.
In~\cite{ber17} de Berg et al. present an approximate query structure for $\kNN$ and $\RNN$ queries.
The structure uses $\O(|\Scal|/\varepsilon^{2\eta})$ space, where $\varepsilon>0$ is a quality parameter and $\eta$ the \emph{fixed number} of vertices that every query trajectory $Q$ is restricted to have.
The query algorithm returns $S \subseteq \Scal$ with an additive error of at most $\ereachq$ in $\O(1+|S|)$ time, where $\reach(Q)$ denotes the maximum distance from the start vertex of $Q$ to any of its other vertices.
Though the structure is dynamic, the vertex number of a query trajectory $\eta$ must be fixed \emph{prior} to construction and space usage is \emph{exponential} with respect to it.
Driemel and Silvestri~\cite{dre17} provide asymptotic analysis on a set of data structures and query algorithms for approximate $\NN$ searches under the Discrete Fr\'echet distance, and even for the Dynamic Time Warping similarity measure.
They utilize an asymmetric version of Locality Sensitive Hashing which maps similar trajectories to the same hash table buckets.
However the space and queries bounds are exponential in $n$, i.e. the number of points per trajectory, already for constant factor approximations.

Recently, Xie et al.~\cite{xie-17} provided a data structure for performing distributed $\kNN$ queries on trajectories using either a `Discrete Segment Hausdorff Distance` or a `Discrete Segment Fr\'echet Distance'.
The data structure is constructed by uniformly randomly sampling a set of trajectory segments, which are then used to compute a set of spatial partition boundaries.
Within each spatial partition a variation of an R-Tree~\cite{guttman-84} data structure is constructed by computing the centroid of the bounding box of trajectory segments.
Their experiments for exact $10$-$\NN$ queries under the Discrete Segment Fr\'echet Distance on a synthetic trajectory data set ($|\Scal|=3$M) shows an average run-time of $4.5$ seconds, performing $6,000$ distance calls, on a cluster of $16$ compute nodes with $152$ parallel threads and $512$GB total RAM.

%
%
There are numerous approaches that seek to extend simple binary serach trees to the proximity search problem for general sets $\Scal$ under a metric (see Table 9.1 in \cite{hetland-09} for a basic overview).
Classic \emph{metric tree indexes} partition the input along generalized metric balls or bisector planes, which offer structures using only $\O(|\Scal|)$ space.
Proximity searches attempt to prune sub-trees by means of the query element's distance to a sub-tree representative and the triangle inequality.
For example, the static and binary VP-Tree~\cite{Yianilos93-vp-tree} is balanced due to recursively choosing a ball radius, around the picked vantage point, which coincides with the median distance.
In contrast, the dynamic and binary BS-Tree~\cite{kalantari-83} recursively partitions elements into the closer of two ball pivots, resulting in a potentially unbalanced tree.
The well known M-Tree~\cite{ciaccia1997}, which is essentially a multi-way BS-Tree, offers strategies to tune I/O disk accesses. 
None of the above methods provide worst-case guarantees for proximity searches since ball overlap depends on on the underlying input set $\Scal$.
In fact,
all pairwise distances can have roughly the same value, which enforces a worst-case query performance of $\Theta(|\Scal|)$ for all such structures.

More recent
approaches build upon clustering ideas to obtain a small set of `compact' metric balls with little `overlap' that cover all elements.
More formally, for a resolution $\varepsilon$, an $\varepsilon$-net of a finite metric is a set of centers of distance at least $\varepsilon$ whose $\varepsilon$-balls cover all elements
-- e.g. Quadtree cell centers of a certain level.
Since packing and covering problems strongly depend on the dimension of Euclidean spaces, authors seek to capture the `intrinsic dimensionality' of metric spaces for algorithm analysis with measures thereof.
Gonzalez' farthest-first clustering~\cite{gonzalez85} provides $\varepsilon$-nets of size no bigger than an optimal $\tfrac\varepsilon2$-net, however straight-forward implementations perform $\O(|\Scal|^2)$ distance calls.
Navigating-Nets~\cite{DBLP:conf/soda/KrauthgamerL04} connect layers of nets, having shrinking resolutions, with additional links for a data-structure, in which the worst-case $\NN$ search time can be bounded in terms of the spread and doubling-constant of the finite metric.
However, the factor for $|\Scal|$ in the space bound depends on non-trivial terms over the doubling-constant.
The expansion constant $\gamma$ of \cite{kargerR02} is another data set parameter, which is weaker than the doubling constant (c.f. Section~\ref{ssec:DimensionMeasure}).
The Cover-Tree~\cite{BeygelzimerKL06} offers a simpler, yet dynamic, approach within the confines of $\O(|\Scal|)$ space, irrespective of $d$ and `intrinsic dimensionality' measures of the metric.
The authors maintain $\varepsilon$-net properties of tree levels during insert and delete operations, which provides hierarchical cluster trees of arity $\gamma^4$ and depth $\O(\gamma^2 \log |\Scal|)$ whose \emph{form} depend on the expansion-constant $\gamma$.
Moreover, their $\NN$ search tree traversal takes no more than $\O(\gamma^{12}\log|\Scal|)$ operations.
On the other hand, the experiments by Kibriya and Frank~\cite{KibriyaF07-KD-is-best}, on the performance of exact $\NN$ search over low dimensional real-world data under Euclidean distances, report a query performance ordering of KD-Trees over Cover-Trees over VP-Trees.
The naive scan sporadically outperforms each even on low dimensional real-world data and performances of either method converge on synthetic data with $d \geq 16$, as the curse suggests.

\begin{figure}\centering \vspace{-.3cm}  
    \includegraphics[width=\columnwidth, height=.5\columnwidth]{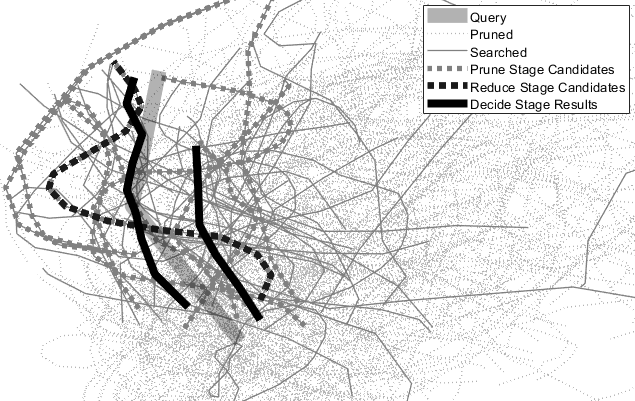} \\
	\includegraphics[width=\columnwidth, height=.2\columnwidth]{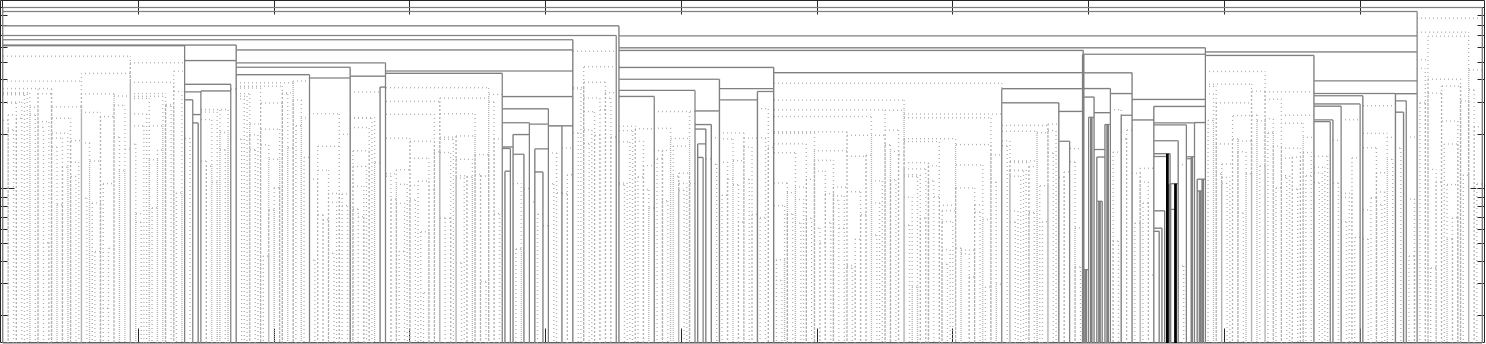}
\vspace{-.7cm}	
\caption{
An example of a $2$-$\NN$ query on $545$ bat trajectories~\cite{bats15}. The top plots $2$D trajectories: query, pruned, searched, prune stage candidates, reduce stage candidates, and decide stage results. The bottom shows the corresponding CCT dendrogram for nodes that were pruned (dotted line) or searched (solid line) (c.f. Section~\ref{sec:queries}).}
	\label{fig:teaser} 
 	\vspace{-.3cm}
\end{figure}

Many real-world trajectory data sets $\Scal$ consist of ten thousand or more elements and the number of vertices $n$ per trajectory is often in the thousands. 
Since the performance penalty for a single Fréchet proximity decision $\delta_{F\!D}$ or distance computation $\delta_{F}$ is \emph{huge} (e.g. $n^2 \approx |\Scal|$ or $n^2 \gg n\log|\Scal| $), our main objective is to minimize the absolute number of these expensive computations at query time.
This is in the same spirit as analysis in the I/O-model~\cite{av-iocsr-98} of computation, which measures the cost of answering a query as the number of expensive I/O operations performed by the query algorithm.
In our setting, the cost is primarily measured in the number of continuous Fr\'echet distance computation calls performed by the query algorithms.

\subsection{Contribution and Paper Outline} \label{ssec:contribution}
We present a scalable and extendable framework for approximate and exact $\kNN$ and $\RNN$ proximity queries under computationally expensive metric distance functions that is suitable for practical use in information systems -- e.g. proximity queries under the continuous Fréchet distance on high-dimensional trajectory data.
In contrast to known approaches, we describe how to effectively extend clustering based, generic metric indexes to dynamic data structures that answer proximity queries correctly but perform only a \emph{very small absolute number} of expensive distance calls.
We call this metric index structure Cluster Center Tree ($\CCT$).
Using contemporary desktop hardware, our publicly available, single threaded Matlab implementation allows to answer exact proximity queries over a $10$M trajectory data set with $1.04$ distance calls (latency below $1$ second) on average.

\begin{table}[h!]
	\centering   
	\begin{tabular}{p{0.22\textwidth}|p{0.3\textwidth}|p{0.4\textwidth}}
		
		~ & Proposed CCT & Related Work\\
		\hline
		Data Structure Size & linear & \multirow{2}{*}{exponential~\cite{ber17,dre17,ind03}} \\ \cline{1-2}
		Construction Time & Variants with $\mathcal{O}(|\Scal|^2)$, but practically fewer, or zero distance calls. 
		\\
		\hline
		Query Types &
		Exact, approximate, and min-error queries for $\NN$, $\kNN$, and $\RNN$ under $\delta_{F}$. &
		Not for $\delta_{F}$~\cite{dre17,ind03,xie-17}, only approximate~\cite{ber17,dre17,ind03}, only $\RNN$~\cite{bal17,buch17,dut17}, or $\NN$~\cite{dre17,ind03} only.\\ \hline
		$\delta_{F}$ Calls & Very few in constructions and queries. & Order of magnitude more~\cite{ciaccia1997,BeygelzimerKL06}.\\
        \hline
		Empirical Evaluation & $16$ real and over $20$ synthetic data sets with up to ${|\Scal|=10\text{M}}$ and $d=32$. & No experiments~\cite{dre17,ind03} or few for ${d=2}$ only~\cite{buch17,dut17,ber17,ciaccia1997,gud15,xie-17}.
		\\ \hline
	\end{tabular}
		\caption{CCTs jointly satisfy many relevant practical aspects whereas related works (c.f. Section~\ref{ssec:related_work}) typically neglect at least one aspect.}
	\label{tab:comp_rel_work}
\end{table}

Our approach is based on an extendable set of heuristic distance and decision algorithms, which is exchangeable for indexing other computationally expensive metric distance functions.
We improve on known heuristic bounds for $\delta_{F}$ and $\delta_{F\!D}$, which are also practical for high dimensional trajectory data (c.f. Section \ref{sec:Bounds}).

Known, generic clustering methods are transferable to $\CCT$s.
However, dynamic constructions with $\O(|\Scal|\gamma^6 \log|\Scal|)$ distance calls provide coarse cluster radii and static constructions with compactness guarantees use $\O(|\Scal|^2)$ distance calls.
The proposed dynamic and batch construction heuristics achieve $\CCT$s with compact clusters using only very few distance calls -- e.g. sub-linear on some instances.
Moreover, our approximate radii construction (not excluding exact proximity searches) still achieves compact clusters \emph{without any} distance calls (c.f. Section~\ref{sec:data_structures}).

We propose heuristic query algorithms that exploit low intrinsic dimensionality in the underlying metric for search space pruning -- i.e. excluding clusters of trajectories based on the triangle inequality.
To delay unavoidable $\delta_{F}$ and $\delta_{F\!D}$ calls to later stages, our methods leverage cluster compactness and bounds, exclude candidate trajectories based on orderings of the approximation intervals, and finally resolve remaining ambiguity with randomized pivoting for correct query results.
Inexpensive heuristic checks further save on some bound computations and our search algorithms naturally extend to queries that may contain approximate results
(c.f. Section~\ref{sec:queries}).

Given the aforementioned hardness of exact proximity searches and Fréchet distance computations, we evaluate scalability across various data set characteristics, quality of our $\CCT$ constructions, overall query efficiency, and pruning effectiveness with extensive experiments.
Observed query performances follow the proposed overlap and compactness metrics for $\CCT$ quality.
Our experimental results show improvement over recent, state-of-the-art approaches for $\RNN$ (even for $d=2$) and improvement over the generic Cover-Tree, M-Tree and the linear scan (even for $d>16$).
Moreover, the majority of the exact $\NN$ queries on our real world-data sets are solved \emph{without any} distance calls and further speed-ups are achieved on approximate queries (c.f. Section~\ref{sec:experiments}).

Summarizing aforementioned in Table~\ref{tab:comp_rel_work}, CCTs jointly satisfy many relevant practical aspects whereas related works (c.f. Section~\ref{ssec:related_work}) typically neglect at least one aspect.



\section{Preliminaries} \label{sec:preliminaries}
A trajectory $\Prm$ of size $m$ is a polygonal curve through a sequence of $m$ vertices $\langle p_1, \ldots , p_m\rangle$ in $\mathbb{R}^d$, where each contiguous pair of vertices in $\Prm$ is connected by a straight-line segment.
Let $n$ denote the maximum size of all trajectories in $\Scal$.
We reserve the term \emph{length} of a trajectory for the sum of the Euclidean lengths of its segments. 

\paragraph{Fr\'echet distance} 
\hyphenation{re-pa-ra-me-tri-za-tion} \hyphenation{pa-ra-me-tri-za-tion}
The continuous Fréchet distance $\FreDist$ between two trajectories $\Prm$ and $\Qrm$ can be illustrated as the minimum `leash length' required between a girl, who walks monotonously along $\Prm$, and her dog, who walks monotonously along $\Qrm$.
To simplify notation, we associate with a trajectory $P$ its natural parametrization ${P : [0,1] \to \mathbb{R}^d}$, which maps positions relative to the trajectories length to the spatial points -- e.g. $P(0.5)$ is the half-way point.
A continuous, monotonous map $f : [0,1] \to [0,1]$ is called a reparameterization, if $f(0) = 0$ and $f(1) = 1$.
Let $\mathcal{F}$ be the family of all reparameterizations, then the continuous Fr\'echet distance is defined as
\begin{equation*}
\delta_F(\Prm,\Qrm) = \inf_{f,g \in \mathcal{F}} \max_{\alpha \in [0,1]}
\Big\lVert
\Prm \Big(f\left(\alpha\right) \Big) - \Qrm \Big(g(\alpha) \Big) 
\Big\rVert, 
\end{equation*}
where $\lVert \cdot \rVert$ is the Euclidean norm in $\mathbb{R}^d$.
We refer to the continuous Fr\'echet distance as $\delta_{F}$ or \emph{distance} throughout this work, when it is clear from the context.
As noted above, most algorithms that compute $\delta_{F}$ base on several calls to an $\O(dn^2)$ time dynamic program which test if $\delta_{F}$ is at most some given value $\varepsilon$.
We denote this computation with the predicate $\delta_{F\!D}(P,Q,\varepsilon)$.

\paragraph{Discrete Fr\'echet distance}
The closely related discrete Fréchet distance minimizes over discrete, monotonous mappings $f : \{1,\ldots,m\} \to \{1,\ldots,m\}$ for a trajectory $\Prm$ of size $m$.
It is an upper bound to $\delta_{F}$, since only alignments of vertex sequences are considered.
In fact, the additive 
error 
is no more than the length of a longest line-segment in either trajectory ($\Prm$ or $\Qrm$).
Eiter and Mannila~\cite{eit94} gave a quadratic time algorithm, and Agarwal et al.~\cite{aga14} presented a (weakly) sub-quadratic algorithm for computing the discrete Fr\'echet distance which runs in $\mathcal{O}(mn \frac{\log\log n}{\log n})$ time.

Though DTW differs from discrete Fréchet only in replacing maximum with the summed distances of matched points, the triangle inequality can well be violated on irregular sampled trajectories
\footnote{The reader may consider DTW among the three 1D trajectories $\langle0,2\rangle, \langle 0,1,2\rangle$ and $\langle 0,1-\varepsilon, 1+\varepsilon,2\rangle$ as example.}.

\subsection{Proximity Search Problems} \label{ssec:ProximitySearchDefinition}
Our data structure for $\Scal$ is designed to handle both an additive error $\Eadd \geq 0$ and a relative error~$\Erel\geq 0$.
Though the computer science community prefers the later for algorithm analysis, our interaction with domain experts often leads to additive error specifications.
We only state the proximity search problems for the additive error regime, since replacing $+\Eadd$ with $\cdot(1+\Erel)$ provides those for the multiplicative.

\paragraph{\bf The $k$-Nearest-Neighbor Problem:}
\begin{itemize}
	\item[In:] A query trajectory $\Qrm$, an integer $k\geq 1$ and a non-negative real $\Eadd \geq 0$.
	\item[Out:] A set $\Scal_{k\mathsc{n\!n}} \subseteq \Scal$ of $k$ trajectories, such that for all $\Prm \in \Scal_{k\mathsc{n\!n}}$ we have
	$$
	\delta_F(\Prm,\Qrm)\leq \tau_k+\Eadd , 
	$$
	where $\tau_k$ denotes the $\kay$th smallest value in the set ${\{ \delta_{F}(P,Q): P \in \Scal \}}$.
\end{itemize}	

\paragraph{\bf The Range-Search Problem:}
\begin{itemize}
	\item[In:] A query trajectory $\Qrm$ and reals $\tau \geq 0$ and $\Eadd \geq 0$.
	\item[Out:] A set $\Scal_{\mathsc{rnn}} \subseteq \Scal$ of trajectories, such that both
	\begin{align*}
	\Scal_{\mathsc{rnn}} &\supseteq \{ P \in \Scal ~:~ \delta_{F}(P,Q) \leq \tau \} \quad, \text{ and}\\
	\Scal_{\mathsc{rnn}} &\subseteq \{ P \in \Scal ~:~ \delta_{F}(P,Q) \leq \tau + \Eadd\}
	\end{align*}
	hold.
\end{itemize}	

{

	\subsection{Intrinsic Dimensionality Measures of Metric Spaces} \label{ssec:DimensionMeasure}

	Let $\Scal$ be a set and the mapping $\delta: \Scal \times \Scal \to \mathbb{R}^+$ a metric on $\Scal$.
	For $P \in \Scal$ we denote with $B(P,\varepsilon) = \{ Q \in \Scal : \delta(P,Q)\leq \varepsilon\}$ the metric ball of radius $\varepsilon$.

	\paragraph{Doubling Constant \cite{GuptaKL03}}
	Let $\mu \in \mathbb{N}$ be the smallest number such that for every real $\varepsilon>0$, every ball in $\Scal$ of radius $\varepsilon$ can be covered by at most $\mu$ balls of radius $\varepsilon/2$.
	More formally, for every $P\in \Scal$ and $\varepsilon>0$ there exist $Q_1,\ldots,Q_\mu \in \Scal$, such that
	\begin{equation*}
	B(P,\varepsilon ) \subseteq \bigcup_{i=1}^\mu B(Q_i,\varepsilon/2).
	\end{equation*}
	
	\paragraph{Expansion Constant \cite{kargerR02}}
	Let $\gamma \in \mathbb{N}$ be the smallest number such that $$ \Big\lvert B(P,\varepsilon) \Big\rvert \leq  \gamma \Big\lvert B(P, \varepsilon/2) \Big\rvert$$ for every real $\varepsilon>0$ and $P\in \Scal$.
	
	We have $\mu \leq 4 \gamma$ for finite sets $\Scal$ (see e.g. Proposition 1.2 in \cite{GuptaKL03}).

\subsection{González Clustering for Metric Spaces} \label{ssec:gonzalez}
Our batch construction algorithms (c.f. Section~\ref{ssec:cct_construction}) are based on the following farthest-first algorithm for hierarchical, divisive clustering \cite{gonzalez85}.
Given a metric $\delta$ on a set $\Scal$, the algorithm successively adds new cluster centers to a set $L$.

	\begin{enumerate}
		\item[]{\bf González-Clustering ($\Scal,\delta$):}
		\item[] Arrays $\operatorname{dist}[~]=\infty$ and $\operatorname{parent}[~]=\emptyset$
		\item Pick $C \in \Scal$
		\item Set $L= \{C\},~\Scal=\Scal\setminus \{C\}$
		\item \label{algo:Gonzalez:loop}
		FOREACH $X \in \Scal$ with $\delta(X,C) < \operatorname{dist}[X]$\\
		\quad Set $\operatorname{dist}[X] = \delta(X,C)$ and $\operatorname{parent}[X]=C$
		\item Pick $C = \underset{X \in \Scal}{\operatorname{argmax}} \operatorname{dist}[X]$
		\item Set $L = L \cup \{C\}$ and $\Scal=\Scal\setminus \{C\}$
		\item If $\Scal \neq \emptyset$ GOTO \ref{algo:Gonzalez:loop} 
	\end{enumerate}

This algorithm requires no more than $\O(|\Scal|^2)$ distance computations.
The following statements on the algorithm's result quality, in terms of minimum cluster number $N(\Scal,\varepsilon)$ of a $\varepsilon$-cover and minimum cluster size $R(\Scal, k)$ of a $k$-center clustering, are well known~\cite{gonzalez85}.
To simplify notation, we use for subsets ${\mathcal A} \subseteq \Scal$ the abbreviation $\delta(P,{\mathcal A})=\min_{Q\in {\mathcal A}} \delta(P,Q)$ in the following formal definition:
\begin{align*}
R(\Scal, k)           &= \min_{{\mathcal A} \in \binom{\Scal}{k}} \max_{P \in \Scal} \delta(P,{\mathcal A})\\
N(\Scal, \varepsilon) &= \min_{{\mathcal A} \subseteq \Scal} \Big\{ \big\lvert {\mathcal A} \big\rvert : \delta(P,{\mathcal A}) \leq \varepsilon \quad \forall P \in \Scal \Big\}
\end{align*}

\paragraph{Cluster Size and Cover Number}
Let $C_1,\ldots,C_n$ denote the sequence in which the elements were added to $L$ and let $L(\varepsilon) = \{ C \in L : \operatorname{dist}[C] > \varepsilon\}$.
We have
\begin{alignat*}{2} R(\Scal, k) &\leq    \quad&\operatorname{dist}[C_k]      \quad           &\leq 2R(\Scal,k)              \quad\quad~~~\forall k > 1 \\
N(\Scal, \varepsilon)      &\leq  \quad  &|L(\varepsilon)|   \quad       &\leq ~N(\Scal, \varepsilon/2)  \quad\quad\forall \varepsilon > 0.
\end{alignat*}	

The main observation to prove these statements is the following algorithm invariant:
At all times $\varepsilon > 0$, any two elements in $L(\varepsilon)$ have distance of more than $\varepsilon$.
Hence, no metric ball of radius $\varepsilon/2$ can cover more than one element of $L(\varepsilon)$, which shows the Cover Number bounds.
To show the Cluster Size for some $k$, one observes that any two elements in $\{ C_1, \ldots, C_{k+1}\}$ have distance of at least $\operatorname{dist}[C_{k+1}]=:r$. Hence an optimal clustering with $k$ centers has to contain at least one cluster of radius $r/2$ (see e.g. \cite{SDasupta_Lecture}).

On metrics with bounded doubling constant $\mu$, we additionally have $N(\Scal,\varepsilon/2) \leq \mu \cdot N(\Scal,\varepsilon)$ for every $\varepsilon>0$.
This is a key ingredient for the use of `intrinsic dimensionality' in the analysis of nearest neighbor searches with Navigating-Nets~\cite{DBLP:conf/soda/KrauthgamerL04}, since refining the resolution of an optimal $\varepsilon$-net by a constant does not increases the number of clusters by more than a constant.

}



\section{Fr\'echet distance bounds} \label{sec:Bounds}
This section describes several fast algorithms for computing upper and lower bounds on the continuous Fr\'echet distance between two trajectories.
These distance approximations are used to speed up the construction of the data structure (Section~\ref{ssec:cct_construction}) and the query algorithms (Section~\ref{sec:queries}).

Table~\ref{tab:bounds_tab} contains an overview of the bounds together with their time complexities.
There are three groups of bounds: (i) a lower bound group $\LB_{\textsc{f}}$ (maximum of its bounds), (ii) a lower bound decision procedure $\LB_{\fd}$, and (iii) an upper bound group $\UB_{\textsc{f}}$ (minimum of its bounds).
The bound groups are applied in the construction and query algorithms.

Given two trajectories $\Prm=\langle p_1, \ldots , p_n\rangle$ and $\Qrm=\langle q_1,\ldots, q_m\rangle$ in $\Rd$, the aim of the algorithms below is to quickly compute upper and lower bounds on $\delta_F(\Prm,\Qrm)$.

\begin{table} 
	\centering \setlength{\tabcolsep}{0.5em} 
	\begin{tabular}{l|l|l|c|c|c} 
		\hline
		Group & Bound & Novelty & Output & Time & $d$\\
		\hline
        \multirow{4}{*}{$\LB_{\textsc{f}}$} & $\LB_{\textsc{sev}}$ & Known & $\R$ & $\mathcal{O}(d)$ & all\\ 
		& \multirow{2}{*}{$\LB_{\textsc{bb}}$} & \multirow{2}{*}{Improved} & \multirow{2}{*}{$\R$} & $\mathcal{O}(d^2  2^{d-1})$ & $d \le 3$\\
		& & & & $\mathcal{O}(d)$ & $d > 3$\\ 
		& $\LB_{\st}$ & New & $\R$ & $\mathcal{O}(1)$ & all\\
		\hline
		$\LB_{\fd}$ & $\LB_{\tr}$ & New & $\True/\False$ & $\mathcal{O}(d(n+m))$ & all\\
		\hline
	    \multirow{3}{*}{$\UB_{\textsc{f}}$} & \multirow{2}{*}{$\UB_{\textsc{bb}}$} & \multirow{2}{*}{Improved} & \multirow{2}{*}{$\R$} & $\mathcal{O}(2^{2d})$ & $d\le2$\\
	    & & & & $\mathcal{O}(d)$ & $d > 2$\\
		& $\UB_{\adf}$ & Improved & $\R$ & $\mathcal{O}(d(n+m))$ & all\\
		\hline
	\end{tabular}
		\caption{Overview of bounds and their time complexity for varying dimensions $d$ (c.f. Section~\ref{sec:Bounds}).}
	\label{tab:bounds_tab}
	\vspace{-.5cm}
\end{table}


\subsection{Start and End Vertices (SEV)} \label{ssec:SEV}
{\bf Lower bound.} A trivial lower bound on the distance between $P$ and $Q$ is the maximum of the Euclidean distances between start vertices $p_1$ and $q_1$, and between end vertices $p_n$ and $q_m$~\cite{bal17,dut17,buch17}. That is,
$\LB_{\textsc{sev}}(P,Q) = \max \{\| p_1 - q_1 \|, \|p_n - q_m\| \}$,
and it can be computed in $\mathcal{O}(d)$ time.


\subsection{Axis-aligned Bounding Box (BB)} \label{ssec:BB}
Let $\BB(P)$ denote the minimum-size $d$-dimensional axis-aligned box that contains all the vertices of $P$. It can be computed in $\mathcal{O}(dn)$ time, and, similarly, $\BB(Q)$ can be computed in $\mathcal{O}(dm)$ time.

{\bf Lower bound.} For $d > 3$ we use a lower bound described by D\"{u}tsch and Vahrenhold~\cite{dut17} and Baldus and Bringmann~\cite{bal17}. It computes the maximum of the following as a lower bound: the difference between the maximum $x_i$-coordinates of $\BB(P)$ and $\BB(Q)$ for each $1\leq i \leq d$, and the difference between the minimum $x_i$-coordinates of $\BB(P)$ and $\BB(Q)$ for each $1\leq i \leq d$. The running time of their algorithm is $\mathcal{O}(d)$.

For $d\le3$ we use a different algorithm to that in~\cite{dut17,bal17} which can result in a stronger lower bound on $\delta_F(\Prm,\Qrm)$.
Let $f$ be an edge (1-face) of $\BB(P)$ and let $f'$ be the corresponding edge of $\BB(Q)$, then $\lambda(f,f')$ is the minimum Euclidean distance, which may or may not be the perpendicular distance (e.g. Figure~\ref{fig:boundsabc}a).
Compute the maximum $\lambda(f,f')$ for all corresponding edges of $\BB(P)$ and $\BB(Q)$, which is clearly a lower bound on the Fréchet distance.
The number of edges of a $d$-dimensional bounding box is $d  2^{d-1}$, hence the running time is $\mathcal{O}(d\cdot d 2^{d-1})$. The lower bound $BB$ algorithm for $d\le3$ is denoted $\LB_{\textsc{bb1}}(P,Q)$, and the algorithm in~\cite{dut17,bal17} for $d > 3$ is denoted $\LB_{\textsc{bb2}}(P,Q)$.

{\bf Upper bound.}
For $d\le2$ we use the algorithm by D\"{u}tsch and Vahrenhold~\cite{dut17}, which computes the maximum of all pairwise distances between the vertices of $\BB(P)$ and $\BB(Q)$.
Since the running time of the above algorithm is $\mathcal{O}(2^{2d})$ we use the following modification for $d > 2$.
Compute a bounding box that contains all points of $P$ and $Q$, denoted $\BB(P,Q)$.
An upper bound is the Euclidean distance between two vertices of $\BB(P,Q)$, with the first vertex composed of minimum coordinate values for each dimension $d$, and the second vertex composed of maximum coordinate values for each dimension $d$.
The running time of this algorithm is $\mathcal{O}(d)$, though the upper bound in~\cite{dut17} is slightly stronger.
The upper bound $BB$ algorithm in~\cite{bal17} for $d\le2$ is denoted $\UB_{\textsc{bb1}}(P,Q)$, and the algorithm for $d\geq 3$ is denoted $\UB_{\textsc{bb2}}(P,Q)$.

{\bf Rotation.}
We can further improve $\LB_{\textsc{bb1}}$ and $\UB_{\textsc{bb1}}$ for trajectories that do not have a directional spine (direction of maximum variance on the point set) that aligns closely with an axis direction.
Typical examples of such trajectories can be found in some of the real-world data sets~\cite{pigeon16,NBA16,soccer15} used in Section~\ref{sec:experiments}.
To obtain a stronger bound for these cases pre-process two other bounding boxes for each input trajectory $P$ by rotating $P$ $22.5^{\circ}$, and $45^{\circ}$ counter-clockwise around the origin.
At query time, compute the $0^{\circ}$, $22.5^{\circ}$, and $45^{\circ}$ rotation bounding boxes for a query trajectory $Q$ only once.
Then, choose the maximum or minimum result from each of the three rotations as the lower or upper bound, respectively. 
Rotated trajectories can result in a smaller $BB$ and a stronger bound (e.g. Figure~\ref{fig:boundsabc}b). The rotations of $22.5^{\circ}$, and $45^{\circ}$ are heuristic values.

\begin{figure}
	\centering
	\includegraphics[width=\columnwidth]{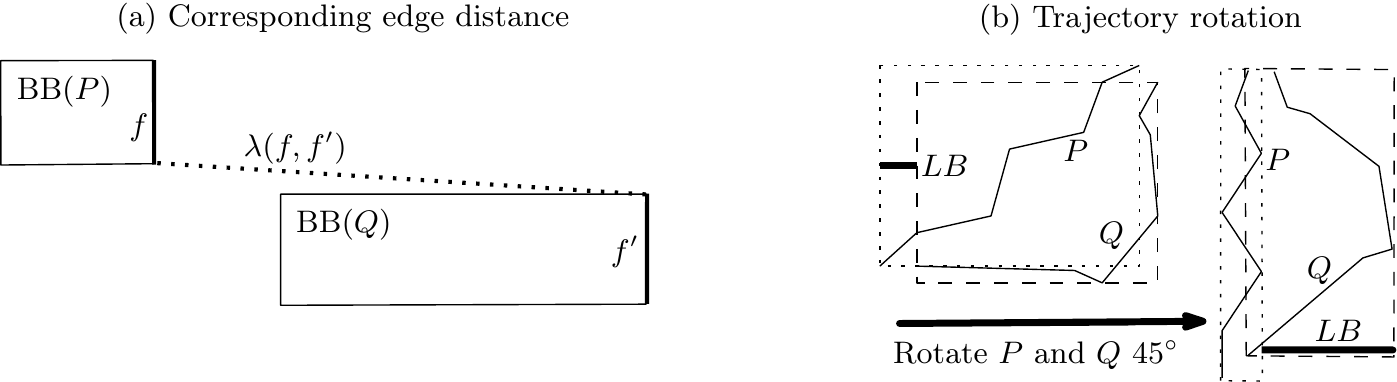} 
    \vspace{-.6cm}
	\caption{Bounding box lower bound: (a) corresponding edge distance, and (b) trajectory rotation resulting in a stronger lower bound (c.f. Section~\ref{ssec:BB}).}
	\label{fig:boundsabc}
	\vspace{-.4cm}
\end{figure}


\subsection{Simplified Trajectory (ST)} \label{ssec:ST}
{\bf Lower bound.} 
Let $P'$ be the straight-line segment between $p_1$ and $p_n$ and let $Q'$ be the straight-line segment between $q_1$ and $q_m$.
We set $\LB_{\st}(P,Q) = | \delta_F(P,P') -\delta_F(Q,Q')|/2$, which we next show is a lower bound for $\delta_F(P,Q)$.

\begin{theorem} 
	$\LB_{\st}(P,Q) \leq \delta_F(P,Q)$.
\end{theorem}
\begin{proof}
	From the triangle inequality, \[\delta_F(Q,Q')\leq \delta_F(Q,P) + \delta_F(\!P,P') + \delta_F(\!P',Q') \Longleftrightarrow\] \[\delta_F(Q,Q') - \delta_F(\!P,P') \leq \delta_F(\!P',Q') + \delta_F(\!P,Q) \le 2\delta_F(\!P,Q),\] \[since~~\delta_F(\!P',Q') \leq \delta_F(\!P,Q).\] 
	
	A similar argument can be used for $\delta_F(\!P,P')$, hence $|\delta_F(\!P,P') - \delta_F(Q,Q')|/2  \leq \delta_F(\!P,Q)$.
\end{proof}

To use this bound pre-compute $\delta_F(P,P')$ for each input trajectory $P\in \Scal$, in $\mathcal{O}(n \log n)$ time (since $P'$ is a single segment). At query time, once $\delta_F(Q,Q')$ is computed in $\mathcal{O}(m \log m)$ time, then every $\LB_{\st}(P,Q)$ check for the same query $Q$ is computed in constant time.


\subsection{Traversal Race (TR)} \label{ssec:NF}
{\bf Lower bound.} Our decision procedure $\LB_{\tr}(P,Q,\alpha)$ for $\alpha \ge 0$, is similar to the negative filter algorithm by Baldus and Bringmann~\cite{bal17}.
The algorithm starts at the beginning of $\Prm$ and $\Qrm$ and iteratively traverses $\Prm$'s vertices and $\Qrm$'s edges towards their respective ends.
To simplify presentation, we add a first edge $\overline{q_1 q_1}$ and a last edge $\overline{q_m q_m}$ to $\Qrm$. 
If the minimum Euclidean distance between $\Prm$'s vertex and $\Qrm$'s edge is less than the given $\alpha$, then advance to $\Prm$'s next vertex, else advance to $\Qrm$'s next edge.
If the end of $\Qrm$ is reached first, then $\alpha < \delta_F(P,Q)$ and answer $\True$, otherwise we have not gained any information and answer $\False$.

This algorithm gives a stronger bound than the algorithm in~\cite{bal17}, especially when the edges of the trajectories are long.
Since the algorithm is not symmetric, we run it a second time with $P$ and $Q$ swapped which gives a total runtime of $\mathcal{O}(d(n+m))$.


\subsection{Approximate Discrete Fr\'echet (ADF)} \label{ssec:ADF}
{\bf Upper bound.} The discrete Fr\'echet distance is known to be an upper bound on the continuous Fr\'echet distance~\cite{eit94}. A greedy algorithm in~\cite{bri16}, denoted $\UB_{\adff}(P,Q)$, approximates the discrete Fr\'echet distance between two trajectories $P$ and $Q$ in $\mathcal{O}(d(n+m))$ time. The approximation algorithm traverses the vertices of $P$ and $Q$ iteratively from start to end, starting at $i := 1$ and $j := 1$, and at each step picks a pair $(i', j') \in \{(i+1, j), (i, j+1), (i+1, j+1) \}$, minimizing the Euclidean distance between vertices $p_{i'}$ and $q_{j'}$.  It holds that $\delta_F(P,Q) \le\UB_{\adff}(P,Q)$~\cite{bri16}.

We include two more variations of the above algorithm. The first, $\UB_{\adfr}(P,Q)$, traverses the vertices of $P$ and $Q$ in \emph{reverse} from end to start, starting at $i := n$ and $j := m$, and at each step looks backwards to pairs $(i', j') \in \{(i-1, j), (i, j-1), (i-1, j-1)\}$, instead. 
The second, $\UB_{\adfd}(P,Q)$, traverses the vertices of $P$ and $Q$ from start to end, starting at $i := 1$ and $j := 1$, and at each step, if $n\ge m$ then increment $i$ and set $j := \lceil m/n \cdot i \rceil$, otherwise increment $j$ and set $i := \lceil n/m \cdot j \rceil$.

We also tried padding trajectories with a small number of new vertices along each edge of $P$ and $Q$ in an attempt to strengthen the bound. However, the experiments 
showed that this approach very rarely gave any improvements.

\section{Indexing Expensive Metrics with Cluster Center Trees} \label{sec:data_structures}

A Cluster Center Tree ($\CCT$) for a set $\Scal$ of trajectories is a rooted tree whose nodes represent clusters, that are metric balls of a certain distance radius.
Each node $v$ of a $\CCT$ stores a distance value $rad(v)$, a reference to some trajectory $C(v)$ (its center), and a list of child nodes.
Every trajectory $P \in \Scal$ appears as the center of a leaf
in the tree.
An internal node~$v$ of a $\CCT$, needs to uphold two properties, which are
(\emph{Nesting})
	one of its children refers to the same center as $v$, and
(\emph{Bounding})
	every descendant $u$ of $v$ has $\delta_{F} \big( C(u), C(v) \big) \leq rad (v)$. 
Since the number of leafs is $|\Scal|$ and each internal node has at 
least
two children, $\CCT$s have a storage consumption within $\O(|\Scal|)$, regardless of trajectories' size $n$ and dimensionality $d$.

The following describes three $\CCT$ batch construction algorithms (Exact, Relaxed, Approximate) and two dynamic insert/update/delete algorithms (Exact, Approximate), as well as a third insert algorithm (Standard) that similar common dynamic tree indexes use (e.g. the M-Tree~\cite{ciaccia1997}).


\subsection{Batch CCT Construction} \label{ssec:cct_construction}
Our 
batch construction
methods are inspired by González' hierarchical, divisive clustering for metric spaces to derive compact clusters (c.f. Section~\ref{ssec:gonzalez}).
Starting with one arbitrary element as the center, the algorithm successively picks an element, as an additional center, that is `farthest' from any of the previous centers,
and then reassigns elements to the additional center if it is closer.
A $k$-center clustering is produced in $k-1$ phases of distance computations and center reassigning. 
In each phase, the current cluster radii are within a factor of $2$ of an optimal $k$-center clustering that covers all elements (c.f. Section~\ref{ssec:gonzalez}).

Our construction heuristics foremost aim to avoid or reuse $\delta_{F}$ calls by applying upper and lower bound computations.

\begin{figure}
	\centering   
	\begin{subfigure}[]{.49\textwidth}
	\begin{tabular}{r|l|c}
		Order & If condition is \True & Return\\
		\hline
		$1$. & $\UB_{\textsc{f}}(\Prm,C_2) \le \LB_{\textsc{f}}(P,C_1)$ & $C_2$\\
		$2$. & $\UB_{\textsc{f}}(\Prm,C_1) \le \LB_{\textsc{f}}(P,C_2)$ & $C_1$\\
		$3$. & $\LB_{\textsc{\fd}}\big(P,C_1,\UB_{\textsc{f}}(\Prm,C_2)\big)$ & $C_2$\\
        $4$. & $\LB_{\textsc{\fd}}\big(P,C_2,\UB_{\textsc{f}}(\Prm,C_1)\big)$ & $C_1$ \\ \hline
        $5$. & $\delta_F(\Prm,C_1) < rad(C_1)/2$ & $C_1$  \\ \hline 
        $6$. & $\delta_F(\Prm,C_1) < \LB_{\textsc{f}}(P,C_2)$ & $C_1$\\
        $7$. & $\delta_F(\Prm,C_1) > \UB_{\textsc{f}}(\Prm,C_2)$ & $C_2$\\
        $8$. & $\LB_{\textsc{\fd}}\big(P,C_2,\delta_F(\Prm,C_1)\big)$ & $C_1$\\
        $9$. & $\delta_{\fd}\big(\Prm,C_2,\delta_F(\Prm,C_1)\big)$ & $C_2$\\
        $10$. & otherwise & $C_1$\\
	\end{tabular}
	\end{subfigure}
~
\begin{subfigure}[]{.49\textwidth}
\includegraphics[width=\columnwidth]{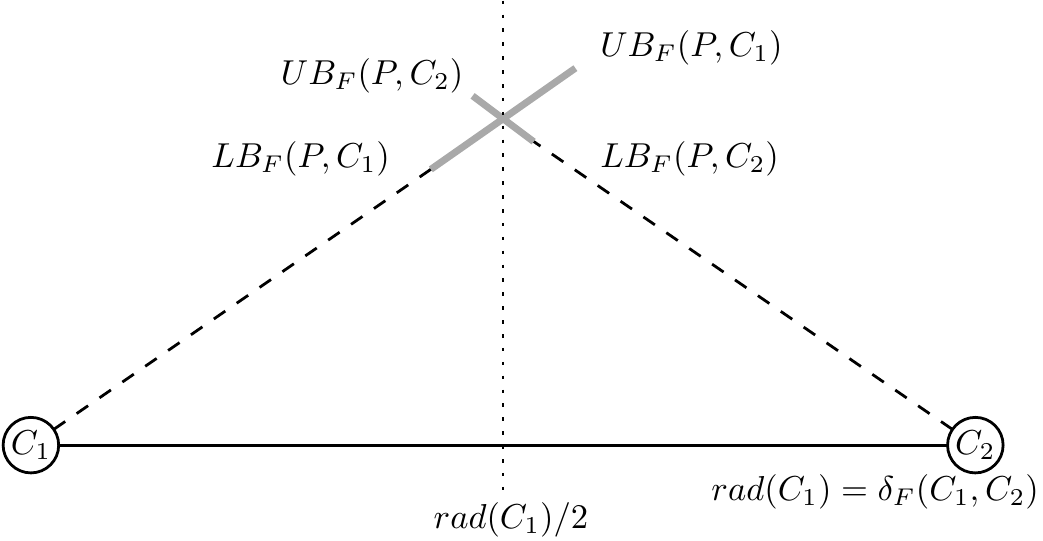}
\end{subfigure}
		\caption{Bisector Localization Predicate for determining if trajectory $\Prm$ is closer to center $C_1$ or $C_2$ (c.f. Section~\ref{ssec:cct_construction}).
		Subsequent checks are only performed if current results are inconclusive. Test $5$ is only performed for the Relaxed $\CCT$ since $C_2$ is a furthest trajectory in the cluster of $C_1$.}
	\label{tab:bisector}
\end{figure}


\subsubsection{Exact $\CCT$ Construction} \label{sssec:exact_cct}
To obtain a binary $\CCT$ from the González clustering algorithm in Section~\ref{ssec:gonzalez}, we consider it an continuous process within the monotonously decreasing
radius parameter $\varepsilon$.
In addition to the leafs $L(\varepsilon)$, we also track a set of tree nodes $T(\varepsilon)$.
Initially, $T$ contains only the root node which is associated to the sole trajectory $C_1$ in $L$ as its center.
Note that the array $\operatorname{parent}[\cdot]$ always points to a leaf for remaining elements in $\Scal$.

Now, whenever a new center $C_i$ is picked and added to the leaf nodes, we perform a split of its node in $T$.
That is, we replace the leaf's node $v$ that is currently associated to $\operatorname{parent}[C_i]$ in $T$ with a node that points to two children $v_1$ and $v_2$, which we associate with the leafs $\operatorname{parent}[C_i]$ and $C_i$.
To reduce the number of distance computations when determining if a given $\Prm$ is closer to $C_i$ or its current center, we use the sequence of bound computations in Figure~\ref{tab:bisector}.

After the tree is built, we compute the cluster radii of the $\CCT$ in a bottom-up fashion from each leaf.
To save $\delta_{F}$ distance calls, we use upper and lower bounds arrays instead of the $\operatorname{dist}[\cdot]$ array and sharpen approximations with $\delta_{F}$ calls only if selecting a furthest element is indecisive.
We use the following \emph{Fix-Ancestor-Radius} logic to save on $\delta_{F}$ calls.
First check the current radius against $\UB_{\textsc{f}}$, then check against it with $\LB_{\textsc{\fd}}$ and then $\delta_{\fd}$.
Only if these checks are indecisive, compute $\delta_{F}$ to update the radius of the node's parent.

The worst-case number of distance calls is $\mathcal{O}(|\Scal|^2)$, since
(i) on every iteration all bounds may fail to be conclusive and distances are computed for all trajectories $\Prm \in \Scal$, 
(ii) the $\CCT$ may degrade to a linear chain on metrics with large spread and asymmetric clusters (e.g. all trajectories are single, $1$D points with coordinates of the form $2^i \in \R$), and
(iii) Fix-Ancestor-Radius logic may perform up to quadratic $\delta_{F}$ calls.
However, our experimental data (Figure~\ref{fig:CCT_const_compare}) shows that this method performs far fewer $\delta_{F}$ calls on real data sets.

\subsubsection{Relaxed $\CCT$ Construction} \label{sssec:relaxed_cct}
This recursive construction algorithm successively performs only one phase of the González algorithm that results in a partition of the trajectories via the metric bisector of the two clusters' centers. This essentially omits the trajectory reassigning in González' clustering.

Pick an arbitrary trajectory $\Prm\in\Scal$ as center of the root node $v$,
that is $C(v):= \Prm$, and let $\Scal(v)=\Scal$ denote the trajectories contained in the cluster of $v$.
The recursive split then determines a trajectory $F(v)$ which is furthest from $C(v)$, which also determines $rad(v)$.
To do this, we first compute the highest lower bound $\alpha$ to the distances of $C(v)$ and elements of $\Scal(v)$.
Then we compute $\delta_{F}$ only for those trajectories whose upper bound distance (to $C(v)$) exceed $\alpha$.

The cluster is then partitioned into (potentially) smaller clusters $v_1$ and $v_2$, which are the children of $v$.
For their centers, we set $C(v_1) := C(v)$, $C(v_2) := F(v)$ and assign each trajectory $\Prm \in S(v)$ to the sub-cluster of the closer center.
To reduce the number of distance computations when determining if $\Prm$ is closer to $C(v_1)$ or $C(v_2)$, we use the test sequence in Figure~\ref{tab:bisector}.

The worst-case number of distance calls is again $\mathcal{O}(|\Scal|^2)$, since the algorithm may need to compute $\mathcal{O}(|\Scal|)$ distances at each level of the tree.
However, experimental results in Figure~\ref{fig:CCT_const_compare} shows that this method typically allows one to build $\CCT$s with $\O(|\Scal|)$ distance calls.


\subsubsection{Approximate $\CCT$ Construction} \label{sssec:approx_cct}
Since distance calls are very expensive, we also describe a construction algorithm that performs no calls at all to $\delta_{F}$ and $\delta_{F\!D}$, that originates from adapting the Relaxed construction.
For this, we only use upper bound computations $\UB_{\textsc{f}}$ to determine the furthest trajectory $F(v)$ and we assign $\Prm$ to the center, i.e. $C(v_1)$ or $C(v_2)$, that realizes a smaller upper bound value.
Compared to the Relaxed method, the approximate method does not perform expensive distance calls but the cluster radii are potentially larger.

 
\subsection{Dynamic $\CCT$ Constructions} \label{ssec:dynamic_cct}
Given the few properties $\CCT$s need to uphold, there are several heuristic strategies to handle dynamic situations.

\subsubsection{Exact Dynamic Inserts} \label{sssec:exact_dynamic_cct}
Exact inserts 
\emph{may} perform distance computations, since cluster radii values are computed exactly.

A new trajectory $\Prm$ is inserted by first locating the leaf $v_1$ that is an \emph{exact} nearest neighbor of $\Prm$ (c.f. Section~\ref{sssec:nn_algorithm}).
Then we create two new leaf nodes $u_1$ (contains trajectory of $v_1$) and $u_2$ (contains $\Prm$), and point $v_1$ to the new nodes.
Then fix the radius of $v_1$ and its ancestors using the already discussed `Fix-Ancestor-Radius' bottom-up process.

The worst-case number of distance calls is $\mathcal{O}(|\Scal|)$, since `Fix-Ancestor-Radius' may need to compute the distance for every tree node.
Hence, constructing a $\CCT$ entirely with dynamic inserts requires $\mathcal{O}(|\Scal|^2)$ distance computations.
However, our experiments show that the number of distance calls is much smaller for our data sets (c.f. Figure~\ref{fig:CCT_const_compare}).

\subsubsection{Approximate Dynamic Inserts} \label{sssec:approx_dynamic_cct}
Approximate inserts 
perform \emph{no} distance computations, and cluster radii are computed based on the largest upper bound value.

A new trajectory $\Prm$ is inserted by first locating the leaf $v_1$ that is an \emph{implicit approximate} nearest neighbor of $\Prm$ (c.f. Section~\ref{ssec:imp_approx_queries}).
Then we create two new leaf nodes $u_1$ (contains trajectory of $v_1$) and $u_2$ (contains $\Prm$), and point $v_1$ to the new nodes.
Then fix the radius of $v_1$ and its ancestors by only checking the current radius against $\UB_{\textsc{f}}$.

\subsubsection{Standard Dynamic Insert} \label{sssec:standard_dynamic_cct}
A classic insertion method for metric tree indexes~\cite{uhlmann1991met,uhlmann1991sat,kalantari-83,hetland-09} is to start at the root and descend to the child node whose center is closest to new trajectory $\Prm$, until a leaf $v$ is reached.
We adapt this algorithm for our setting by descending to the child node with the closest $\LB_{\textsc{f}}$ to locate leaf $v$, and then proceed with the same logic as the approximate insert above.


\subsection{CCT Quality Analysis} \label{ssec:cct_analysis}

\begin{figure}
\begin{tikzpicture}

\begin{groupplot}[
     group style = {group size = 1 by 3,
                    horizontal sep=0.15cm,
                    vertical sep=0.15cm,},
     width=11.0cm,
     height=3.5cm,
     grid style=dashed,
     xmode = normal,
     xtick=data,
     xtick style={draw=none},
     xticklabels={,,},
     xmin=0.5, xmax=6.5,
     ymajorgrids=true,
    ]

\nextgroupplot[
xticklabels={,,},
ybar,
ymin=0.8, ymax=12,
ytick={1,10,100},
yticklabels={$10^0$,$10^1$},
ylabel={Depth},
ylabel shift = 0.09cm,
bar width = 0.15cm,
legend pos=outer north east,
ymode = log,
log ticks with fixed point,
log origin = infty,
]
\addplot[color=black,fill=black!0,]
coordinates {(1,1.5)(2,1.5)(3,1.6)(4,4.5)(5,1.9)(6,0.7)};

\addplot[color=black,fill=black!10,]
coordinates {(1,1.6)(2,1.7)(3,1.5)(4,7.7)(5,1.8)(6,8.7)};

\addplot[color=black,fill=black!25,]
coordinates {(1,1.6)(2,1.7)(3,1.6)(4,8.2)(5,1.9)(6,10.4)};

\addplot[color=black,fill=black!35,]
coordinates {(1,1.4)(2,1.5)(3,1.4)(4,1.4)(5,1.4)(6,1.6)};

\addplot[color=black,fill=black!50,]
coordinates {(1,1.5)(2,1.5)(3,1.5)(4,1.5)(5,1.5)(6,1.6)};

\addplot[color=black,fill=black!90,]
coordinates {(1,1.3)(2,1.3)(3,1.3)(4,1.3)(5,1.2)(6,1.4)};

\legend{Exact $\CCT$}
\addlegendentry{Relaxed $\CCT$}
\addlegendentry{Approx. $\CCT$}
\addlegendentry{Exact Inserts}
\addlegendentry{Approx. Inserts}
\addlegendentry{Standard Inserts}

\nextgroupplot[
xticklabels={,,},
ybar,
ymin=0.65, ymax=0.9,
ytick={0.7,0.75,0.8,0.85},
yticklabels={$0.70$,$0.75$,$0.80$,$0.85$},
ylabel={Compactness},
bar width = 0.15cm,
]
\addplot[color=black,fill=black!0,]
coordinates {(1,0.803)(2,0.772)(3,0.752)(4,0.704)(5,0.767)(6,0.6)};

\addplot[color=black,fill=black!10,]
coordinates {(1,0.736)(2,0.797)(3,0.792)(4,0.704)(5,0.795)(6,0.817)};

\addplot[color=black,fill=black!25,]
coordinates {(1,0.762)(2,0.845)(3,0.811)(4,0.718)(5,0.824)(6,0.829)};

\addplot[color=black,fill=black!35,]
coordinates {(1,0.740)(2,0.793)(3,0.766)(4,0.703)(5,0.767)(6,0.763)};

\addplot[color=black,fill=black!50,]
coordinates {(1,0.755)(2,0.808)(3,0.778)(4,0.707)(5,0.778)(6,0.759)};

\addplot[color=black,fill=black!90,]
coordinates {(1,0.791)(2,0.844)(3,0.840)(4,0.731)(5,0.846)(6,0.835)};

\nextgroupplot[
xticklabels={Hurdat2, Pen, Football, Geolife, Basketball, Taxi},
ybar,
ymin=0.35, ymax=1.1,
ytick={0.4,0.6,0.8,1.0},
yticklabels={$0.4$,$0.6$,$0.8$,$1.0$},
ylabel={Overlap},
ylabel shift = 0.15cm,
bar width = 0.15cm,
]
\addplot[color=black,fill=black!0,]
coordinates {(1,0.696)(2,0.903)(3,0.938)(4,0.751)(5,0.888)(6,0.3)};

\addplot[color=black,fill=black!10,]
coordinates {(1,0.469)(2,0.866)(3,0.821)(4,0.383)(5,0.783)(6,0.761)};

\addplot[color=black,fill=black!25,]
coordinates {(1,0.494)(2,0.961)(3,0.855)(4,0.470)(5,0.852)(6,0.818)};

\addplot[color=black,fill=black!35,]
coordinates {(1,0.811)(2,0.913)(3,0.953)(4,0.957)(5,0.948)(6,0.993)};

\addplot[color=black,fill=black!50,]
coordinates {(1,0.827)(2,0.970)(3,0.964)(4,0.964)(5,0.961)(6,0.995)};

\addplot[color=black,fill=black!90,]
coordinates {(1,0.858)(2,0.991)(3,0.969)(4,0.983)(5,0.981)(6,0.997)};

\end{groupplot}
\end{tikzpicture}
\caption{$\CCT$ Quality for batch construction and insertion algorithms on the six largest real data sets (c.f. Section~\ref{ssec:cct_analysis}).
	{\normalfont
The average leaf depth (top) is normalized to an optimal depth ($\lceil \log_2 |\Scal | \rceil$).
Compactness (middle) is the average ratio of child-parent radii, and overlap (bottom) is the average ratio of each trajectory's leaf depth and number of other node clusters that cover it. For the Taxi~\cite{taxiA11,taxiB10} data set the Exact $\CCT$ batch construction did not finish within $3$ days and is omitted.
}
}
\label{fig:CCT_quality}

\end{figure}
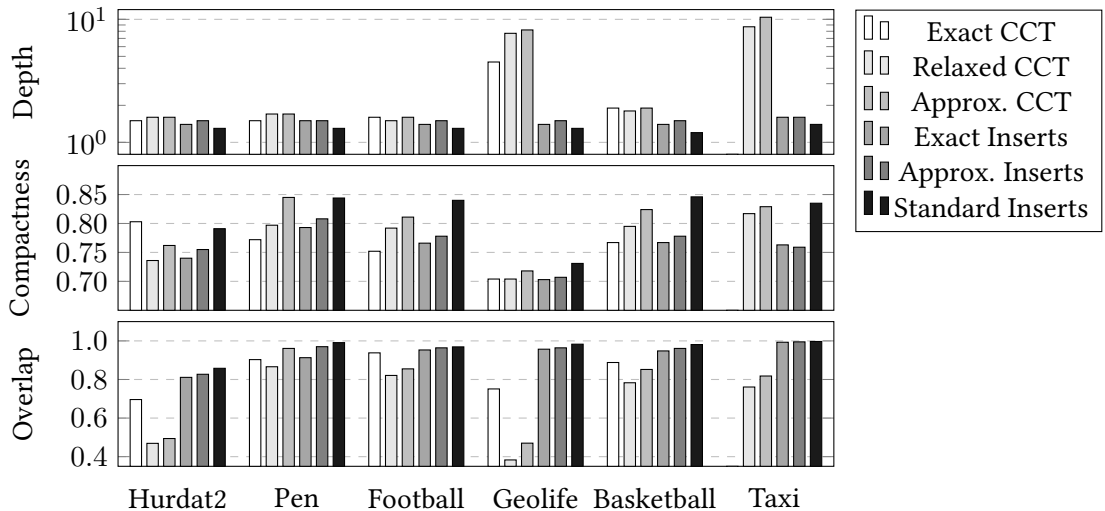

\begin{figure*}\centering
\begin{center}
	\includegraphics[width=\textwidth]{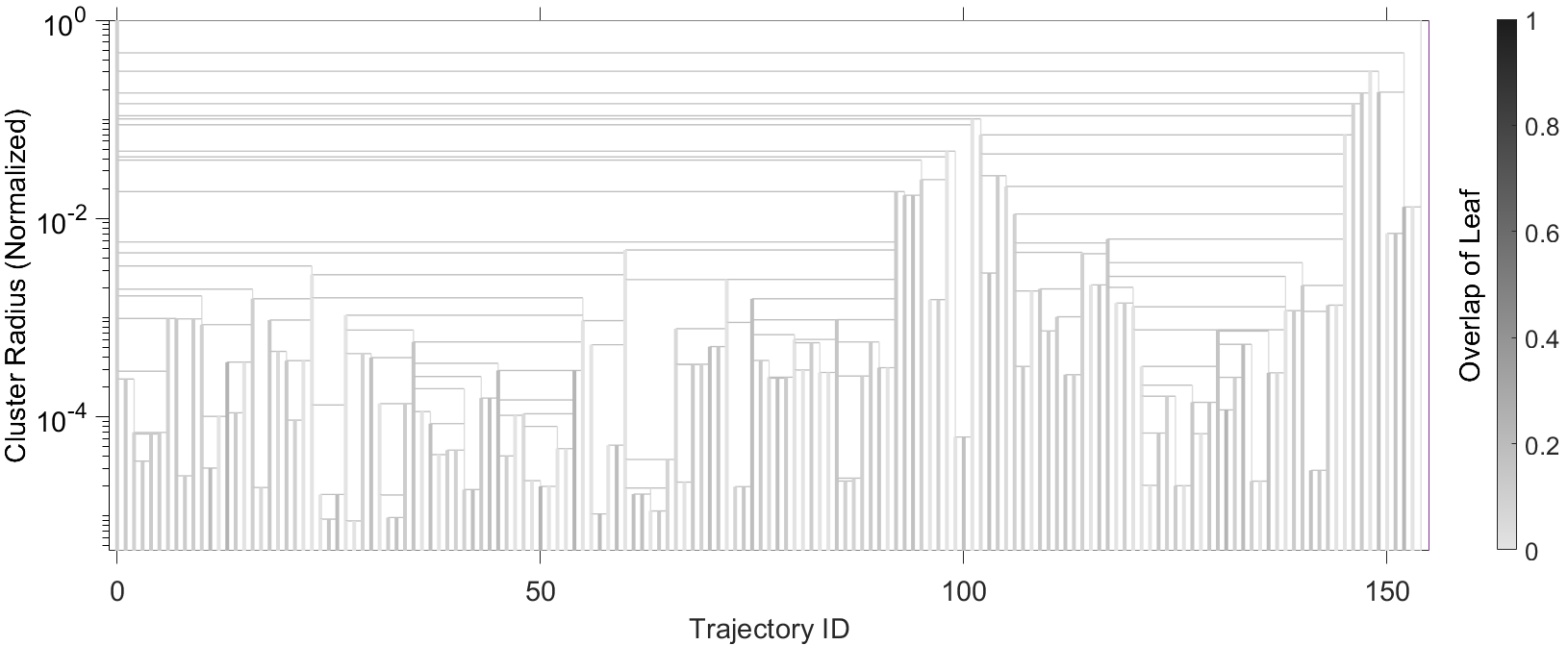} \\
	\vspace{0.5cm}
	\includegraphics[width=\textwidth]{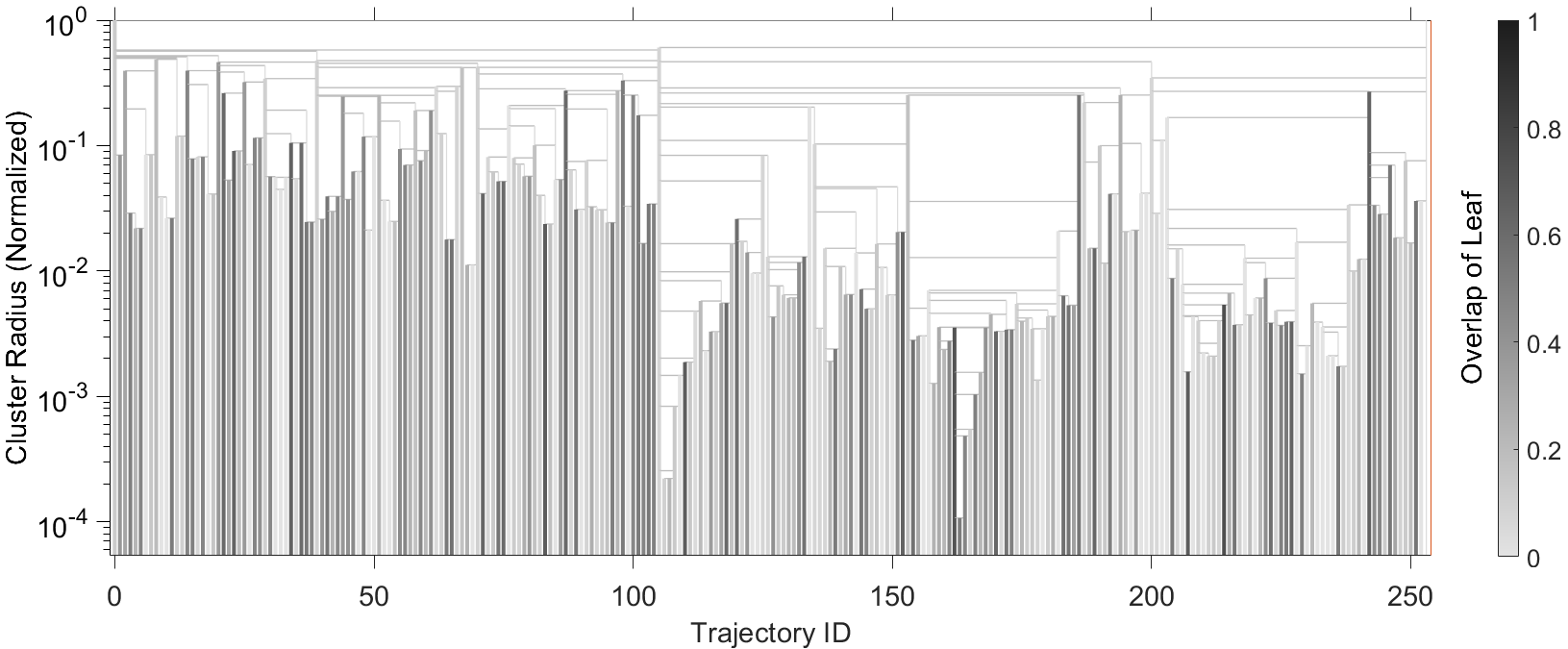}
\end{center}
	\caption{Relaxed $\CCT$ Dendrograms for Cats (top) and Gulls (bottom) real data sets (c.f. Section~\ref{ssec:cct_analysis}).
	{\normalfont
    The x-axis shows the number of input trajectories, y-axis the normalized cluster radii (compactness), and horizontal lines the parent nodes (tree depth). The vertical lines represent the leaf trajectories with lighter and darker shades corresponding to smaller and larger overlap measures, respectively.
    Tree balance is observed by the relative position of the vertical cut line beneath a parent that separates its two children.}
	}
	\label{fig:dendro_samples}
\end{figure*}

To gain insight of the $\CCT$ quality achieved by the various batch and insert algorithms, refer to Figure~\ref{fig:CCT_quality} (see Section~\ref{sec:experiments} for the complete experimental setup).

The average leaf depth is more balanced for the insertion algorithms compared to the batch construction algorithms.
However, it is noteworthy that tree depth is inversely proportional to the performance of the construction and query algorithms (see Figure~\ref{fig:CCT_const_compare} in Section~\ref{sssec:primary_results}).
E.g. unbalanced $\CCT$s do not necessarily incur poor query performance.
This may seem counter-intuitive at first, but surveys have mentioned that this can occur~\cite{hetland-09}, and the next two $\CCT$ quality measures help to explain why.

The compactness measure tends to be largest for the standard insert and smallest for the exact batch construction, which correlates with the experiment performance mentioned above.
So, a smaller compactness results in better performance. Moreover, when isolating just the insert algorithms, the exact method tends to have smaller compactness compared to approximate methods, which also correlates with the experimental results where exact inserts outperform approximate insert methods. 
But the exact and relaxed batch construction compactness measures do not correlate with the experiment performance results.
So we used ``overlap'' to explain the $\CCT$ quality in this case.

To measure overlap, we count all nodes that overlap (cover) a given leaf trajectory. 
We refine this measure by comparing the depth of each leaf with the total number of cover-nodes and averaging over all leafs, but the key point is that it is simply measuring how much of the tree covers each leaf.
Smaller overlap measures result in better query performance, and vice versa. 
Intuitively this method of measuring overlap makes sense, since data sets with higher intrinsic dimensionality contain trajectories that are harder to 'separate' from each other, which can result in higher overlap in a tree.
If a leaf is covered by many nodes, then constructing and searching is harder since there are more potential nodes to traverse.
The batch construction algorithms tend to have smaller overlap than inserts, and exact algorithms have smaller overlap than their approximate counterparts (since approximate algorithms can result in larger radii).

One interesting and initially unexpected result in the experiments was that the Relaxed $\CCT$ outperformed the Exact $\CCT$.
The Relaxed $\CCT$ is constructed with fewer distance calls and essentially omits the trajectory reassigning component, compared to the Exact method, so we anticipated a trade-off at query time for the Relaxed method.
However, the opposite occurred.
The reason for this behavior is due to the overlap difference.
The trajectory reassigning component of the Exact batch construction can lead to a larger overlap since trajectories can be reassigned multiple times during the iterations which can lead to more parent nodes that cover them. 

Various data sets can also exhibit different quality measures depending on their intrinsic dimensionality. Figure~\ref{fig:dendro_samples} compares two real data set Relaxed $\CCT$ dendrograms. The Cats~\cite{cats16} data set has smaller intrinsic dimensionality compared to the Gulls~\cite{gulls15} data set, and the dendrograms show this relationship with Cats having smaller compactness and overlap measures. Experiments (e.g. Figure~\ref{fig:kNN_real}) verify that the Cats Relaxed $\CCT$ outperforms the Gulls Relaxed $\CCT$.

An attempt was made to measure the quality of the underlying data sets using the intrinsic dimensionality measure of~\cite{chavez2001}. 
Calculations showed that this measure was useful for data sets with normal distributions of pairwise distances, however, most real data sets in our study do not have this property and the measure did not accurately convey the underlying intrinsic dimensionality.
In our setting, the overlap measure was a better indicator for the ease or difficulty of searching the data set.

\subsection{Differences to Related Approaches} \label{ssec:cct_other_approaches}

Multi-way metric indexes such as Cover-Trees~\cite{BeygelzimerKL06} also provide the Nesting property, besides additional compactness and separation properties (Cover Trees use $1/1.3\approx0.78$ for compactness and separation in practice to balance arity and depth).
Internal nodes of Cover-Trees have an assigned integer level and the distance between the center of a node with level $i$ and the center of any of its descendants is no more than $2^i$ (c.f. Theorem~$2$ in \cite{BeygelzimerKL06}).
Using these coarse values as radii, we have that every Cover-Tree is a $\CCT$. 
Their dynamic insertion and deletion of a single element 
performs no more than $\O(\gamma^6 \log |\Scal|)$ operations, which are mainly distance computations, where $2 \leq \gamma \leq |\Scal|$ denotes the expansion constant of the data set (c.f. Sections~\ref{ssec:related_work} and \ref{ssec:DimensionMeasure}).
For large trajectory data sets however, $\O(|\Scal|\gamma^6 \log |\Scal|)$ Fréchet distance computations might well be impractical, even for moderate $\gamma$ values.

Though one may modify $\CCT$s such that leafs store `chunks' (fixed size subsets of trajectories) like practical implementations do (e.g. M-Trees~\cite{ciaccia1997}),
this seems detrimental for the computationally expensive Fréchet distance in our setting.

It is important to note that
the bound algorithms in Section~\ref{sec:Bounds} are independent of the $\CCT$ structure.
This allows the flexibility to extend the query algorithms (c.f. Section~\ref{sec:queries}) with further, e.g. data domain specific, heuristic bounds \emph{without} the need to rebuild the data structure.
This is in strong contrast 
to pruning approaches that use $d$D-Trees~\cite{bentley75}, Range-Trees~\cite{bentley79-range-tree}, and grid-based hash structures, as in~\cite{buch17,dut17,ber17}, for e.g. trajectories' start and end points in $\Rd$. 

\section{Proximity Queries} \label{sec:queries}
Our query algorithms for $\CCT$s consists of three stages:

\begin{enumerate}
	\item {\bf Prune:} Collect candidate trajectories into a set $\CandCurv$ by performing a guided depth-first-traversal of the $\CCT$, in which sub-tree's clusters may be excluded in a pre-order fashion using the triangle inequality, the cluster radius, and bound computations.
	\item {\bf Reduce:} Filter trajectories in $\CandCurv$ using heuristic proximity predicates and orderings of the approximate distance intervals to obtain a smaller set $\Sstarcal$.
	\item {\bf Decide:} Finalize the result set by removing ambiguity in $\Sstarcal$ that exceeds the specified query error, by potentially performing $\delta_F$ and/or $\delta_{F\!D}$ calls.
\end{enumerate}
	
\begin{figure*}\centering
	\includegraphics[width=.45\textwidth,height=3.1cm]{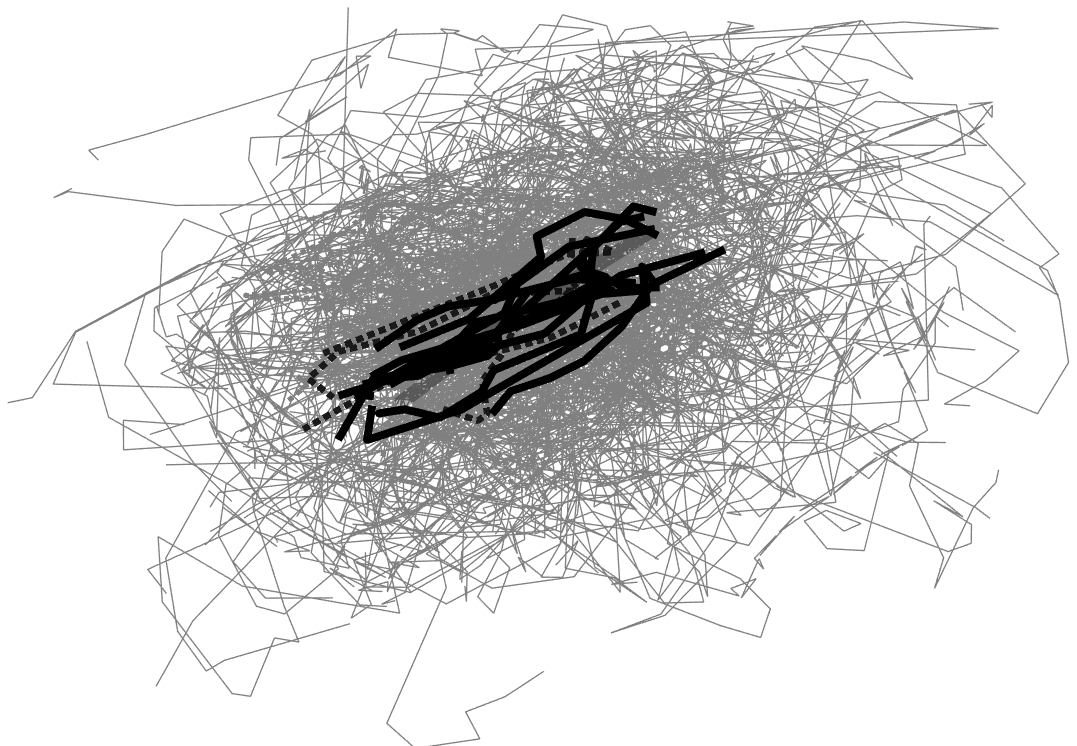} 
	\hspace{0.5cm}
	\includegraphics[width=.45\textwidth,height=3.1cm]{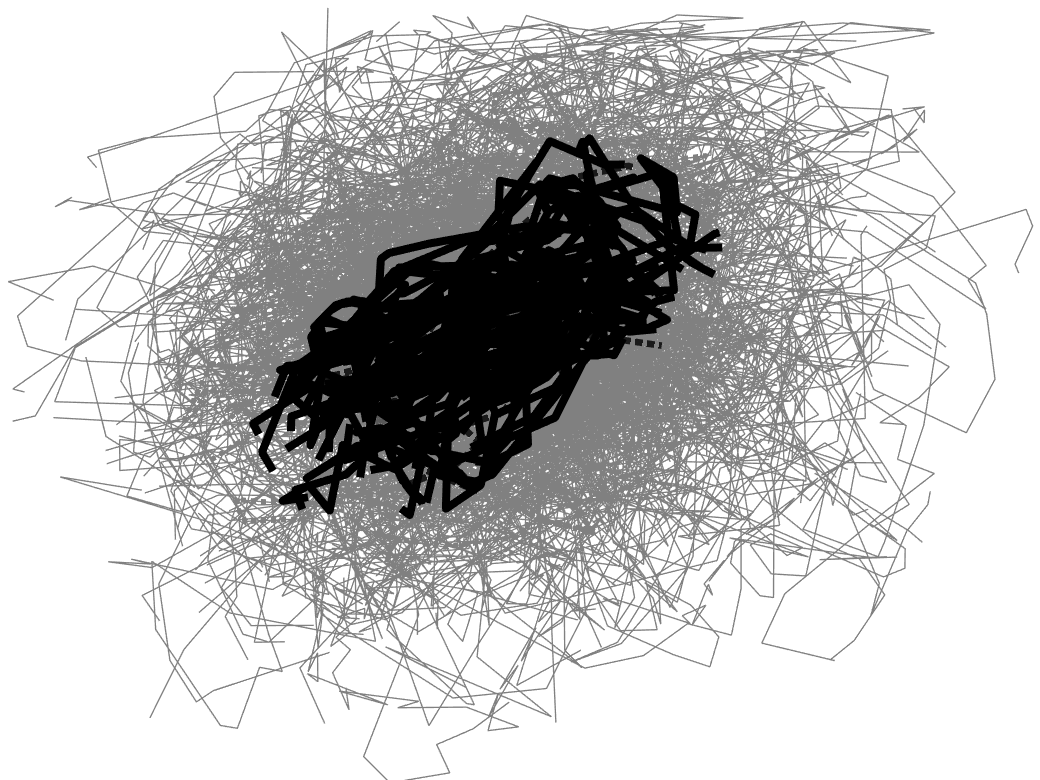} \\
	\vspace{0.5cm}
	\includegraphics[width=.45\textwidth,height=3.2cm]{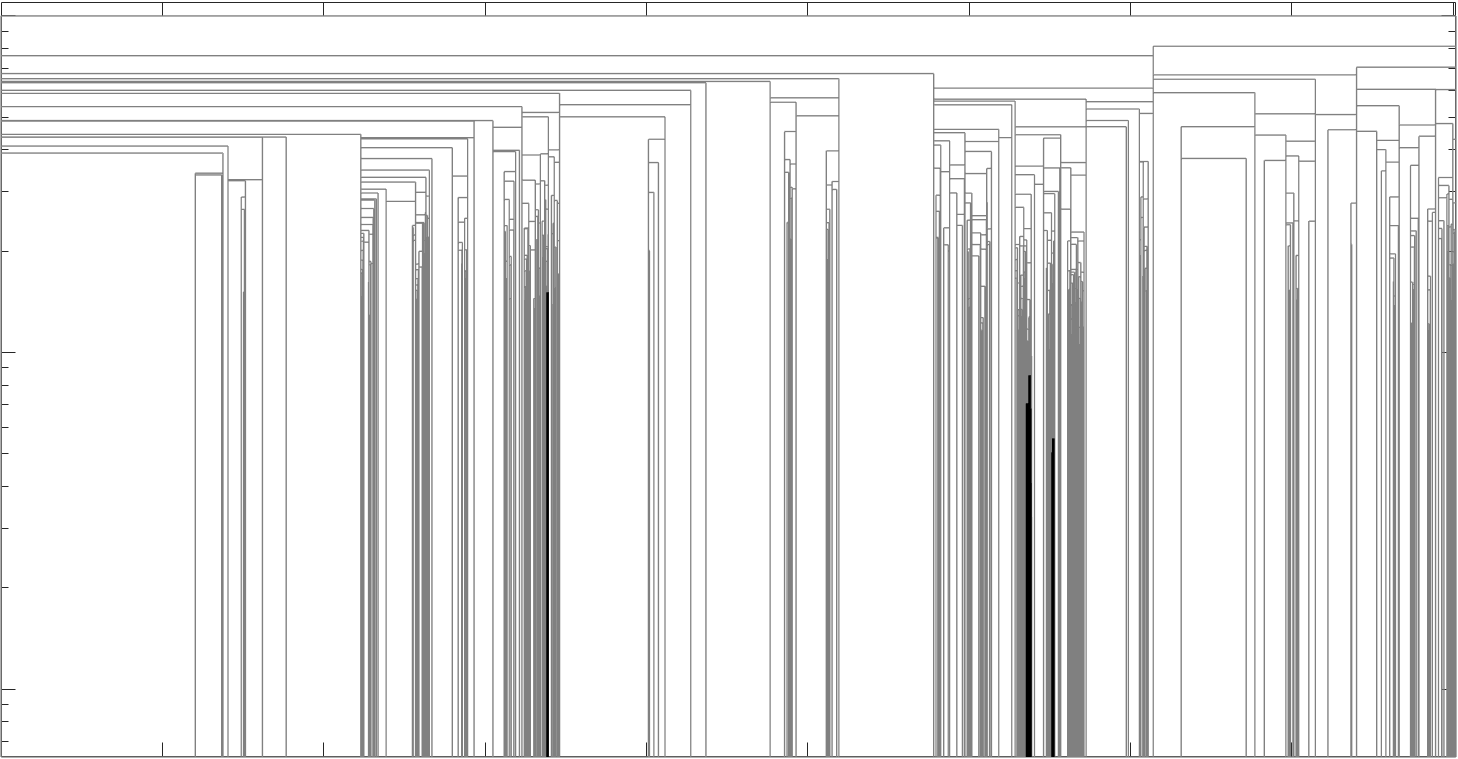} 
	\hspace{0.5cm}
	\includegraphics[width=.45\textwidth,height=3.2cm]{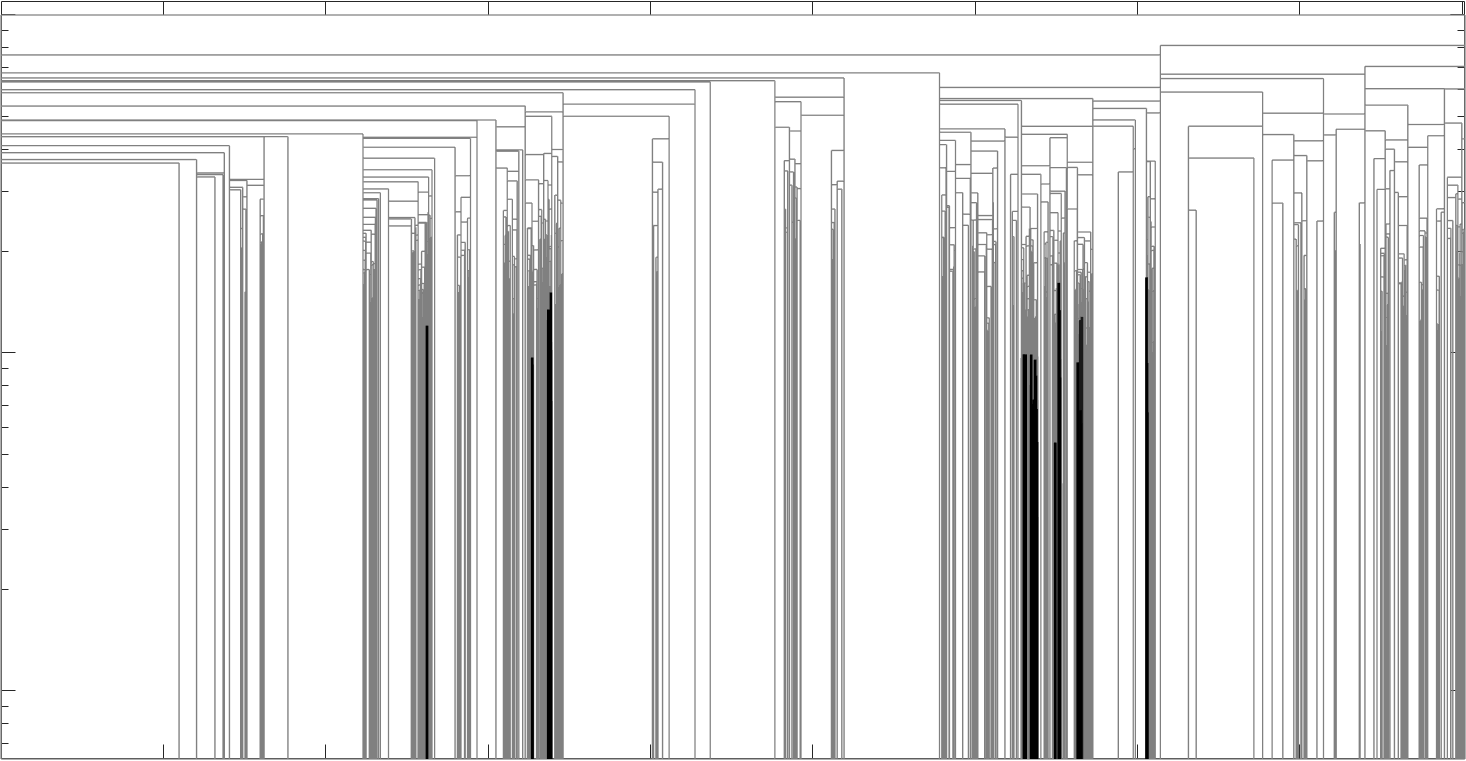} \\
	\vspace{0.5cm}
	\includegraphics[width=.45\textwidth,height=4.7cm]{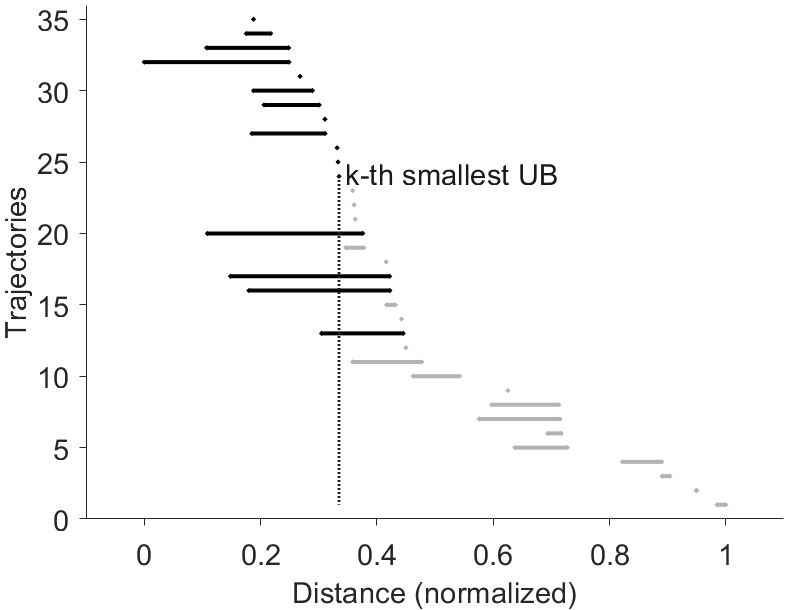} 
	\hspace{0.5cm}
	\includegraphics[width=.45\textwidth,height=4.7cm]{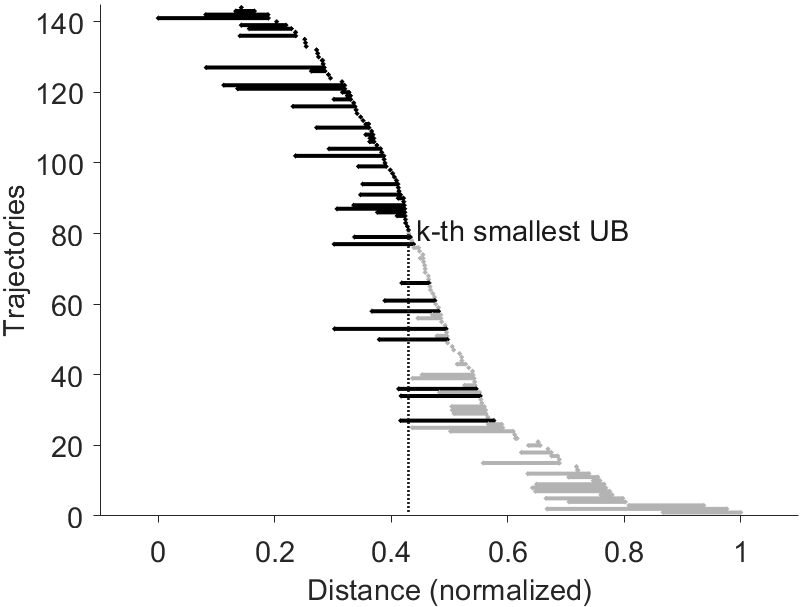} \\
	\vspace{0.5cm}
	\includegraphics[width=.45\textwidth,height=4.7cm]{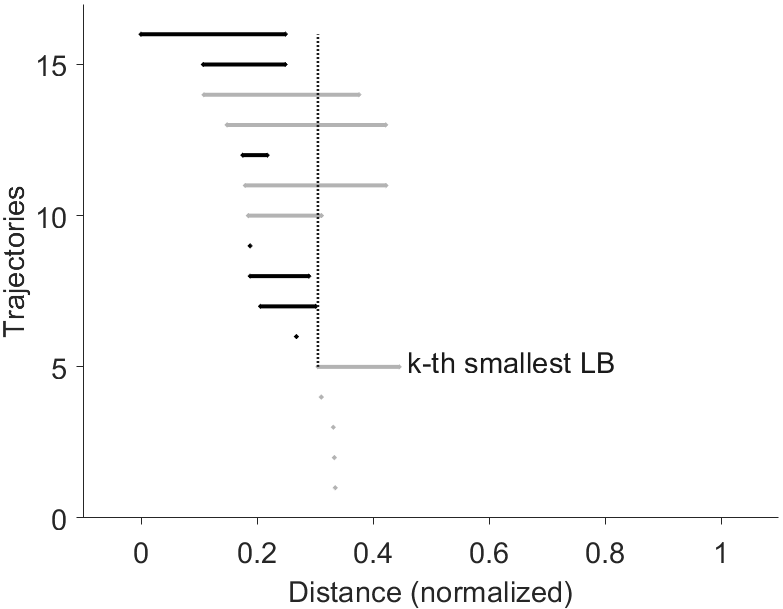} 
	\hspace{0.5cm}
	\includegraphics[width=.45\textwidth,height=4.7cm]{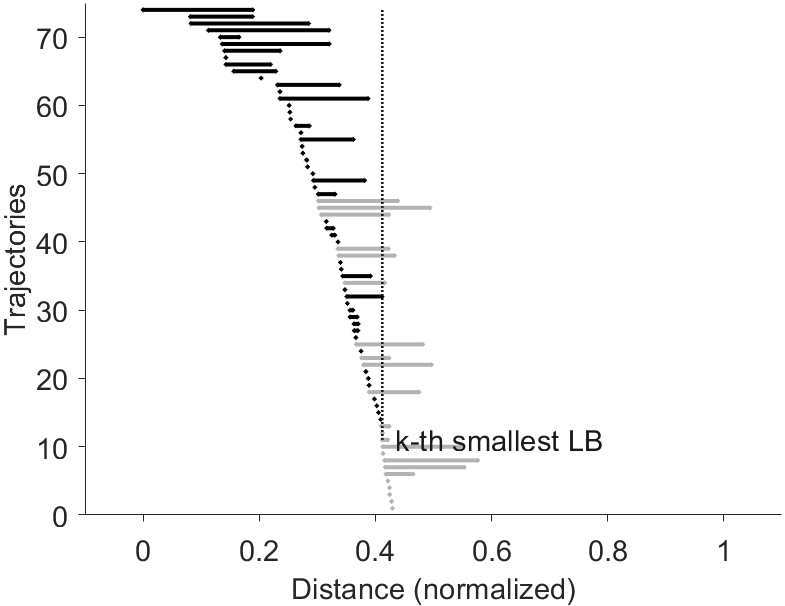}
	\caption{Two exact $\kNN$ queries for $k$=$12$ (first column) and $k$=$64$ (second column) on the Football~\cite{soccer15} data set using the Relaxed $\CCT$ (c.f. Section~\ref{sec:queries}).
	{\normalfont
	The first row shows $2$D trajectory plots and the second row contains dendrograms that show the $\CCT$ prune stage search, both with the same legend as in Figure~\ref{fig:teaser} (pruned trajectories are omitted).
    The third row shows trajectory bound intervals in $\Sstarcal$, i.e. the upper/lower bound distances of a trajectory to the query. 
    The trajectories in light grey show those that can be deleted in the reduce stage, since $\LB_{\textsc{f}}(P,\!\Qrm) + \Eadd > \beta_k$. The last row shows trajectory bound intervals in $\Sstarcal$, including those that can be included (black) in result set $\Scal_{k\mathsc{n\!n}}$ in the reduce stage, since $\UB_{\textsc{f}}(P,\!\Qrm) - \Eadd < \alpha_k$. }
	}
	\label{fig:football_knn_12}
	\vspace{-1em}
\end{figure*}

To gain some intuition regarding the effectiveness of this $3$ stage approach, refer to Figure~\ref{fig:football_knn_12}, which shows the Prune and Reduce stages for two $\kNN$ queries.
The Prune stage generally searches a small subset of the $\CCT$ (by eliminating sub-trees) and returns a small candidate set.
The Reduce stage can further exclude candidates, and also include candidates in the final result set.
Distance calls are only employed in the Decide stage, by which time the number of remaining candidates are typically small (or often zero).

The following describes each query algorithm in the additive error model and the changes for the multiplicative error model are briefly noted in each section. 

\subsection{Approximate and Exact kNN Queries} \label{ssec:knn_query}
Consider a query $k\mathsc{nn}(Q,\Eadd \geq 0,k \geq 1)$ on $\Scal$, as defined in Section~\ref{ssec:ProximitySearchDefinition}.
We describe the three stages of our query algorithm.

{\bf 1. Prune:}
Our query method heuristically guides the tree traversal towards a potentially close leaf. Recursively traverse the tree from the root, and for an internal node $v$, first descend to the child $u$ that has the smallest lower bound $\LB_{\textsc{f}}(\Qrm,C(u))$ among the children of $v$. When a leaf is reached, append its trajectory to the initially empty set $\Scal_1$.

Once $|\Scal_1| \ge k$, prune sub-trees as follows.
Track the $\kay$th smallest upper bound $\beta_k$ in $\Scal_1$ using a heap, 
and only descend below node $v$ if $\LB_{\textsc{f}}(C(v),\Qrm) \leq \beta_k + rad(v) - \Eadd $.
When a leaf node is reached, append its trajectory $P$ to $\Scal_1$ only if
$\LB_{\textsc{f}}(P,Q) < \beta_k$ and either $\UB_{\textsc{f}}(P,Q) < \beta_k$ or $\LB_{\textsc{\fd}}(P,Q,\beta_k)=\False$.

{\bf 2. Reduce:}
From $\CandCurv$, we filter with the final $\beta_k$ value to obtain at least $k$ elements in $\Sstarcal$.
That is, for those $P \in \CandCurv$ having  $\UB_{\textsc{f}}(P,\!\Qrm) > \beta_k$,
keep only those trajectories with $\LB_{\textsc{f}}(P,\!\Qrm) < \beta_k - \Eadd$ and $ \LB_{\textsc{\fd}}(P,Q,\beta_k - \Eadd) = \False$.

If $|\Sstarcal| = k$, we are done and return the set $\Scal_{k\mathsc{n\!n}} := \Sstarcal$.
Otherwise, locate the ($\kay+1$)-th smallest lower bound $\alpha_{k+1}$ in $\Sstarcal$. 
For each $P \in \Sstarcal$ with $\UB_{\textsc{f}}(P,\!\Qrm) - \Eadd < \alpha_{k+1}$ immediately move $P$ from $\Sstarcal$ to the initially empty set $\Scal_{k\mathsc{n\!n}}$.

For the relative error model, first compute the $\kay$th smallest lower bound $\alpha_k$ in $\CandCurv$, set $\Eadd:=\Erel \cdot \alpha_k$, and run stage two exactly as described above.

{\bf 3. Decide:} 
Perform the following until $|\Scal_{k\mathsc{n\!n}}| = k$.
Randomly choose a pivot trajectory $P \in \Sstarcal$, compute $\pi:=\delta_F(P,Q)$, and partition $\Sstarcal$ by computing if the trajectory is closer or further from $Q$ than $\pi$  (use upper/lower bounds, and if it's undetermined compute the Fr\'echet decision procedure).

If the number of trajectories closer to $Q$ than $\pi$ is at most $k - |\Scal_{k\mathsc{n\!n}}|$, append the closer trajectories to $\Scal_{k\mathsc{n\!n}}$ and delete them from $\Sstarcal$.
Otherwise, delete the trajectories further from $Q$ than $\pi$ from~$\Sstarcal$.

{\bf Algorithm Analysis.}
Using a similar analysis as in the QuickSelect algorithm~\cite{eppstein-QuickSelect}, the number of $\delta_F$ calls and $\delta_{\fd}$ calls in the Decide stage is $\mathcal{O}(\log|\Sstarcal|)$ expected and $\mathcal{O}(|\Sstarcal|)$ expected, respectively.
In the worst-case, no trajectories are discarded in the first two stages and $|\Sstarcal| = |\Scal|$. 
However, experiments (c.f. Section~\ref{sssec:primary_results}) show much fewer distance computations than this worst-case analysis.

\subsubsection{Optimization for $\NN$ Queries} \label{sssec:nn_algorithm}

We describe modifications for a $\NN$ algorithm that empirically performs slightly fewer distance computations than the $\kNN$ algorithm when $k=1$ (c.f. Section~\ref{sssec:supplementary_results}). 

{\bf 1. Prune:}
We perform the following additional check when at a leaf node $v$: If $\UB_{\textsc{f}}(C(v),\!\Qrm) \le \Eadd$ is $\True$ proceed to the next stage with $\Scal_1 := \{ C(v) \}$.

{\bf 2. Reduce:} Same as $\kNN$.

{\bf 3. Decide:}
If $|\Sstarcal| = 1$, we are done and return $\Sstarcal$.
Otherwise, compute the second-smallest lower bound $\alpha_2$ in $\Sstarcal$, with associated trajectory $\Prm$.
If $\LB_{\fd}(P,Q,\alpha_2) = \False$ but $\delta_{\fd}(\Prm,Q,\alpha_2) = true$ then return  $\{\Prm\}$.

Otherwise, sort $\Sstarcal$ ascending by the upper bound, and loop on each $\Prm \in \Sstarcal$ to track the 
current best trajectory $\Prm'$ and its distance $\pi := \delta_F(P',Q)$.
For subsequent $\Prm \in \Sstarcal$, if $\LB_{\fd}(P,Q,\pi) = \False$ but $\delta_{\fd}(\Prm,Q,\pi) = true$, then set $\Prm' := \Prm$ and $\pi := \delta_F(P,Q)$.
Finally return  $\{ \Prm' \}$.


\subsection{Approximate and Exact RNN Queries} \label{ssec:rnn_query}
Consider a range query $\mathsc{rnn}(Q,\Elmt\geq0,\Eadd\geq0)$ on $\Scal$, as defined in Section~\ref{ssec:ProximitySearchDefinition}.
For the queries under the relative error model, we set $\Eadd:=\Erel\cdot \tau$.

{\bf 1. Prune:}
Recursively traverse the tree from the root.
For an internal node $v$, only descend to its children if $\LB_{\textsc{f}}(C(v),\!\Qrm) \leq \Elmt + rad(v)$.
That is, the associated cluster of $v$ may contain trajectories within distance $\Elmt$ of $\Qrm$. 
When a leaf is reached, 
append its stored trajectory $P$ to the initially empty set $\CandCurv$ if $ \LB_{\textsc{f}}(P,\!\Qrm) \leq \Elmt$.

All trajectories within the cluster of a node $v$ may immediately belong in the result set $\Scal_{\mathsc{rnn}}$, so we can potentially finish the sub-tree of $v$ with a $\UB_{\textsc{f}}$ call.
Since our $\UB_{\textsc{f}}$ call is more expensive than $\LB_{\textsc{f}}$ calls, we speed up the search 
using a heuristic parameter\footnote{Our experiments use $\kappa=1.25$, since this matches the average upper/lower bound ratio we observe on elements of the data sets.}
$\kappa \geq 1$ in the following:
Only if $\kappa \cdot \LB_{\textsc{f}}(C(v),\!\Qrm) + rad(v) < \Elmt$ check $\UB_{\textsc{f}}(C(c),\!\Qrm) + rad(c) \leq \Elmt$ and, on success, simply append all leafs beneath $v$ to the initially empty set~$\Scal_{\mathsc{rnn}}$.

{\bf 2. Reduce:} For each trajectory $P \in \CandCurv$, if $\UB_{\textsc{f}}(P,\!\Qrm) < \Elmt + \Eadd$, then append $P$ to $\Scal_{\mathsc{rnn}}$, else if $\LB_{\textsc{\fd}}(P,Q,\Elmt) = f\!alse$ then append $P$ to initially empty set $\Sstarcal$, otherwise $P$ is discarded. 

{\bf 3. Decide:} For each trajectory $P \in \Sstarcal$, if $\delta_{\fd}(P,\!\Qrm,\Elmt) = \True$, then append $P$ to $\Scal_{\mathsc{rnn}}$.

{\bf Algorithm Analysis.}
In the worst case no trajectories are discarded in the first two stages, hence,  the query algorithm might perform $\mathcal{O}(|\Scal|)$ bound computations in the Prune and Reduce stages, and $\mathcal{O}(|\Scal|)$ Fr\'echet decision procedure computations in the Decide stage.

However, our experiments in Section~\ref{sssec:primary_results} (see Figure~\ref{fig:RNN_bringmann_compare}) show much fewer bound computations and $\delta_{\fd}$ calls.


\subsection{Implicit Approximate Queries} \label{ssec:imp_approx_queries}
We also describe a variant of $\kNN$ and $\RNN$ query algorithms that perform no distance and no Fr\'echet decision procedure computations.
Instead, implicit approximation query algorithms return trajectory results with the smallest additive $\Eadd$ or relative $\Erel$ approximation error, which is part of the output. 
Since results are determined by the set of heuristic bounds, this method can result in a significant computational speed-up over aforementioned query algorithms.

The Prune and Reduce stages of the implicit approximate $\RNN$ and $\kNN$ query algorithms are the same as their counterparts above with $\Eadd := 0$.
The modified Decide stages are as follows.

{\bf $\kNN$ Decide:}
If $|\Sstarcal| = k$, then set $\Scal_{k\mathsc{n\!n}} := \Sstarcal$. Otherwise, sort $\Sstarcal$ by upper bound ascending, and set $\Scal_{k\mathsc{n\!n}}$ to the first $k$ elements in $\Sstarcal$.

To compute $\Eadd$ and $\Erel$, set $\beta_k$ to the $k$-th smallest upper bound in $\Sstarcal$.
Delete the first $k$ elements in $\Sstarcal$, sort $\Sstarcal$ by lower bound ascending, and set $\alpha_k$ to the lower bound of the first element in $\Sstarcal$. 
Set $\Eadd := \beta_k - \alpha_k$. Set $\Erel := (\beta_k - \alpha_k) / \alpha_k$.

{\bf $\RNN$ Decide:}
Set $\Scal_{k\mathsc{n\!n}} := \Sstarcal$.

To compute $\Eadd$ and $\Erel$, set $\beta_k$ to the largest upper bound in $\Sstarcal$.
Set $\Eadd := \beta_k - \Elmt$ and $\Erel := (\beta_k - \Elmt) / \Elmt$.


\section{Experiments} \label{sec:experiments}

We experimentally evaluate the scalability, effectiveness and efficiency of bounds in Section~\ref{sec:Bounds}, data structure constructions in Section~\ref{sec:data_structures}, and query algorithms in Section~\ref{sec:queries}.
As introduced in Section~\ref{sec:introduction}, our measurements focus on the primary empirical goal of measuring the number of distance computations, with a subordinate goal of measuring the query I/O (tree node accesses).

We compare our contribution to several competitors, including a recent state-of-the-art contribution~\cite{bri19} for $\RNN$ queries among $2$D trajectories (which improves upon previous $\RNN$ search approaches on $2$D data~\cite{bal17,buch17,dut17}), a standard M-Tree~\cite{ciaccia1997}, a standard Cover-Tree~\cite{BeygelzimerKL06}, and an improved linear scan algorithm (Section~\ref{sssec:improved_linear_scan}).
Although the approach \cite{ber17} is most similar in regard of the supported operations, it does not allow practical comparison on our test data sets due to its exponential construction time and data structure size.

\subsection{Experiment Setup}
\label{ssec:experiment_setup}
We now describe how the experiments are setup whereas Section~\ref{ssec:experiment_results} discusses the results\footnote{See \url{https://github.com/japfeifer/frechet-queries} for more detailed experimental results, the code, and the data sets.}.

\subsubsection{Real Data Sets} \label{sec:real_data_sets}
We obtained sixteen real-world data sets \cite{kruger09,truckbus05,pigeon16,bats15,cats16,vessel05,geo2012,hurdat217,masked17,NBA16,soccer15,gulls15,pentip06,taxiA11,taxiB10} of diverse origin and characteristics to evaluate our data structure construction and query algorithms (see Table~\ref{tab:NNRes}).
To broaden our experiments, but also to challenge our bound algorithms, we use the trajectory simplification algorithm of \cite{aga05} to obtain trajectories whose sampling are irregular (c.f. Section~\ref{sec:real_data_sets}).
Given an error bound {$\widehat{\varepsilon} \geq 0$}, this simplification algorithm returns a trajectory over a subset of the original vertices whose Fréchet distance is within the specified bound.
For every $\Prm \in \Scal$, we set $\widehat{\varepsilon}$ to be a small percentage (typically $1\%$ or $2\%$) of $\reach(\Prm)$, where $\reach$ denotes as the maximum distance from a trajectory's start vertex to any of its other vertices (see e.g. \cite{ber17}).
We found that this substantially reduces the time required to run the experiments, without materially changing the results.

Though some of these real data sets have a small number of trajectories (e.g. Vessel-Y vs. Taxi), they are included in our experiments since they show that proximity queries in small sets can cause more distance calls than searches in larger sets (e.g. Figures~\ref{fig:NN_real}, \ref{fig:kNN_query_short_narrow}, \ref{fig:bound_effectiveness}, and \ref{fig:NN_vs_kNN}).

We use two methods to generate query trajectories for the real data sets.
Method one randomly selects an input trajectory $\Prm$, perturbs its vertices up to $3\%$ and translates it up to $5\%$ of $\reach(\Prm)$ uniformly at random.
For direct comparison, method two uses the query generator of~\cite{bri19}, that returns exactly $10, 100$ or $1000$ results for a $\RNN$ query.
We generated $1000$ query trajectories per data set with either method.
Results based on the second query generation method indicate that in the respective figure.

\begin{table*}[] 
	\centering
{ \setlength{\tabcolsep}{0.5em}
	\begin{tabular}{l|r|r|rr|l}
		\hline
		& & & \multicolumn{2}{c|}{Vertices} & \\
		Data Set & $|\Scal|$ & $d$ & orig. & simpl. & Trajectory Description \\
		\hline
Vessel-M~\cite{vessel05} 	& $106$ 	& $2$ 	& $23.0$ 	& $7.3$ 	& Mississippi river shipping vessels Shipboard AIS.\\
Pigeon~\cite{pigeon16} 		& $131$ 	& $2$ 	& $970.0$ 	& $26.0$ 	& Homing Pigeons (release sites to home site).\\
Seabird~\cite{masked17} 	& $134$ 	& $2$ 	& $3175.8$ 	& $43.5$ 	& GPS of Masked Boobies in Gulf of Mexico.\\
Bus~\cite{truckbus05} 		& $148$ 	& $2$ 	& $446.6$ 	& $40.3$ 	& GPS of School buses.\\
Cats~\cite{cats16} 			& $154$ 	& $2$ 	& $526.1$ 	& $34.2$ 	& Pet house cats GPS in Raleigh-Durham, NC, USA.\\
\hline
Buffalo~\cite{kruger09} 	& $165$ 	& $2$ 	& $161.3$ 	& $54.5$ 	& Radio-collared Kruger Buffalo, South Africa.\\
Vessel-Y~\cite{vessel05} 	& $187$ 	& $2$ 	& $155.2$ 	& $4.0$ 	& Yangtze river shipping Vessels Shipboard AIS.\\
Gulls~\cite{gulls15} 		& $253$ 	& $2$ 	& $602.1$ 	& $33.7$ 	& Black-backed gulls GPS (Finland to Africa).\\
Truck~\cite{truckbus05} 	& $276$ 	& $2$ 	& $406.5$ 	& $41.4$ 	& GPS of 50 concrete trucks in Athens, Greece.\\
Bats~\cite{bats15} 			& $545$ 	& $2$ 	& $44.1$ 	& $7.3$ 	& Video-grammetry of Daubenton trawling bats.\\
\hline
Hurdat2~\cite{hurdat217} 	& $1788$ 	& $2$ 	& $27.7$ 	& $7.9$ 	& Atlantic tropical cyclone and sub-cyclone paths.\\
Pen~\cite{pentip06} 		& $2858$ 	& $2$ 	& $119.8$ 	& $24.4$ 	& Pen tip characters on a WACOM tablet.\\
Football~\cite{soccer15} 	& $18034$ 	& $2$ 	& $203.4$ 	& $15.4$ 	& European football player ball-possession.\\
Geolife~\cite{geo2012} 		& $18670$ 	& $2$ 	& $1332.5$ 	& $14.2$ 	& People movement, mostly in Beijing, China.\\
Basketball~\cite{NBA16} 	& $20780$ 	& $3$ 	& $44.1$ 	& $7.3$ 	& NBA basketball three-point shots-on-net.\\
Taxi~\cite{taxiA11,taxiB10} & $180736$ 	& $2$ 	& $343.0$ 	& $41.3$ 	& $10$,$357$ Partitioned Beijing taxi trajectories.\\
\hline
	\end{tabular}
}
	\caption{Real data sets, showing number of input trajectories |\Scal|, dimensions $d$, average number of original vertices per trajectory, average number of simplified vertices per trajectory, and a description.}
	\label{tab:NNRes}
\end{table*}

\subsubsection{Synthetic Data Sets} \label{sec:syn_data_sets}

Testing on synthetic data sets helps to analyze which characteristics most impact the number of $\delta_{F}$ calls and overall query efficiency.
By varying a single characteristic while holding others constant, the impact of the particular characteristic on the measurements can be assessed.
The routine to create these data sets is parameterized by the following characteristics:

\begin{itemize}
\item cluster size $\alpha_{CS}$ (number of trajectories per cluster), 
\item trajectory straightness factor $\alpha_{SF}$ and maximum edge distance $\alpha_{ED}$, 
\item average trajectory size $n$,
\item number of trajectories $|\Scal|$, and
\item spatial dimensions $d$.
\end{itemize}

Our baseline synthetic data set is generated with the values
${\alpha_{CS} = 10}$, $\alpha_{SF}= 0.95$ with $\alpha_{ED}= 0.6$,
$n=15$,
$|\Scal|=5000$, and
$d=2$.
For the experiments, we vary
${\alpha_{CS} \in \{ 1,10,25,50,100\}}$, 
${\alpha_{SF} \in \{ 0.5,0.8,0.9,0.95,0.99\} }$, 
${n \in \{ 15,25,35,45,55\} }$, 
$d \in \{ 2,4,8,16,32\}$, and the number of trajectories
$|\Scal|$ in ${\{ 5\text{K}, 10\text{K}, 20\text{K}, 30\text{K}, 40\text{K}, 1\text{M}, 10\text{M}\}}$~.

Synthetic data sets and their associated query trajectories are created in the following four steps.

{\em Step 1: Unique (non-clustered) trajectories.}
First, increase the designated number of trajectories $|\Scal|$ by $500$.
Generate each of the $|\Scal|/\alpha_{CS}$ trajectories with the following random-walk routine.
Choose a number of vertices $z \in \left[\frac{n}{2},\frac{3n}{2}\right]$ uniformly at random and then choose the initial vertex $p_1 \in [0,1]^d$ uniformly at random.
Subsequent vertices $p_i$ are created with
$$
p_i := \alpha_{ED} \cdot \sigma + p_{i-1} + \alpha_{SF} \cdot (p_{i-1} - p_{i-2}) \quad,
$$
where each random step $\sigma \in [0,1]^d$ is chosen uniformly. 

{\em Step 2: Clustered trajectories.}
For each unique trajectory generate a copy of it, perturb uniformly at random the copy's vertices up to the maximum edge distance $\alpha_{ED}$, and then translate uniformly at random the copy up to the maximum edge distance.
This process is performed $\alpha_{CS} - 1$ times per unique trajectory.

{\em Step 3: Sample query trajectories.}
Out of the above set $\Scal$, we choose $1000$ trajectories uniformly at random without replacement.

{\em Step 4: Add `noisy' trajectories.}
Finally, $500$ additional `noise' trajectories are generated as in Step 1.

\subsubsection{Improved NN Linear Scan}
\label{sssec:improved_linear_scan}
Given the lack of available algorithms for exact nearest-neighbor search under the Fréchet distance and our discussion on the `curse of dimensionality' (c.f. Section~\ref{sec:introduction}), we implemented a competitor, called improved NN linear scan, suitable for high dimensional trajectory data.

The improved $\NN$ linear scan algorithm leverages our bounds of Section~\ref{sec:Bounds} by checking each $\Prm \in \Scal$, and appending $\Prm$ to the initially empty set $\Scal_1$ if $\LB_{\textsc{f}}(\Prm,\Qrm) < \beta$ and $\LB_{\fd}(P,Q,\beta) = f\!alse$.
The smallest upper bound $\beta$ is tracked, upper bound $\UB_{\textsc{f}}(P,Q)$ is only computed when $\Prm$ is appended to $\Scal_1$, and $\LB_{\fd}(P,Q,\beta)$ is only computed when $\LB_{\textsc{f}}(\Prm,\Qrm) < \beta$.

\subsubsection{Quality of the Data Structure} \label{sssec:Exp-Quality-Data-Structure}


\subsection{Experimental Results}
\label{ssec:experiment_results}
Note that the results on the quality of the $\CCT$ data structure are in Section~\ref{ssec:cct_analysis}.
Experimental results are separated into primary results, which evaluate the proposed Relaxed $\CCT$ method on real and synthetic data sets and compare it with related work, and supplementary results, which compare the different exact and approximate variations of our approaches against each other.

\subsubsection{Primary Results} \label{sssec:primary_results}
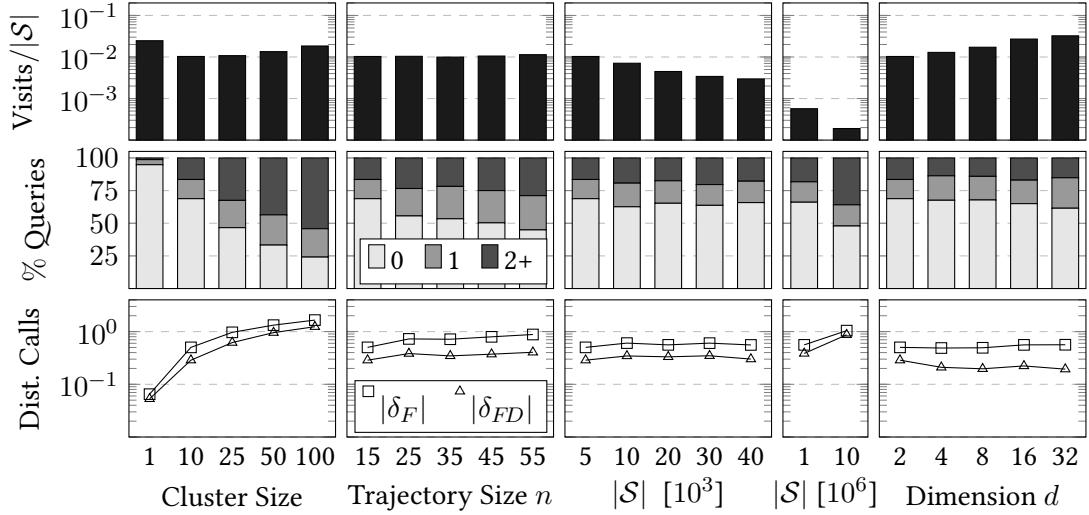
\begin{figure}[]
\vspace{-0.05cm}

\begin{tikzpicture}

\begin{groupplot}[
     group style = {group size = 5 by 3,
                    horizontal sep=0.15cm,
                    vertical sep=0.15cm,},
     height=3.4cm,
     width=4.3cm,
     grid style=dashed,
     xmode = normal,
     xtick=data,
     xtick style={draw=none},
     xticklabels={,,},
     xmin=0.5, xmax=5.5,
     ybar stacked,
     ymin=0, ymax=105,
     ytick={25,50,75,100},
     ymajorgrids=true,
    ]
    

\nextgroupplot   [
ylabel={Visits$/|\Scal|$},
ybar,
ymin=0.0001, ymax=0.2,
ytick={0.0001,0.001,0.01,0.1},
yticklabels={,$10^{-3}$,$10^{-2}$,$10^{-1}$},
ymode = log,
log ticks with fixed point,
log origin = infty,
]

\addplot[color=black,fill=black!90,]
coordinates {(1,0.02465)(2,0.01030)(3,0.01078)(4,0.01347)(5,0.01841)};

\nextgroupplot [
ymin=0.0001, ymax=0.2,
ytick={0.0001,0.001,0.01,0.1},
yticklabels={,,},
ybar,
ymode = log,
log ticks with fixed point,
log origin = infty,
]

\addplot[color=black,fill=black!90,]
coordinates {(1,0.01030)(2,0.01039)(3,0.00995)(4,0.01055)(5,0.01138)};

\nextgroupplot   [
ymin=0.0001, ymax=0.2,
ytick={0.0001,0.001,0.01,0.1},
yticklabels={,,},
ybar,
ymode = log,
log ticks with fixed point,
log origin = infty,
]

\addplot[color=black,fill=black!90,]
coordinates {(1,0.01030)(2,0.00708)(3,0.00448)(4,0.00341)(5,0.00297)};

\nextgroupplot   [
width=2.7cm,
xmin=0.5, xmax=2.5,
ymin=0.0001, ymax=0.2,
ytick={0.0001,0.001,0.01,0.1},
yticklabels={,,},
ybar,
ymode = log,
log ticks with fixed point,
log origin = infty,
]

\addplot[color=black,fill=black!90,]
coordinates {(1,0.00057)(2,0.00019)};

\nextgroupplot  [
ymin=0.0001, ymax=0.2,
ytick={0.0001,0.001,0.01,0.1},
yticklabels={,,},
ybar,
ymode = log,
log ticks with fixed point,
log origin = infty,
]

\addplot[color=black,fill=black!90,]
coordinates {(1,0.01030)(2,0.01292)(3,0.01718)(4,0.02704)(5,0.03235)};


\nextgroupplot   [
ylabel={$\%$ Queries},
ylabel shift = 0.05cm,
]

\addplot[color=black,fill=black!10,]
coordinates {(1,94.8)(2,68.8)(3,46.6)(4,33.3)(5,24.1)};

\addplot[color=black,fill=black!40,]
coordinates {(1,3.9)(2,14.7)(3,20.9)(4,23.1)(5,21.7)};

\addplot[color=black,fill=black!70,]
coordinates {(1,1.3)(2,16.5)(3,32.5)(4,43.6)(5,54.2)};

\nextgroupplot [
legend style={/tikz/every even column/.append style={column sep=0.2cm}},
legend pos=south east,
legend columns = 3,
yticklabels={,,},
]

\addplot[color=black,fill=black!10,]
coordinates {(1,68.8)(2,55.6)(3,53.4)(4,50.3)(5,44.9)};

\addplot[color=black,fill=black!40,]
coordinates {(1,14.7)(2,21)(3,24.8)(4,24.7)(5,26.2)};

\addplot[color=black,fill=black!70,]
coordinates {(1,16.5)(2,23.4)(3,21.8)(4,25)(5,28.9)};

\legend{0}
\addlegendentry{1}
\addlegendentry{2+}

\nextgroupplot   [
yticklabels={,,},
]

\addplot[color=black,fill=black!10,]
coordinates {(1,68.8)(2,62.6)(3,65.3)(4,63.7)(5,65.8)};

\addplot[color=black,fill=black!40,]
coordinates {(1,14.7)(2,18.1)(3,17.2)(4,15.9)(5,16.4)};

\addplot[color=black,fill=black!70,]
coordinates {(1,16.5)(2,19.3)(3,17.5)(4,20.4)(5,17.8)};

\nextgroupplot   [
width=2.7cm,
xmin=0.5, xmax=2.5,
yticklabels={,,},
]

\addplot[color=black,fill=black!10,]
coordinates {(1,66.1)(2,47.9)};

\addplot[color=black,fill=black!40,]
coordinates {(1,15.6)(2,16.2)};

\addplot[color=black,fill=black!70,]
coordinates {(1,18.3)(2,35.9)};

\nextgroupplot  [
yticklabels={,,},
]

\addplot[color=black,fill=black!10,]
coordinates {(1,68.8)(2,67.6)(3,67.8)(4,65)(5,61.5)};

\addplot[color=black,fill=black!40,]
coordinates {(1,14.7)(2,18.7)(3,18.1)(4,18)(5,23.3)};

\addplot[color=black,fill=black!70,]
coordinates {(1,16.5)(2,13.7)(3,14.1)(4,17)(5,15.2)};


\nextgroupplot[
sharp plot,
stack plots = false,
xlabel={Cluster Size},
xticklabels={1,10,25,50,100},
ymode = log,
ymin=0.01, ymax=4,
ytick={0.01,0.1,1,10},
yticklabels={,$10^{-1}$,$10^{0}$},
ylabel={Dist. Calls},
log ticks with fixed point,
log origin = infty,
]

\addplot[
color=black,
mark=square,
]
coordinates {(1,0.065)(2,0.502)(3,0.968)(4,1.326)(5,1.653)};

\addplot[
color=black,
mark=triangle,
]
coordinates {(1,0.053)(2,0.287)(3,0.619)(4,0.956)(5,1.238)};

\nextgroupplot [
sharp plot,
stack plots = false,
legend style={/tikz/every even column/.append style={column sep=0.3cm}},
legend pos=south east,
legend columns = 2,
xlabel={Trajectory Size $n$},
xticklabels={15,25,35,45,55},
ymin=0.01, ymax=4,
ytick={0.01,0.1,1,10},
yticklabels={,,},
ymode = log,
log ticks with fixed point,
log origin = infty,
]

\addplot[
color=black,
mark=square,
]
coordinates {(1,0.502)(2,0.726)(3,0.716)(4,0.796)(5,0.881)};

\addplot[
color=black,
mark=triangle,
]
coordinates {(1,0.287)(2,0.383)(3,0.346)(4,0.373)(5,0.405)};

\legend{$|\delta_F|$}
\addlegendentry{$|\delta_{F\!D}|$}

\nextgroupplot[
sharp plot,
stack plots = false,
xlabel={$|\Scal|$~~[$10^3$]},
xticklabels={5,10,20,30,40},
xlabel shift = -0.1cm,
ymin=0.01, ymax=4,
ytick={0.01,0.1,1,10},
yticklabels={,,},
ymode = log,
log ticks with fixed point,
log origin = infty,
]

\addplot[
color=black,
mark=square,
]
coordinates {(1,0.502)(2,0.605)(3,0.560)(4,0.604)(5,0.558)};

\addplot[
color=black,
mark=triangle,
]
coordinates {(1,0.287)(2,0.345)(3,0.332)(4,0.348)(5,0.301)};

\nextgroupplot[
sharp plot,
stack plots = false,
width=2.7cm,
xmin=0.5, xmax=2.5,
xlabel={$|\Scal|$~[$10^6$]},
xticklabels={1,10},
xlabel shift = -0.1cm,
ymin=0.01, ymax=4,
ytick={0.01,0.1,1,10},
yticklabels={,,},
ymode = log,
log ticks with fixed point,
log origin = infty,
]

\addplot[
color=black,
mark=square,
]
coordinates {(1,0.560)(2,1.043)};

\addplot[
color=black,
mark=triangle,
]
coordinates {(1,0.385)(2,0.870)};

\nextgroupplot [
sharp plot,
stack plots = false,
xlabel={Dimension $d$},
xticklabels={2,4,8,16,32},
ymin=0.01, ymax=4,
ytick={0.01,0.1,1,10},
yticklabels={,,},
ymode = log,
log ticks with fixed point,
log origin = infty,
]

\addplot[
color=black,
mark=square,
]
coordinates {(1,0.502)(2,0.488)(3,0.493)(4,0.56)(5,0.563)};

\addplot[
color=black,
mark=triangle,
]
coordinates {(1,0.287)(2,0.21)(3,0.197)(4,0.223)(5,0.194)};

\end{groupplot}
\end{tikzpicture}
\vspace{-.4cm}
\caption{Effectiveness of exact $\NN$ queries on synthetic data set Relaxed $\CCT$s, averaged over 1000 queries (c.f. Section~\ref{sssec:primary_results}). 
	{\normalfont
The top row shows average number of tree node visits (normalized to a factor of $|\Scal|$).
The middle row shows the percentage of queries that performed $0$, $1$, or more than $1$ distance computation. The bottom row shows the \emph{absolute number} (not normalized) of $\delta_{F}$ and $\delta_{F\!D}$ calls.}
}
\label{fig:NN_syn}
\end{figure}
\begin{figure}[]

\begin{tikzpicture}    

\begin{groupplot}[
     group style = {group size = 1 by 3,
                    horizontal sep=0.15cm,
                    vertical sep=0.15cm,},
     width=14.0cm,
     height=3.4cm,
     grid style=dashed,
     xmode = normal,
     xmin=0, xmax=17,
     xtick style={draw=none},
     xtick=data,
     xticklabels={,,},
     ybar stacked,
     ymin=0, ymax=105,
     ytick={25,50,75,100},
     ymajorgrids=true,
    ]


\nextgroupplot[
ymode = log,
log ticks with fixed point,
log origin = infty,
ybar,
ymin=0.003, ymax=1.5,
ylabel={Visits$/|\Scal|$},
ylabel shift = 0.05cm,
ytick={0.01,0.1,1},
yticklabels={$10^{-2}$,$10^{-1}$,$10^{0}$},
]

\addplot[color=black,fill=black!90,]
coordinates {
	(9,0.154)(4,0.217)(5,0.126)(1,0.150)(7,0.234)(13,0.014)(16,0.012)(8,0.078)(2,0.127)(3,0.202)(6,0.197)(10,0.072)
	(15,0.027)(14,0.013)(11,0.025)(12,0.161)
};


\nextgroupplot   [
ylabel={$\%$ Queries},
ylabel shift = 0.05cm,
legend style={/tikz/every even column/.append style={column sep=0.2cm}},
legend pos=south east,
legend columns = 3,
]

\addplot[color=black,fill=black!10,]
coordinates {(9,99.3)(4,97.7)(5,100.0)(1,97.3)(7,70.9)(13,99.8)(16,98.7)(8,90.6)(2,96.3)(3,100.0)(6,99.3)(10,100.0)(15,64.9)(14,84.4)(11,99.8)(12,99.6)};

\addplot[color=black,fill=black!40,]
coordinates {(9,0.3)(4,1.5)(5,0.0)(1,0.8)(7,2.8)(13,0.2)(16,1.0)(8,5.8)(2,2.3)(3,0.0)(6,0.3)(10,0.0)(15,4.9)(14,5.6)(11,0.2)(12,0.2)};

\addplot[color=black,fill=black!70,]
coordinates {(9,0.4)(4,0.8)(5,0.0)(1,1.9)(7,26.3)(13,0.0)(16,0.3)(8,3.6)(2,1.4)(3,0.0)(6,0.4)(10,0.0)(15,30.2)(14,10.0)(11,0.0)(12,0.2)};

\legend{0}
\addlegendentry{1}
\addlegendentry{2+}

\nextgroupplot[
sharp plot,
stack plots = false,
legend style={/tikz/every even column/.append style={column sep=0.3cm}},
legend pos=north west,
legend columns = 2,
xticklabels={Vessel-M, Pigeon, Seabird, Bus, Cats, Buffalo, Vessel-Y, Gulls, Truck, Bats, Hurdat2, Pen,Football, Geolife, Basketball, Taxi},
xticklabel style={rotate = 50, anchor=north east,inner sep=0.1cm,}, 
ymode = log,
ymin=0.0005, ymax=4,
ytick={0.001,0.01,0.1,1,10},
yticklabels={$10^{-3}$,$10^{-2}$,$10^{-1}$,$10^{0}$},
ylabel={Dist. Calls},
log ticks with fixed point,
log origin = infty,
]

\addplot[
color=black,
mark=square,
]
coordinates {
    (1,0.053)(2,0.053)(3,0.001)(4,0.032)(5,0.001)(6,0.011)(7,0.973)(8,0.133)(9,0.011)(10,0.001)(11,0.002)(12,0.006)(13,0.002)(14,0.343)(15,0.989)(16,0.016)
};

\addplot[
color=black,
mark=triangle,
]
coordinates {
    (1,0.026)(2,0.023)(3,0.001)(4,0.017)(5,0.001)(6,0.004)(7,0.696)(8,0.059)(9,0.004)(10,0.001)(11,0.001)(12,0.003)(13,0.001)(14,0.239)(15,0.762)(16,0.034)
};

\legend{$|\delta_F|$}
\addlegendentry{$|\delta_{F\!D}|$}

\end{groupplot}

\end{tikzpicture}
\vspace{-0.2cm}
\caption{Effectiveness of exact $\NN$ queries on real data set Relaxed $\CCT$s, averaged over 1000 queries (c.f. Section~\ref{sssec:primary_results}). 
	{\normalfont
The top row shows average number of tree node visits (normalized to a factor of $|\Scal|$).
The middle row shows the percentage of queries that performed $0$, $1$, or more than $1$ distance computation. The bottom row shows the \emph{absolute number} (not normalized) of $\delta_{F}$ and $\delta_{F\!D}$ calls.}
}

\label{fig:NN_real}
\end{figure}
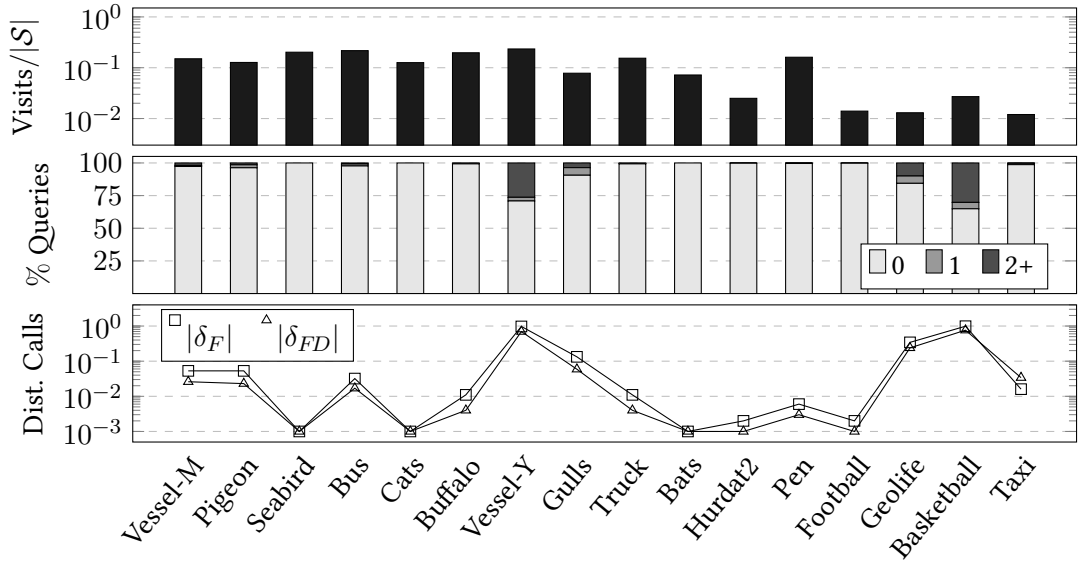

Figures~\ref{fig:NN_syn}~and~\ref{fig:NN_real} show the effectiveness of exact $\NN$ queries on Relaxed $\CCT$s for synthetic and real data sets, respectively.
On most data sets, the average number of expensive $\delta_{F}$ distance calls per query is \emph{one or fewer}, and only increases slightly for highly clustered data sets.
Surprisingly, the majority of queries require no distance computations at all for many of the data sets.
The $10$M trajectory data set performs on average only $1.04$ expensive $\delta_{F}$ calls per query.
Interestingly, the Vessel-Y~\cite{vessel05} data set requires a similar average of $0.97$ $\delta_{F}$ calls, even though it is a much smaller data set.
The Vessel-Y data set has higher intrinsic dimensionality, so this shows that clustering of data has a much larger influence on distance calls than the number of trajectories does.
The number of node visits (normalized to a factor of $|\Scal|$) decreases as the number of trajectories increases, showing effective pruning of the search space.

\begin{figure}[h]
\vspace{-0.05cm}

\begin{tikzpicture}

\begin{groupplot}[
     group style = {group size = 5 by 2,
                    horizontal sep=0.15cm,
                    vertical sep=0.15cm,},
     height=3.4cm,
     width=4.3cm,
     grid style=dashed,
     xmode = normal,
     xtick=data,
     xtick style={draw=none},
     xticklabels={,,},
     xmin=0.5, xmax=5.5,
     ybar,
     ymin=0, ymax=105,
     ytick={25,50,75,100},
     ymajorgrids=true,
     ymode = log,
     log ticks with fixed point,
     log origin = infty,
    ]
    

\nextgroupplot   [
ylabel={Visits$/|\Scal|$},
ymin=0.0001, ymax=0.2,
ytick={0.0001,0.001,0.01,0.1},
yticklabels={,$10^{-3}$,$10^{-2}$,$10^{-1}$},
]

\addplot[color=black,fill=black!90,]
coordinates {(1,0.03698)(2,0.01258)(3,0.01165)(4,0.01441)(5,0.01971)};

\nextgroupplot [
ymin=0.0001, ymax=0.2,
ytick={0.0001,0.001,0.01,0.1},
yticklabels={,,},
]

\addplot[color=black,fill=black!90,]
coordinates {(1,0.01258)(2,0.01226)(3,0.01155)(4,0.01283)(5,0.01341)};

\nextgroupplot   [
ymin=0.0001, ymax=0.2,
ytick={0.0001,0.001,0.01,0.1},
yticklabels={,,},
]

\addplot[color=black,fill=black!90,]
coordinates {(1,0.01258)(2,0.00848)(3,0.00551)(4,0.00425)(5,0.00372)};

\nextgroupplot   [
width=2.7cm,
xmin=0.5, xmax=2.5,
ymin=0.0001, ymax=0.2,
ytick={0.0001,0.001,0.01,0.1},
yticklabels={,,},
]

\addplot[color=black,fill=black!90,]
coordinates {(1,0.00077)(2,0.000257)};

\nextgroupplot  [
ymin=0.0001, ymax=0.2,
ytick={0.0001,0.001,0.01,0.1},
yticklabels={,,},
]

\addplot[color=black,fill=black!90,]
coordinates {(1,0.01258)(2,0.01547)(3,0.02180)(4,0.03613)(5,0.03913)};


\nextgroupplot   [
sharp plot,
stack plots = false,
xlabel={Cluster Size},
xticklabels={1,10,25,50,100},
ylabel={Dist. Calls},
ylabel shift = 0.15cm,
ymin=0.1, ymax=30,
ytick={0.1,1,10,100},
yticklabels={,$10^{0}$,$10^{1}$},
]

\addplot[
color=black,
mark=square,
]
coordinates {(1,0.349)(2,0.623)(3,2.004)(4,3.056)(5,4.086)};

\addplot[
color=black,
mark=triangle,
]
coordinates {(1,0.430)(2,0.874)(3,4.830)(4,10.245)(5,18.458)};

\nextgroupplot [
sharp plot,
stack plots = false,
legend style={/tikz/every even column/.append style={column sep=0.1cm}},
legend pos=north west,
legend columns = 2,
xlabel={Trajectory Size $n$},
xticklabels={15,25,35,45,55},
ymin=0.1, ymax=30,
ytick={0.1,1,10,100},
yticklabels={,,},
]

\addplot[
color=black,
mark=square,
]
coordinates {(1,0.623)(2,0.869)(3,0.997)(4,0.979)(5,1.072)};

\addplot[
color=black,
mark=triangle,
]
coordinates {(1,0.874)(2,1.366)(3,1.470)(4,1.473)(5,1.576)};

\legend{$|\delta_F|$}
\addlegendentry{$|\delta_{F\!D}|$}

\nextgroupplot   [
sharp plot,
stack plots = false,
xlabel={$|\Scal|$~~[$10^3$]},
xticklabels={5,10,20,30,40},
ymin=0.1, ymax=30,
ytick={0.1,1,10,100},
yticklabels={,,},
]

\addplot[
color=black,
mark=square,
]
coordinates {(1,0.623)(2,0.798)(3,0.896)(4,0.909)(5,0.839)};

\addplot[
color=black,
mark=triangle,
]
coordinates {(1,0.874)(2,1.158)(3,1.351)(4,1.463)(5,1.299)};

\nextgroupplot   [
sharp plot,
stack plots = false,
xlabel={$|\Scal|$~[$10^6$]},
xticklabels={1,10},
xlabel shift = -0.1cm,
width=2.7cm,
xmin=0.5, xmax=2.5,
ymin=0.1, ymax=30,
ytick={0.1,1,10,100},
yticklabels={,,},
]

\addplot[
color=black,
mark=square,
]
coordinates {(1,1.749)(2,3.497)};

\addplot[
color=black,
mark=triangle,
]
coordinates {(1,4.095)(2,13.253)};

\nextgroupplot  [
sharp plot,
stack plots = false,
xlabel={Dimension $d$},
xticklabels={2,4,8,16,32},
ymin=0.1, ymax=30,
ytick={0.1,1,10,100},
yticklabels={,,},
]

\addplot[
color=black,
mark=square,
]
coordinates {(1,0.623)(2,0.754)(3,0.896)(4,1.031)(5,1.076)};

\addplot[
color=black,
mark=triangle,
]
coordinates {(1,0.874)(2,0.926)(3,1.105)(4,1.346)(5,1.453)};

\end{groupplot}
\end{tikzpicture}
\vspace{-.4cm}
\caption{Effectiveness of exact $\kNN$ queries ($k=5$) on synthetic data set Relaxed $\CCT$s, averaged over 1000 queries (c.f. Section~\ref{sssec:primary_results}). 
	{\normalfont
The top row shows average number of tree node visits (normalized to a factor of $|\Scal|$).
The bottom row shows the \emph{absolute number} (not normalized) of $\delta_{F\!D}$ and $\delta_{F}$ calls.}
}
\label{fig:kNN_syn}
\end{figure}
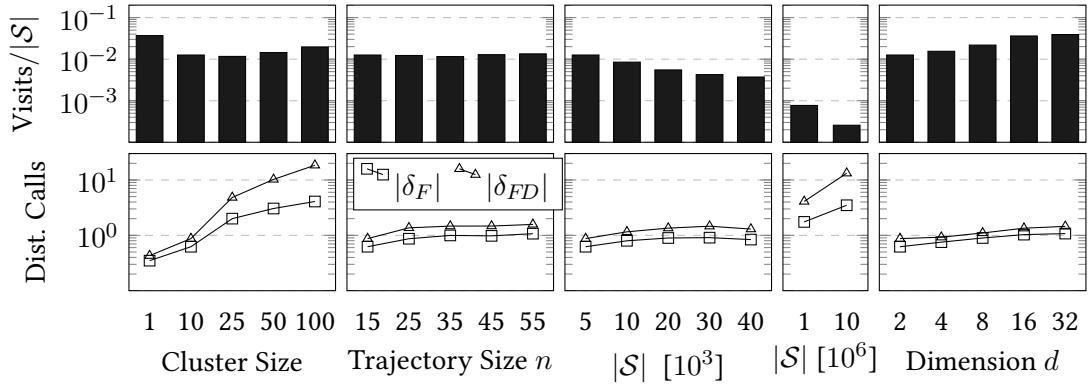

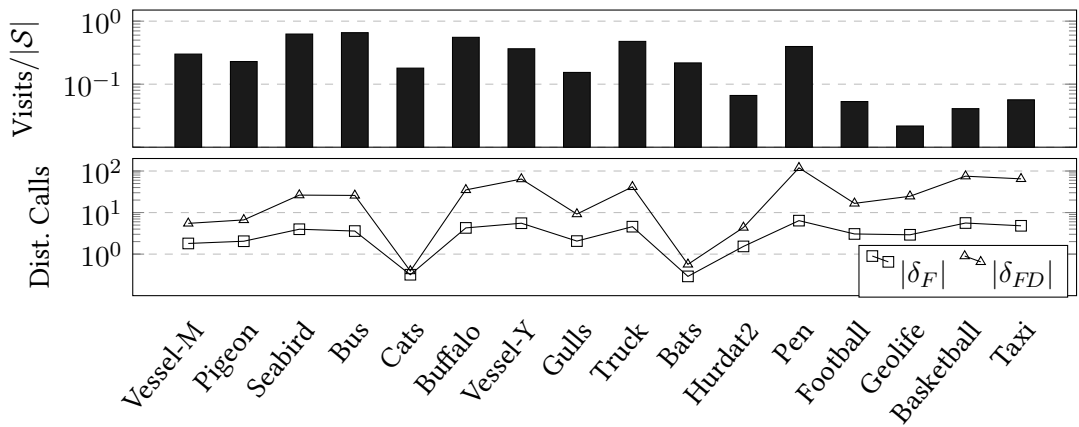
\begin{figure}[h]

\begin{tikzpicture}    

\begin{groupplot}[
     group style = {group size = 1 by 2,
                    horizontal sep=0.15cm,
                    vertical sep=0.15cm,},
     width=14.0cm,
     height=3.4cm,
     grid style=dashed,
     xmode = normal,
     xmin=0, xmax=17,
     xtick style={draw=none},
     xtick=data,
     xticklabels={,,},
     ybar,
     ymin=0, ymax=105,
     ytick={25,50,75,100},
     ymajorgrids=true,
    ]


\nextgroupplot[
ybar,
ymin=0.01, ymax=1.5,
ylabel={Visits$/|\Scal|$},
ylabel shift = 0.05cm,
ytick={0.01,0.1,1},
yticklabels={,$10^{-1}$,$10^{0}$},
ymode = log,
log ticks with fixed point,
log origin = infty,
]

\addplot[color=black,fill=black!90,]
coordinates {
	(9,0.4773)(4,0.6565)(5,0.1798)(1,0.3002)(7,0.3646)(13,0.0530)(16,0.0565)(8,0.1536)(2,0.2288)(3,0.6244)(6,0.5538)(10,0.2174)
	(15,0.0409)(14,0.0217)(11,0.0662)(12,0.3960)
};


\nextgroupplot[
sharp plot,
stack plots = false,
legend style={/tikz/every even column/.append style={column sep=0.1cm}},
legend style={at={(0.88,0.37)},anchor=north},
legend columns = 2,
xticklabels={Vessel-M, Pigeon, Seabird, Bus, Cats, Buffalo, Vessel-Y, Gulls, Truck, Bats, Hurdat2, Pen,Football, Geolife, Basketball, Taxi},
xticklabel style={rotate = 50, anchor=north east,inner sep=0.1cm,}, 
ylabel={Dist. Calls},
ylabel shift = 0.15cm,
ybar,
ymin=0.1, ymax=200,
ytick={1,10,100},
yticklabels={$10^{0}$,$10^{1}$,$10^{2}$},
ymode = log,
log ticks with fixed point,
log origin = infty,
]

\addplot[
color=black,
mark=square,
]
coordinates {(1,1.805)(2,2.042)(3,3.971)(4,3.589)(5,0.319)(6,4.273)(7,5.542)(8,2.059)(9,4.574)(10,0.291)(11,1.532)(12,6.391)(13,3.055)(14,2.919)(15,5.601)(16,4.802)};

\addplot[
color=black,
mark=triangle,
]
coordinates {(1,5.492)(2,6.655)(3,26.316)(4,25.856)(5,0.392)(6,35.114)(7,64.076)(8,9.215)(9,41.561)(10,0.567)(11,4.375)(12,119.420)(13,16.640)(14,24.692)(15,75.150)(16,64.782)};

\legend{$|\delta_F|$}
\addlegendentry{$|\delta_{F\!D}|$}

\end{groupplot}

\end{tikzpicture}
\vspace{-0.2cm}
\caption{Effectiveness of exact $\kNN$ queries ($k=5$) on real data set Relaxed $\CCT$s, averaged over 1000 queries (c.f. Section~\ref{sssec:primary_results}). 
	{\normalfont
The top row shows average number of tree node visits (normalized to a factor of $|\Scal|$).
The bottom row shows the \emph{absolute number} (not normalized) of $\delta_{F\!D}$ and $\delta_{F}$ calls.}
}

\label{fig:kNN_real}
\end{figure}

Figures~\ref{fig:kNN_syn}~and~\ref{fig:kNN_real} show the effectiveness of exact $\kNN$ queries on Relaxed $\CCT$s for synthetic and real data sets, respectively.
The results correspond to the $\NN$ query results above.

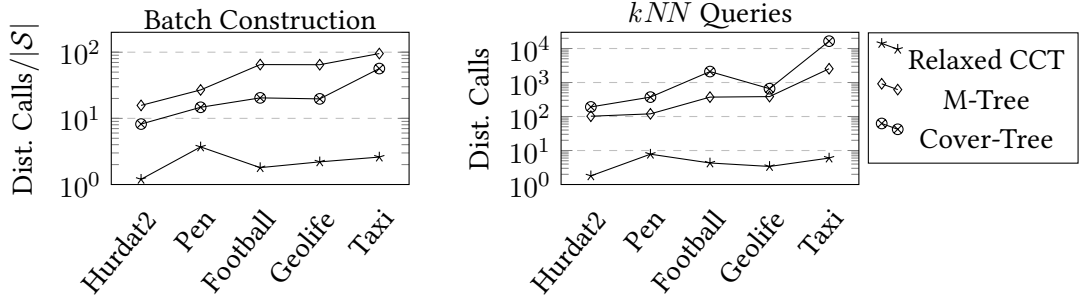
\begin{figure}[h]

\begin{tikzpicture}    
\pgfplotsset{every axis title/.append style={at={(0.5,0.85)}}}
\begin{groupplot}[
     group style = {group size = 2 by 1,
                    horizontal sep=2.0cm,
                    vertical sep=0.15cm,},
     width=5.5cm,
     height=3.6cm,
     grid style=dashed,
     xmode = normal,
     xmin=0.5, xmax=5.5,
     xtick style={draw=none},
     xtick=data,
     xticklabels={,,},
     ybar,
     ymajorgrids=true,
     ymode = log,
     log ticks with fixed point,
     log origin = infty,
    ]


\nextgroupplot[
sharp plot,
stack plots = false,
xticklabels={Hurdat2, Pen, Football, Geolife, Taxi},
xticklabel style={rotate = 50, anchor=north east,inner sep=0.1cm,}, 
ymin=1, ymax=200,
ytick={1,10,100,1000},
yticklabels={$10^{0}$,$10^{1}$,$10^{2}$},
ylabel={Dist. Calls$/|\Scal|$},
title={Batch Construction},
]

\addplot[
color=black,
mark=star,
]
coordinates {(1,1.2)(2,3.7)(3,1.8)(4,2.2)(5,2.6)};

\addplot[
color=black,
mark=diamond,
]
coordinates {(1,15.7)(2,26.9)(3,65.0)(4,64.7)(5,95.4)};

\addplot[
color=black,
mark=otimes,
]
coordinates {(1,8.2)(2,14.7)(3,20.4)(4,19.7)(5,56.6)};

\nextgroupplot[
sharp plot,
stack plots = false,
legend pos=outer north east,
xticklabels={Hurdat2, Pen, Football, Geolife, Taxi},
xticklabel style={rotate = 50, anchor=north east,inner sep=0.1cm}, 
ymin=1, ymax=30000,
ytick={1,10,100,1000,10000,100000},
yticklabels={$10^{0}$,$10^{1}$,$10^{2}$,$10^{3}$,$10^{4}$},
ylabel={Dist. Calls},
title={$\kNN$ Queries},
]

\addplot[
color=black,
mark=star,
]
coordinates {(1,1.8)(2,7.8)(3,4.3)(4,3.4)(5,6.0)};

\addplot[
color=black,
mark=diamond,
]
coordinates {(1,102.0)(2,119.7)(3,372.1)(4,385.4)(5,2540.0)};

\addplot[
color=black,
mark=otimes,
]
coordinates {(1,191.4)(2,372.0)(3,2098.7)(4,659.1)(5,16449.7)};

\legend{Relaxed $\CCT$}
\addlegendentry{M-Tree}
\addlegendentry{Cover-Tree}

\end{groupplot}

\end{tikzpicture}
\vspace{-0.2cm}
\caption{Performance of Relaxed $\CCT$s vs. standard M-Tree~\cite{ciaccia1997} and Cover-Tree~\cite{BeygelzimerKL06} implementations for the five largest real $d=2$ data sets (c.f. Section~\ref{sssec:primary_results}).
	{\normalfont
	    The left chart shows construction $\delta_{F}$ calls normalized over the data set size $|\Scal|$. The right chart shows exact $\kNN$ query ($k=10$) \emph{absolute number} (not normalized) $\delta_{F}$ calls, averaged over 1000 queries (query method two).}
}

\label{fig:cct_mtree_covertree}
\end{figure}

The experimental results for comparison of the Relaxed $\CCT$ vs. standard, `off-the-shelf' metric indexing methods M-Tree~\cite{ciaccia1997} and Cover-Tree~\cite{BeygelzimerKL06} are in Figure~\ref{fig:cct_mtree_covertree}.
The Fréchet distance function is 'plugged' into the generic Cover-Tree, whose implementation uses a 'scaling' constant of $1.3$ which results in $1/1.3\approx0.78$ for compactness and separation to balance arity and depth.
For the M-Tree, we used the random promote method, as it performs the fewest distance calls during construction, and set the maximum arity to $100$.
We also attempted to improve M-Tree performance by first testing $\delta_{\fd}$, and if it fails then calling $\delta_F$, for both construction and queries. 
The results show that both for construction and query the number of $\delta_F$ calls for the $\CCT$ are usually at least an order of magnitude smaller than required for the standard M-Tree and Cover-Tree.
For example, the $\kNN$ queries on the Taxi~\cite{taxiA11,taxiB10} data set performed $6.0$ $\delta_F$ calls on average using the $\CCT$, and $16.4\times10^3$ calls using the Cover-Tree.

\input{figs/fig-RNN_bringmann_compare.tex}

Figure~\ref{fig:RNN_bringmann_compare} compares the performance of our approach with those of the recent contribution by Bringmann et al.~\cite{bri19} that performs exact $\RNN$ queries under the Fr\'echet distance in $2$ dimensional space, using an $8$ dimensional KD-tree (c.f. Section~\ref{sec:introduction}).
For the KD-Tree based approach, the number of visits is defined as the total number of nodes visited during the tree traversal.  
In the bound invocation metric, four bounds ($\LB_{\fd}$,$\UB_{\adff}$,$\UB_{\adfr}$,$\UB_{\adfd}$) may be counted for the $\CCT$ and only three (adaptive equal-time, negative filter, and greedy) for~\cite{bri19}.
In comparison, the $\RNN$ queries using $\CCT$s have fewer node visits, compute fewer bound computations, and perform fewer Fr\'echet decision calls by an average factor of $3$ for synthetic data sets.
Though our queries may perform up to four bound computations per trajectory, and not just three, it is surprising that $\CCT$s perform fewer total bound computations for all but one of the inputs.
This improvement is due to stronger bounds and clustering of trajectories, which allows the algorithm to test if all trajectories within a cluster belong in the result.

\begin{figure*}[]
\vspace{-0.05cm}


\begin{tikzpicture}

\begin{groupplot}[
     group style = {group size = 4 by 2,
                    horizontal sep=0.1cm,
                    vertical sep=0.1cm,},
     height=3.6cm,
     width=4.5cm,
     grid style=dashed,
     xmode = normal,
     xtick=data,
     xtick style={draw=none},
     xticklabels={,,},
     xmin=0.5, xmax=5.5,
     ybar,
     ymin=0.003, ymax=1.5,
     ytick={0.01,0.1,1},
     ymajorgrids=true,
     ymode = log,
     log ticks with fixed point,
     log origin = infty,
    ]
    

\nextgroupplot [bar width = 0.15cm,
ylabel={Visits$/|\Scal|$},
yticklabels={$10^{-2}$,$10^{-1}$,$10^{0}$},
]
\addplot[color=black,fill=black!10,bar shift = -0.1cm,]
coordinates {(1,0.0247)(2,0.0103)(3,0.0108)(4,0.0135)(5,0.0184)};
\addplot[color=black,fill=black!80,bar shift = 0.1cm,]
coordinates {(1,1)(2,1)(3,1)(4,1)(5,1)};




\nextgroupplot [bar width = 0.15cm,
yticklabels={,,},
]
\addplot[color=black,fill=black!10,bar shift = -0.1cm,]
coordinates {(1,0.0103)(2,0.0071)(3,0.0045)(4,0.0034)(5,0.0030)};
\addplot[color=black,fill=black!80,bar shift = 0.1cm,]
coordinates {(1,1)(2,1)(3,1)(4,1)(5,1)};


\nextgroupplot [bar width = 0.15cm,
yticklabels={,,},
]
\addplot[color=black,fill=black!10,bar shift = -0.1cm,]
coordinates {(1,0.0103)(2,0.0129)(3,0.0172)(4,0.0270)(5,0.0324)};
0.0103
\addplot[color=black,fill=black!80,bar shift = 0.1cm,]
coordinates {(1,1)(2,1)(3,1)(4,1)(5,1)};


\nextgroupplot [bar width = 0.15cm,
yticklabels={,,},
xshift = 0.4cm,
]
\addplot[color=black,fill=black!10,bar shift = -0.1cm,]
coordinates {(1,0.0250)(2,0.1607)(3,0.0145)(4,0.0128)(5,0.0123)};
\addplot[color=black,fill=black!80,bar shift = 0.1cm,]
coordinates {(1,1)(2,1)(3,1)(4,1)(5,1)};


\nextgroupplot [bar width = 0.15cm,
ymode = log,
log ticks with fixed point,
log origin = infty,
xlabel={Cluster Size},
xticklabels={1,10,25,50,100},
ymin=0.0001, ymax=0.4,
ytick={0.0001,0.001,0.01,0.1,1},
yticklabels={,$10^{-3}$,$10^{-2}$,$10^{-1}$},
ylabel={Bounds$/|\Scal|$},
]

\addplot[color=black,fill=black!10,bar shift = -0.1cm,]
coordinates {(1,0.0007)(2,0.0009)(3,0.0019)(4,0.0034)(5,0.0062)};
\addplot[color=black,fill=black!80,bar shift = 0.1cm,]
coordinates {(1,0.0036)(2,0.0095)(3,0.0112)(4,0.0131)(5,0.0205)};


\nextgroupplot [bar width = 0.15cm,
xlabel={$|\Scal|$~~[$10^3$]},
xticklabels={5,10,20,30,40},
xlabel shift = -0.1cm,
ymin=0.0001, ymax=0.4,
ytick={0.0001,0.001,0.01,0.1,1},
yticklabels={,,},
legend pos=north west,
legend columns = 2,
]
\addplot[color=black,fill=black!10,bar shift = -0.1cm,]
coordinates {(1,0.0009)(2,0.0005)(3,0.0003)(4,0.0002)(5,0.0001)};
\addplot[color=black,fill=black!80,bar shift = 0.1cm,]
coordinates {(1,0.0095)(2,0.0050)(3,0.0028)(4,0.0019)(5,0.0016)};
\addlegendentry{$\CCT$}
\addlegendentry{Scan}

\nextgroupplot [bar width = 0.15cm,
xlabel={Dimension $d$},
xticklabels={2,4,8,16,32},
ymin=0.0001, ymax=0.4,
ytick={0.0001,0.001,0.01,0.1,1},
yticklabels={,,},
]
\addplot[color=black,fill=black!10,bar shift = -0.1cm,]
coordinates {(1,0.0009)(2,0.0012)(3,0.0013)(4,0.0013)(5,0.0013)};
\addplot[color=black,fill=black!80,bar shift = 0.1cm,]
coordinates {(1,0.0095)(2,0.0098)(3,0.0086)(4,0.0089)(5,0.0083)};

\nextgroupplot [bar width = 0.15cm,
xticklabels={Hurdat2, Pen, Football, Geolife, Taxi},
xticklabel style={rotate = 50, anchor=north east,inner sep=0.1cm,}, 
ymin=0.0001, ymax=0.4,
ytick={0.0001,0.001,0.01,0.1,1},
yticklabels={,,},
]
\addplot[color=black,fill=black!10,bar shift = -0.1cm,]
coordinates {(1,0.0012)(2,0.0738)(3,0.0005)(4,0.0009)(5,0.0004)};
\addplot[color=black,fill=black!80,bar shift = 0.1cm,]
coordinates {(1,0.0108)(2,0.2302)(3,0.0022)(4,0.0038)(5,0.0011)};

\end{groupplot}

\end{tikzpicture}
\vspace{-.4cm}
\caption{Performance of the Relaxed $\CCT$ Prune stage vs. Improved Linear Scan for exact $\NN$-queries on synthetic (first $3$ columns) and real (last column) data (c.f. Section~\ref{sssec:primary_results}).
	{\normalfont
		Bar charts show average metrics over $1000$ queries for $\CCT$s (light shade) and the improved Scan (dark shade).
		All metrics are normalized over the data set size $|\Scal|$.
		The rows denote number of node visits (top) during the pruning stage, and number of bound computations (bottom) in the reduce stage. }
}
\label{fig:NN_naive_compare}
\end{figure*}
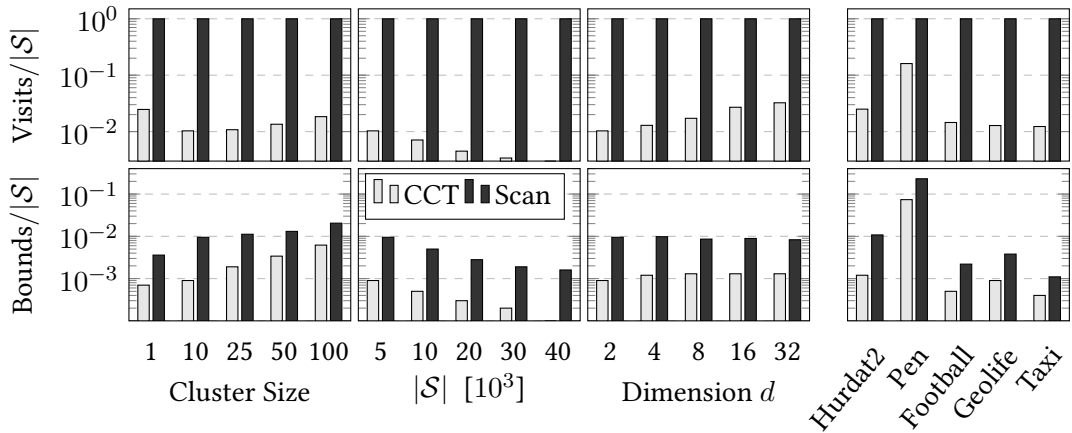

Figure~\ref{fig:NN_naive_compare} compares the Prune stage of our $\NN$ query to the improved $\NN$ linear scan.
Linear scan visits are defined as total trajectories scanned.
With exception of the Pen~\cite{pentip06} data set,
the number of $\CCT$ visits are factors between ten to over one hundred times smaller than the linear scan's,
and the number of $\CCT$ bound computations are ten times smaller than the linear scan's, especially for datasets with a large number of trajectories.
Even in higher dimensions (e.g. $d=32$), the $\CCT$ performs a factor of thirty fewer visits.

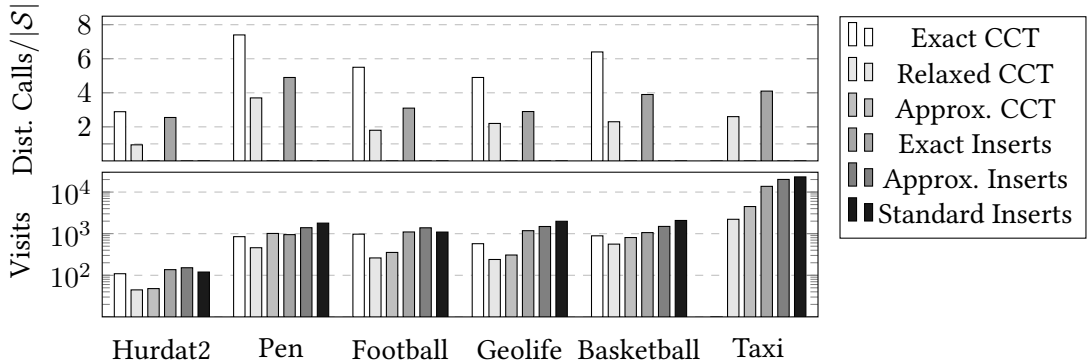
\begin{figure}[]
\begin{tikzpicture}

\begin{groupplot}[
     group style = {group size = 1 by 2,
                    horizontal sep=0.15cm,
                    vertical sep=0.15cm,},
     width=11.0cm,
     height=3.5cm,
     grid style=dashed,
     xmode = normal,
     xtick=data,
     xtick style={draw=none},
     xticklabels={,,},
     xmin=0.5, xmax=6.5,
     ymajorgrids=true,
    ]

\nextgroupplot[
xticklabels={,,},
ybar,
ymin=0, ymax=8.5,
ytick={1,2,4,6,8},
yticklabels={,$2$,$4$,$6$,$8$},
ylabel={Dist. Calls$/|\Scal|$},
ylabel shift = 0.24cm,
bar width = 0.15cm,
legend pos=outer north east,
]

\addplot[color=black,fill=black!0,]
coordinates {(1,2.89)(2,7.4)(3,5.5)(4,4.9)(5,6.4)(6,-1)};

\addplot[color=black,fill=black!10,]
coordinates {(1,0.944)(2,3.7)(3,1.8)(4,2.2)(5,2.3)(6,2.6)};

\addplot[color=black,fill=black!25,]
coordinates {(1,-1)(2,-1)(3,-1)(4,-1)(5,-1)(6,-1)};

\addplot[color=black,fill=black!35,]
coordinates {(1,2.55)(2,4.9)(3,3.1)(4,2.9)(5,3.9)(6,4.1)};

\addplot[color=black,fill=black!50,]
coordinates {(1,-1)(2,-1)(3,-1)(4,-1)(5,-1)(6,-1)};

\addplot[color=black,fill=black!90,]
coordinates {(1,-1)(2,-1)(3,-1)(4,-1)(5,-1)(6,-1)};

\legend{Exact $\CCT$}
\addlegendentry{Relaxed $\CCT$}
\addlegendentry{Approx. $\CCT$}
\addlegendentry{Exact Inserts}
\addlegendentry{Approx. Inserts}
\addlegendentry{Standard Inserts}

\nextgroupplot[
xticklabels={Hurdat2, Pen, Football, Geolife, Basketball, Taxi},
ybar,
ymin=10, ymax=30000,
ytick={10,100,1000,10000,100000},
yticklabels={,$10^2$,$10^3$,$10^4$},
ylabel={Visits},
ymode = log,
log ticks with fixed point,
log origin = infty,
bar width = 0.15cm,
]

\addplot[color=black,fill=black!0,]
coordinates {(1,109.1)(2,847.8)(3,973.8)(4,574.5)(5,888.0)(6,1)};

\addplot[color=black,fill=black!10,]
coordinates {(1,44.6)(2,459.4)(3,261.3)(4,239.8)(5,561.1)(6,2217.7)};

\addplot[color=black,fill=black!25,]
coordinates {(1,47.9)(2,1012.6)(3,354.2)(4,307.1)(5,807.5)(6,4487.6)};

\addplot[color=black,fill=black!35,]
coordinates {(1,136.6)(2,948.5)(3,1099.2)(4,1180.2)(5,1062.1)(6,13727.9)};

\addplot[color=black,fill=black!50,]
coordinates {(1,151.8)(2,1392.2)(3,1383.9)(4,1487.6)(5,1494.4)(6,20115.8)};

\addplot[color=black,fill=black!90,]
coordinates {(1,119.7)(2,1794.1)(3,1093.7)(4,1987.3)(5,2077.6)(6,23145.1)};

\end{groupplot}
\end{tikzpicture}
\caption{Effectiveness of $\CCT$ constructions and index query performance on the six largest real data sets (c.f. Section~\ref{sssec:primary_results}). 
	{\normalfont
The top row shows $\delta_{F}$ calls of batch constructions and dynamic insertions, normalized over the data set size $|\Scal|$. 
The bottom row shows tree node visits of exact $\NN$ queries, averaged over 1000 queries.
For the Taxi~\cite{taxiA11,taxiB10} data set the Exact $\CCT$ batch construction did not finish within $3$ days and is omitted.
}
}
\label{fig:CCT_const_compare}
\end{figure}

Figure~\ref{fig:CCT_const_compare} results show that the number of $\delta_F$ calls for the six types of $\CCT$ constructions, and corresponding node visits for $\NN$ queries.
For $\CCT$ construction methods that perform $\delta_F$ calls, the Relaxed $\CCT$ performs the fewest, even sub-linear on Hurdat2, hence significantly fewer than $\mathcal{O}(|\Scal|^2)$.
Note that the Exact $\CCT$ batch construction for the Taxi data set did not complete in a reasonable time due to the quadratic nature of the algorithm. We attempted to speed-up the Exact $\CCT$ batch construction algorithm by quickly eliminating trajectories outside of a 'neighborhood', but this improvement became less effective as $|\Scal|$ grew. The Relaxed $\CCT$ does not have this issue, and also shows the best query performance.

The node visits for all $\CCT$ constructions correlate with the overlap quality measure (see Section~\ref{ssec:cct_analysis}, Figure~\ref{fig:CCT_quality}).
The Relaxed $\CCT$ performs the fewest $\NN$ node visits at query time.
Interestingly, the Approximate $\CCT$ has relatively good query performance, and can be useful in practice since its construction is faster than the Relaxed $\CCT$ since no $\delta_F$ calls are performed.
The insert algorithms typically result in more query node accesses compared with batch constructions.
The standard insert algorithm usually performs the worst at query time, especially if the data set has higher intrinsic dimensionality.

\subsubsection{Supplementary Results} \label{sssec:supplementary_results}

Figure~\ref{fig:kNN_query_short_narrow} shows the gain in effectiveness from approximate over exact $\kNN$  queries, with $k=5$ and $\Erel=0.5$, on our real-world data sets.
For the majority of the approximate queries, the number of $\delta_{F}$ and $\delta_{F\!D}$ calls are a factor of two or more smaller than those of exact queries.
For the Pen~\cite{pentip06} data set, the number of distance calls in an approximate query decreases by a factor of \emph{forty}, suggesting that small approximation factors can result in significant performance gains.

Our new and improved bounds in Section~\ref{sec:Bounds} result in better query performance, as shown in Figure~\ref{fig:bound_effectiveness}.
For example, without the bound enhancements (using only previously existing bounds), the $\RNN$ queries perform a factor of $4.7$ more $\delta_{F\!D}$ calls on average for the five largest $d=2$ real data sets.

Figure~\ref{fig:implicit_error} shows that implicit approximate queries return, on average, results with small $\Erel$ errors.
All real data sets show $\Erel < 0.5$ for $NN$ queries, and $\Erel < 1.8$ for $\kNN$ queries.
Lower intrinsic dimensionality correlates with smaller $\Erel$, and vice versa.

In Section~\ref{sssec:nn_algorithm} we state that our optimized $\NN$ algorithm can outperform the $\kNN$ when $k=1$, and results in Figure~\ref{fig:NN_vs_kNN} provide evidence for the claim.
For example, the $\NN$ query on the Basketball~\cite{NBA16} data set performs a factor of two fewer $\delta_{F}$ calls and a factor of ten fewer $\delta_{F\!D}$ calls.

\label{fig:implicit_error}

\begin{figure}[p]
\begin{tikzpicture}

\begin{groupplot}[
     group style = {group size = 1 by 2,
                    horizontal sep=0.15cm,
                    vertical sep=0.15cm,},
     width=14.0cm,
     height=3.5cm,
     grid style=dashed,
     xmode = normal,
     xtick=data,
     xtick style={draw=none},
     xticklabels={,,},
     xmin=0, xmax=17,
     ymajorgrids=true,
    ]

\nextgroupplot[
legend pos=north west,
legend style={/tikz/every even column/.append style={column sep=0.1cm}},
legend columns = 2,
xticklabels={,,},
ybar stacked,
ymin=0, ymax=7.5,
ylabel={Nof. $\delta_F$},
ytick={2,4,6},
ylabel shift = 0.05cm,
]
\addplot[color=black,fill=black!10,]
coordinates {
	(9,1.873)(4,0.442)(5,0.000)(1,1.042)(7,3.561)(13,0.178)(16,0.508)(8,0.422)(2,0.871)(3,1.316)(6,0.988)(10,0.000)
	(15,0.775)(14,0.713)(11,0.130)(12,0.146)
};
\addplot[color=black,fill=black!70,]
coordinates {
	(9,2.701)(4,3.147)(5,0.319)(1,0.763)(7,1.981)(13,2.877)(16,4.294)(8,1.637)(2,1.171)(3,2.655)(6,3.285)(10,0.291)
	(15,4.826)(14,2.206)(11,1.402)(12,6.245)
};
\legend{$\Erel := 0.5$}
\addlegendentry{$\Erel := 0$}

\nextgroupplot[
xticklabels={Vessel-M, Pigeon, Seabird, Bus, Cats, Buffalo, Vessel-Y, Gulls, Truck, Bats, Hurdat2, Pen,Football, Geolife, Basketball, Taxi},
xticklabel style={rotate = 60, anchor=north east,inner sep=0pt,},
ybar stacked,
ymin=0.1, ymax=200,
ytick={1,10,100},
yticklabels={$10^0$,$10^1$,$10^2$},
ylabel={Nof. $\delta_{F\!D}$},
ylabel shift = -0.2cm,
ymode = log,
log ticks with fixed point,
log origin = infty,
]
\addplot[color=black,fill=black!10,]
coordinates {
	(1,3.169)(2,3.189)(3,8.143)(4,1.397)(5,0.001)(6,4.8541)(7,32.517)(8,1.964)(9,13.881)(10,0.001)(11,0.364)(12,1.109)(13,0.584)(14,5.706)(15,8.414)(16,4.075)
};

\addplot[color=black,fill=black!70,]
coordinates {
	(1,2.323)(2,3.466)(3,18.173)(4,24.459)(5,0.392)(6,30.260)(7,31.559)(8,7.251)(9,27.680)(10,0.567)(11,118.311)(12,4.011)(13,16.056)(14,18.986)(15,66.736)(16,60.707)
};

\end{groupplot}
\end{tikzpicture}
\caption{Effectiveness of exact vs. approximate ($\varepsilon^*=0.5$) $\kNN$ queries on real data sets, for $k=5$ (c.f. Section~\ref{sssec:supplementary_results}). 
	{\normalfont
Bar charts show average absolute values over $1000$ queries for approximate (dark gray) and exact (light grey) queries, on Relaxed $\CCT$s. The rows denote the number of distance $\delta_{F}$ (top) and Fr\'echet decision procedure $\delta_{F\!D}$ (bottom) computations during the Decide stage.
}
}
\label{fig:kNN_query_short_narrow}
\end{figure}
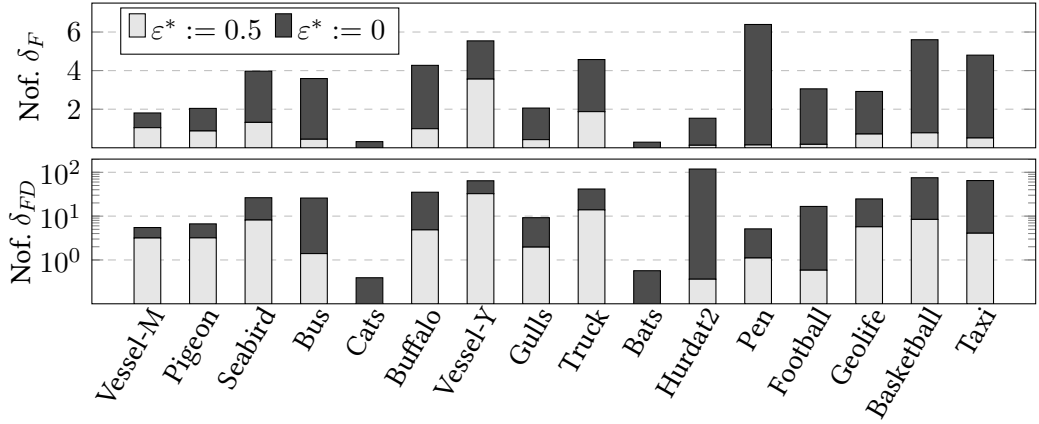

\begin{figure}[p]
\begin{tikzpicture}

\begin{groupplot}[
     group style = {group size = 2 by 2,
                    horizontal sep=0.90cm,
                    vertical sep=0.15cm,},
     width=10.2cm,
     height=3.5cm,
     grid style=dashed,
     xmode = normal,
     xtick=data,
     xtick style={draw=none},
     xticklabels={,,},
     xmin=0.5, xmax=16.5,
     ymajorgrids=true,
    ]

\nextgroupplot[
legend style={/tikz/every even column/.append style={column sep=0.1cm}},
legend columns = 2,
legend pos=south west,
xticklabels={,,},
ybar stacked,
ymin=0.006, ymax=0.4,
ylabel={Visits$/|\Scal|$},
xlabel={$\NN$ Query},
x label style={at={(axis description cs:0.5,1.25)},anchor=north},
ytick={0.001,0.01,0.1,1},
yticklabels={,$10^{-2}$,$10^{-1}$},
ylabel shift = 0.05cm,
ymode = log,
log ticks with fixed point,
log origin = infty,
]
\addplot[color=black,fill=black!10,]
coordinates {
	(9,0.1541)(4,0.2165)(5,0.1261)(1,0.1500)(7,0.2343)(13,0.0145)(16,0.0123)(8,0.0780)(2,0.1271)(3,0.2023)(6,0.1972)(10,0.0722)
	(15,0.0270)(14,0.0128)(11,0.0249)(12,0.1607)
};

\addplot[color=black,fill=black!70,]
coordinates {
	(9,0.0647)(4,0.0637)(5,0.0010)(1,0.0020)(7,0.0008)(13,0.0037)(16,0.0106)(8,0.0041)(2,0.0079)(3,0.0699)(6,0.0266)(10,0.0050)
	(15,0.0024)(14,0.0026)(11,0.0009)(12,0.0586)
};

\legend{include}
\addlegendentry{exclude}

\nextgroupplot[
width=4.4cm,
xmin=0.5, xmax=5.5,
xticklabels={,,},
xlabel={$\RNN$ Query},
x label style={at={(axis description cs:0.5,1.25)},anchor=north},
ybar stacked,
ymin=0.03, ymax=0.8,
ytick={0.01,0.1,1},
yticklabels={,$10^{-1}$,$10^{0}$},
ymode = log,
log ticks with fixed point,
log origin = infty,
]
\addplot[color=black,fill=black!10,]
coordinates {
	(3,0.0925)(5,0.0715)(4,0.0417)(1,0.1784)(2,0.4982)
};

\addplot[color=black,fill=black!70,]
coordinates {
	(3,0.0066)(5,0.0159)(4,0.0063)(1,0.0013)(2,0.0962)
};

\nextgroupplot[
xticklabels={Vessel-M, Pigeon, Seabird, Bus, Cats, Buffalo, Vessel-Y, Gulls, Truck, Bats, Hurdat2, Pen,Football, Geolife, Basketball, Taxi},
xticklabel style={rotate = 60, anchor=north east,inner sep=0pt,},
ybar stacked,
ymin=0.005, ymax=2,
ytick={0.001,0.01,0.1,1,10},
yticklabels={,$10^{-2}$,$10^{-1}$,$10^{0}$},
ylabel={Dist. Calls},
ymode = log,
log ticks with fixed point,
log origin = infty,
]
\addplot[color=black,fill=black!10,]
coordinates {
	(1,0.053)(2,0.053)(3,0.01)(4,0.032)(5,0.01)(6,0.011)(7,0.973)(8,0.133)(9,0.011)(10,0.01)(11,0.002)(12,0.006)(13,0.002)(14,0.343)(15,0.989)(16,0.016)
};

\addplot[color=black,fill=black!70,]
coordinates {
	(1,0.100)(2,0.301)(3,0.193)(4,0.194)(5,0.01)(6,0.040)(7,0.253)(8,0.282)(9,0.176)(10,0.002)(11,0.021)(12,0.306)(13,0.029)(14,0.270)(15,0.540)(16,0.066)
};

\nextgroupplot[
width=4.4cm,
xmin=0.5, xmax=5.5,
xticklabels={Hurdat2, Pen, Football, Geolife, Taxi},
xticklabel style={rotate = 60, anchor=north east,inner sep=0pt,},
ybar stacked,
ymin=1, ymax=700,
ytick={1,10,100,1000},
yticklabels={,$10^{1}$,$10^{2}$},
ymode = log,
log ticks with fixed point,
log origin = infty,
]
\addplot[color=black,fill=black!10,]
coordinates {
	(1,1.9)(2,99.9)(3,30.4)(4,26.5)(5,110.0)
};

\addplot[color=black,fill=black!70,]
coordinates {
	(1,7.9)(2,487.1)(3,74.4)(4,67.7)(5,487.0)
};

\end{groupplot}
\end{tikzpicture}
\vspace{-.4cm}
\caption{Effectiveness of including/excluding proposed bound enhancements (Section~\ref{sec:Bounds}), on real data set Relaxed $\CCT$s, averaged over 1000 queries (c.f. Section~\ref{sssec:supplementary_results}). 
	{\normalfont
The left side shows exact $\NN$ queries, and the right side shows exact $\RNN$ queries chosen to return exactly 100 results.
The top row shows average number of tree node visits (normalized to a factor of $|\Scal|$). The bottom row shows the \emph{absolute} $\delta_{F}$ and $\delta_{F\!D}$ calls for $\NN$ and $\RNN$ queries, respectively.
}
}
\label{fig:bound_effectiveness}
\end{figure}
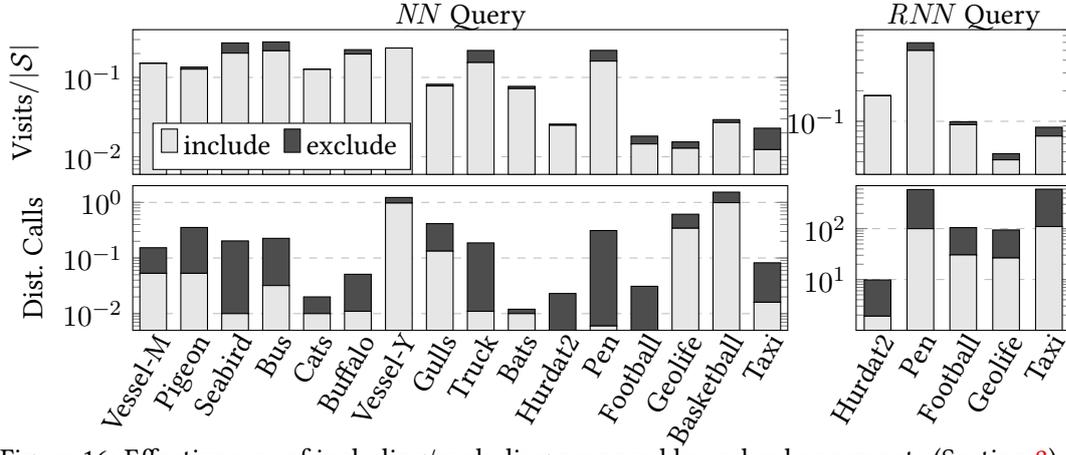

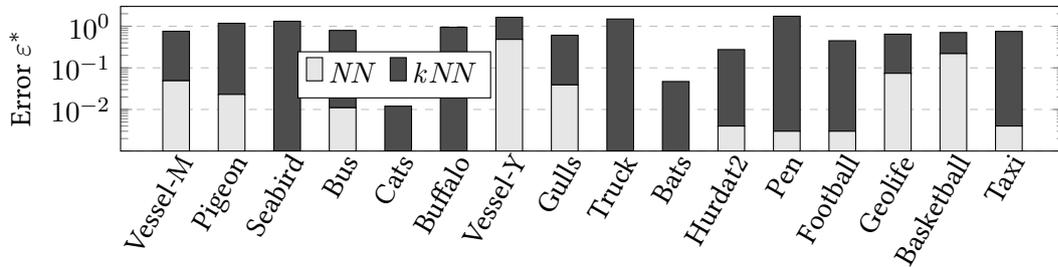
\begin{figure}[p]
\begin{tikzpicture}

\begin{groupplot}[
     group style = {group size = 1 by 1,
                    horizontal sep=0.15cm,
                    vertical sep=0.15cm,},
     width=14.0cm,
     height=3.5cm,
     grid style=dashed,
     xmode = normal,
     xtick=data,
     xtick style={draw=none},
     xticklabels={,,},
     xmin=0, xmax=17,
     ymajorgrids=true,
    ]

\nextgroupplot[
legend style={at={(0.29,0.69)},anchor=north},
legend style={/tikz/every even column/.append style={column sep=0.1cm}},
legend columns = 2,
xticklabels={Vessel-M, Pigeon, Seabird, Bus, Cats, Buffalo, Vessel-Y, Gulls, Truck, Bats, Hurdat2, Pen,Football, Geolife, Basketball, Taxi},
xticklabel style={rotate = 60, anchor=north east,inner sep=0pt,},
ybar stacked,
ymin=0.001, ymax=3,
ytick={0.001,0.01,0.1,1,10},
yticklabels={,$10^{-2}$,$10^{-1}$,$10^{0}$},
ylabel={Error $\Erel$},
ymode = log,
log ticks with fixed point,
log origin = infty,
]
\addplot[color=black,fill=black!10,]
coordinates {
	(1,0.049)(2,0.023)(3,0.001)(4,0.011)(5,0.001)(6,0.001)(7,0.484)(8,0.039)(9,0.001)(10,0.001)(11,0.004)(12,0.003)(13,0.003)(14,0.074)(15,0.220)(16,0.004)
};

\addplot[color=black,fill=black!70,]
coordinates {
	(1,0.708)(2,1.153)(3,1.318)(4,0.783)(5,0.011)(6,0.944)(7,1.160)(8,0.568)(9,1.492)(10,0.046)(11,0.271)(12,1.739)(13,0.446)(14,0.571)(15,0.487)(16,0.748)
};

\legend{$\NN$}
\addlegendentry{$\kNN$}

\end{groupplot}
\end{tikzpicture}
\vspace{-.4cm}
\caption{Implicit approximate query multiplicative errors on real data sets (c.f. Section~\ref{sssec:supplementary_results}). 
	{\normalfont
Bar chart shows average worst-case $\Erel$ values over $1000$ queries for $\NN$ (light shade) and $\kNN$ $k=5$ (dark shade) queries, on Frugal $\CCT$s.
}
}
\label{fig:implicit_error}
\end{figure}

\begin{figure}[p]
\begin{tikzpicture}

\begin{groupplot}[
     group style = {group size = 1 by 2,
                    horizontal sep=0.15cm,
                    vertical sep=0.15cm,},
     width=14.0cm,
     height=3.5cm,
     grid style=dashed,
     xmode = normal,
     xtick=data,
     xtick style={draw=none},
     xticklabels={,,},
     xmin=0, xmax=17,
     ymajorgrids=true,
    ]

\nextgroupplot[
xticklabels={,,},
ybar stacked,
ylabel={Nof. $\delta_F$},
ymin=0.001, ymax=4,
ytick={0.001,0.01,0.1,1,10},
yticklabels={,$10^{-2}$,$10^{-1}$,$10^{0}$},
ymode = log,
log ticks with fixed point,
log origin = infty,
]
\addplot[color=black,fill=black!10,]
coordinates {
	(9,0.011)(4,0.032)(5,0.0001)(1,0.053)(7,0.973)(13,0.002)(16,0.016)(8,0.133)(2,0.053)(3,0.0001)(6,0.011)(10,0.0001)
	(15,0.989)(14,0.343)(11,0.002)(12,0.006)
};

\addplot[color=black,fill=black!70,]
coordinates {
	(9,0.0001)(4,0.008)(5,0.0001)(1,0.213)(7,1.057)(13,0.004)(16,0.004)(8,0.031)(2,0.037)(3,0.0001)(6,0.0001)(10,0.0001)
	(15,0.605)(14,0.215)(11,0.015)(12,0.004)
};

\nextgroupplot[
legend style={/tikz/every even column/.append style={column sep=0.1cm}},
legend columns = 2,
legend pos=north west,
xticklabels={Vessel-M, Pigeon, Seabird, Bus, Cats, Buffalo, Vessel-Y, Gulls, Truck, Bats, Hurdat2, Pen,Football, Geolife, Basketball, Taxi},
xticklabel style={rotate = 60, anchor=north east,inner sep=0pt,},
ybar stacked,
ymin=0.001, ymax=20,
ytick={0.001,0.01,0.1,1,10,100},
yticklabels={,$10^{-2}$,$10^{-1}$,$10^{0}$,$10^{1}$},
ylabel={Nof. $\delta_{F\!D}$},
ymode = log,
log ticks with fixed point,
log origin = infty,
]
\addplot[color=black,fill=black!10,]
coordinates {
	(1,0.026)(2,0.023)(3,0.0001)(4,0.017)(5,0.0001)(6,0.004)(7,0.696)(8,0.059)(9,0.004)(10,0.0001)(11,0.0001)(12,0.003)(13,0.0001)(14,0.239)(15,0.762)(16,0.034)
};

\addplot[color=black,fill=black!70,]
coordinates {
	(1,0.394)(2,0.072)(3,0.0001)(4,0.021)(5,0.0001)(6,0.0001)(7,9.875)(8,0.093)(9,0.002)(10,0.0001)(11,0.020)(12,0.008)(13,0.005)(14,2.025)(15,6.826)(16,0.004)
};

\legend{$\NN$}
\addlegendentry{$\kNN$, $k=1$}

\end{groupplot}
\end{tikzpicture}
\vspace{-.4cm}
\caption{Effectiveness of exact $\NN$ (light shade) vs. $\kNN$ $k=1$ (dark shade) queries on real data sets (c.f. Section~\ref{sssec:supplementary_results}). 
	{\normalfont
Bar charts show average absolute values over $1000$ queries, on Relaxed $\CCT$s. The rows denote the number of distance $\delta_{F}$ (top) and Fr\'echet decision procedure $\delta_{F\!D}$ (bottom) computations during the Decide stage.
}
}
\label{fig:NN_vs_kNN}
\end{figure}
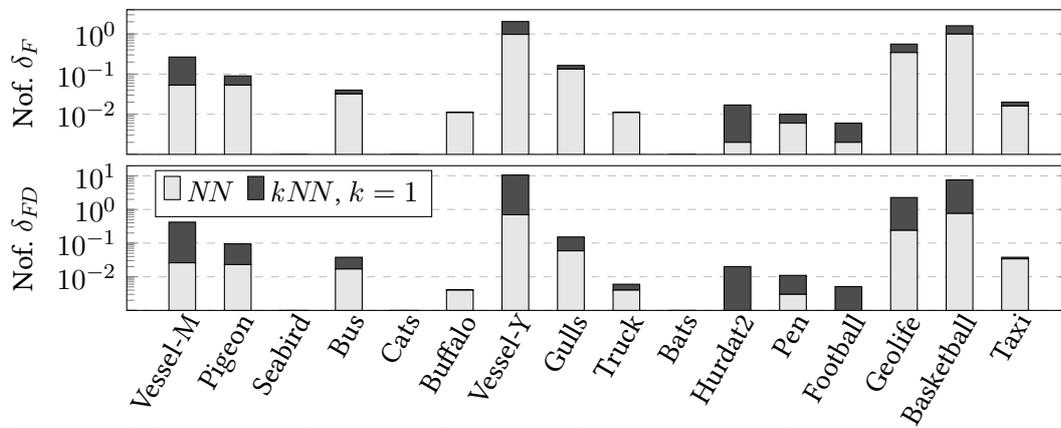

\section{Directions for Future Work} \label{sec:conclusions}
Our experiments show that even slightly larger cluster radii can negatively impact metric pruning efficiency.
We are therefore interested in other practical batch construction variants using Gonzalez' algorithm~\cite{gonzalez85}, or more recent techniques such as CLIQUE~\cite{agrawal2005}, SUBCLU~\cite{kailing2004}, genetic algorithm clustering~\cite{auffarth2010}, mutual information hierarchical clustering~\cite{kraskov2005}, or belief propagation clustering~\cite{frey2007}.

The proposed `Fix-Ancestor-Radius' primitive, which enables dynamic insertions, also allows to rectify radii that are affected from trajectory deletions in $\CCT$s.
We are interested in experiments on $\CCT$ quality and query performance in the fully-dynamic setting including identifying index sub-trees that benefit from a rebuild.
It is also worthwhile exploring changes required to implement $\CCT$ algorithms on multi-way trees such as the M-tree~\cite{ciaccia1997}, due to it's practical disk-based properties.
It may also be interesting to extend this work to other trajectory distance metrics such as the Hausdorff~\cite{alt2009}, discrete Fréchet~\cite{bri16}, and Wasserstein~\cite{vaserstein1969} distances, depending on application-specific requirements.

The $\kNN$ query algorithm analysis and experiment results show that the decide stage can perform $\mathcal{O}(|\Sstarcal|)$ Fréchet decision procedure computations.
Techniques, such as heuristic-guided pivot selection, may further reduce the number of $\delta_{\fd}$ calls.

Finally, our future work seeks to investigate changes required to support proximity searches on sub-trajectories~\cite{ber13}.
Algorithm modifications would need to balance cluster tree construction time, space consumption, and query time.

\clearpage
\bibliographystyle{abbrv}
\balance
\bibliography{references} 

\appendix
\clearpage

\section{Construction and Query Runtime} \label{sec:const_query_runtime}

The main focus of this work was to measure the number of distance computations and query I/O, per~\cite{hetland-09} which underscores that reducing these two measures (especially the first) should dominate algorithm design and experimentation analysis. However, it can also be useful to measure algorithm construction and query runtimes so that one can get a 'ballpark' estimate of how much time is spent. It can also be interesting to see which characteristics impact runtimes and what the trends are.

To this end, Figure~\ref{fig:query_latency} shows Relaxed $\CCT$ construction and exact query runtimes using synthetic data sets. An increase in cluster size, $n$, $|\Scal|$, and $d$ result in increased runtimes. This is expected since increases in these characteristics can result in more $\delta_{F}$ calls and node visits, and increases in $n$ can lead to longer runtimes when computing $\delta_{F}$ and linear bounds.

It is noteworthy that for a given algorithm time complexity, experiment runtimes can vary depending on the underlying hardware and use of software engineering techniques. Indeed, factor speedups can be achieved using approaches such as reducing memory consumption and access, parallelization, caching, using inline functions, multi-threading, or avoiding square root operations. Furthermore, in our setting runtimes are dependent on the choice of distance measure and its implementation details. For example, in this study we used a cubic complexity algorithm that computes $\delta_{F}$ exactly (other approaches such as a divide and conquer search can improve the $\delta_{F}$ time complexity at the expense of precision). For this work, runtimes were not part of core results and so we did not spend effort to improve this measure.

Our experiments were performed on a desktop computer with a $3.60$GHz Intel Core i7-7700 CPU, $32$GB RAM, running on a Matlab R2018b implementation over a Windows $10$ $64$-bit OS. If better runtimes are a paramount consideration, then a C++ implementation employing similar engineering techniques may significantly improve runtimes.

\begin{figure}[H] 
\vspace{-0.05cm}

\begin{tikzpicture}

\begin{groupplot}[
     group style = {group size = 5 by 2,
                    horizontal sep=0.15cm,
                    vertical sep=0.15cm,},
     height=3.4cm,
     width=4.3cm,
     grid style=dashed,
     xmode = normal,
     xtick=data,
     xtick style={draw=none},
     xticklabels={,,},
     xmin=0.5, xmax=5.5,
     ybar stacked,
     ymin=8, ymax=500,
     ytick={1,10,100,1000},
     yticklabels={,,},
     ymajorgrids=true,
     ymode = log,
     log ticks with fixed point,
     log origin = infty,
    ]
    

\nextgroupplot   [
ybar,
ylabel={Constr. [ms]},
yticklabels={,$10^{1}$,$10^{2}$},
]

\addplot[color=black,fill=black!90,]
coordinates {(1,12.7)(2,9.8)(3,14.3)(4,15.8)(5,16.3)};

\nextgroupplot [
ybar,
]

\addplot[color=black,fill=black!90,]
coordinates {(1,12.7)(2,22.3)(3,43.3)(4,69.3)(5,112.1)};

\nextgroupplot   [
ybar,
]

\addplot[color=black,fill=black!90,]
coordinates {(1,12.7)(2,13.5)(3,14.0)(4,14.1)(5,14.8)};

\nextgroupplot   [
ybar,
width=2.7cm,
xmin=0.5, xmax=2.5,
]

\addplot[color=black,fill=black!90,]
coordinates {(1,14.0)(2,15.8)};

\nextgroupplot  [
ybar,
]

\addplot[color=black,fill=black!90,]
coordinates {(1,12.7)(2,10.7)(3,10.8)(4,10.3)(5,12.2)};


\nextgroupplot   [
ybar stacked,
ylabel={Query [ms]},
yticklabels={,$10^{1}$,$10^{2}$},
xlabel={Cluster Size},
xticklabels={1,10,25,50,100},
]

\addplot[color=black,fill=black!10,]
coordinates {(1,18.8)(2,12.5)(3,16.3)(4,21.4)(5,28.6)};

\addplot[color=black,fill=black!70,]
coordinates {(1,9.4)(2,3.8)(3,10.1)(4,18.6)(5,29.6)};

\nextgroupplot [
xlabel={Trajectory Size $n$},
xticklabels={15,25,35,45,55},
]

\addplot[color=black,fill=black!10,]
coordinates {(1,12.5)(2,18.7)(3,25.3)(4,41.1)(5,67.5)};

\addplot[color=black,fill=black!70,]
coordinates {(1,3.8)(2,4.0)(3,9.9)(4,11.3)(5,15.7)};

\nextgroupplot   [
xlabel={$|\Scal|$~~[$10^3$]},
xticklabels={5,10,20,30,40},
]

\addplot[color=black,fill=black!10,]
coordinates {(1,12.5)(2,15.9)(3,18.0)(4,20.0)(5,22.1)};

\addplot[color=black,fill=black!70,]
coordinates {(1,3.8)(2,2.9)(3,4.4)(4,4.4)(5,4.5)};

\nextgroupplot   [
width=2.7cm,
xmin=0.5, xmax=2.5,
xlabel={$|\Scal|$~[$10^6$]},
xticklabels={1,10},
xlabel shift = -0.1cm,
]

\addplot[color=black,fill=black!10,]
coordinates {(1,83.2)(2,263.5)};

\addplot[color=black,fill=black!70,]
coordinates {(1,19.7)(2,46.5)};

\nextgroupplot  [
xlabel={Dimension $d$},
xticklabels={2,4,8,16,32},
legend style={/tikz/every even column/.append style={column sep=0.1cm}},
legend columns = 2,
legend pos=north west,
]

\addplot[color=black,fill=black!10,]
coordinates {(1,12.5)(2,12.1)(3,14.8)(4,20.8)(5,24.1)};

\addplot[color=black,fill=black!70,]
coordinates {(1,3.8)(2,2.6)(3,3.7)(4,5.4)(5,3.7)};

\legend{$\NN$}
\addlegendentry{$\kNN$}

\end{groupplot}
\end{tikzpicture}
\vspace{-.4cm}
\caption{Construction and query runtimes (milliseconds) on synthetic data set Relaxed $\CCT$s. The top shows average construction runtime per trajectory. The bottom shows query latency (end-to-end query runtime) of exact $\NN$ (light shade) and $\kNN$ $k=5$ (dark shade) queries, averaged over 1000 queries.
}
\label{fig:query_latency}
\end{figure}
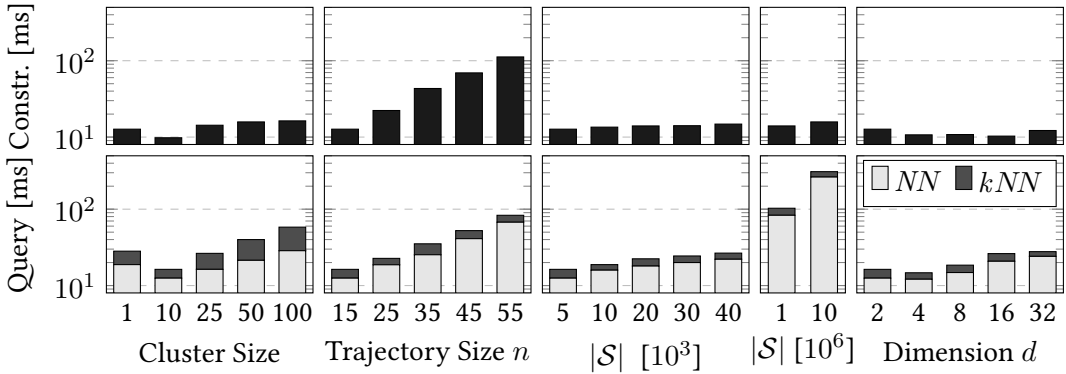

\end{document}